\documentclass[a4paper,reqno]{amsart}
\usepackage[T1]{fontenc}
\usepackage[utf8]{inputenc}
\usepackage[english]{babel}

\usepackage{amsmath,amsthm}
\usepackage{xfrac}
\usepackage{upgreek}

\usepackage{enumitem}
\setlist[enumerate]{label={\upshape\arabic*.},ref={\upshape\arabic*}}

\usepackage[dvipsnames]{xcolor}
\usepackage[colorlinks,linkcolor=BrickRed,urlcolor=Blue,citecolor=Blue,pdfusetitle]{hyperref}

\usepackage{algorithm}
\usepackage[noend]{algpseudocode}
\algnewcommand\algorithmicinput{\textbf{Input:}}
\algnewcommand\Input{\item[\algorithmicinput]}
\algnewcommand\algorithmicoutput{\textbf{Output:}}
\algnewcommand\Output{\item[\algorithmicoutput]}

\newcommand{\ZZ}{\mathbb{Z}}
\newcommand{\NN}{\mathbb{N}}
\newcommand{\FF}{\mathbb{F}}

\newcommand{\rdiv}{\mid_r}

\newcommand{\s}[2]{S_{#1}^{({#2})}}
\newcommand{\sq}[2]{s_{#1}^{({#2})}}
\newcommand{\st}[2]{\tilde S_{#1}^{({#2})}}

\DeclareMathOperator{\rk}{rk}

\DeclareMathOperator{\hw}{w_H}
\DeclareMathOperator{\hd}{d_H}
\DeclareMathOperator{\rkw}{w_{rk}}
\DeclareMathOperator{\rkd}{d_{rk}}

\newcommand{\lclm}[1]{\left[ #1 \right]_\ell}

\newtheorem{theorem}{Theorem}
\newtheorem{lemma}[theorem]{Lemma}
\newtheorem{prop}[theorem]{Proposition}
\newtheorem{corollary}[theorem]{Corollary}
\theoremstyle{definition}
\newtheorem{defn}[theorem]{Definition}
\newtheorem{example}[theorem]{Example}
\newtheorem{assumption}{Assumption}
\newtheorem{problem}{Problem}
\theoremstyle{remark}
\newtheorem{remark}[theorem]{Remark}

\begin{document}

\title{Decoding up to Hartmann--Tzeng and Roos bounds for rank codes}

\thanks{Research funded under grants PID2023-149565NB-I00 and PRE2020-093254, both financed by the Spanish Research Agency (MCIN/AEI / 10.13039/501100011033), the second one being also financed by the European Social Fund ``FSE invierte en tu futuro''.}

\author{José Manuel Muñoz}
\address{Dept. of Algebra, University of Granada, Spain}
\email{munoz@ugr.es}

\begin{abstract}
	A class of linear block codes which simultaneously generalizes Gabidulin codes and a class of skew cyclic codes is defined. For these codes, both a Hartmann--Tzeng-like bound and a Roos-like bound, with respect to their rank distance, are described, and corresponding nearest-neighbor decoding algorithms are presented. Additional necessary conditions so that decoding can be done up to the described bounds are studied. Subfield subcodes and interleaved codes from the considered class of codes are also described, since they allow an unbounded length for the codes, providing a decoding algorithm for them; additionally, both approaches are shown to yield equivalent codes with respect to the rank metric.
\end{abstract}


 % https://mathscinet.ams.org/mathscinet/msc/msc2020.html

\maketitle

\section{Introduction}

Cyclic codes are one of the main elementary classes of error-correcting linear codes. Cyclic codes can be described as ideals in the commutative ring $F[x]/(x^n - 1)$, where $F$ is a field and $n$ is the considered code length. This results in $F$-linear codes which, as a consequence of the well-studied properties of such ideals, inherit additional structure beyond the one for a general linear code over the field $F$. This additional structure results in a better understanding of the parameters of the codes and provides tools for designing error-correcting decoding algorithms which are more efficient than the ones available for general linear codes.
In particular, in \cite{BC60} and \cite{Hocquengheim59}, it is shown that a lower bound, known as the BCH bound, on the minimum Hamming distance of a cyclic code results from the existence of patterns in the set of powers of a field element which are roots of the generator polynomial of the code; that is, the defining set of the code. Notably, Reed--Solomon codes can be seen as designed in such a way that they have a BCH bound reaching the Singleton upper bound, being thus MDS codes. This bound was first generalized in \cite{HT72} into what we will refer to as the original Hartmann--Tzeng bound, which was itself generalized in \cite{Roos83} into the Roos bound. The procedure to choose the generator of the code in such a way that a lower bound on its minimum distance results in one of those bounds, therefore getting a code with a prescribed lower bound on their minimum distance, is straightforward, as it amounts to choosing a defining set with suitable structure. As shown in e.g. \cite{FT91}, nearest-neighbor decoding, with respect to the Hamming metric, can be performed up to the mentioned bounds. Another generalization of the bound in \cite{HT72}, which is neither a generalization nor a particular case of the one in \cite{Roos83}, was given in \cite{Roos82}. Despite also being described by Roos, we shall refer to this bound as \textit{the} Hartmann--Tzeng bound, as done in later works (mainly \cite{GLNN18}).

Skew cyclic codes, first presented in \cite{BGU07}, result from replacing the ring of commutative polynomials $F[x]$, used for cyclic codes, with the ring of skew polynomials $F[x;\sigma]$, defined by a variable $x$ and a field automorphism $\sigma : F \to F$ so that $xa = \sigma(a)x$ for each $a \in F$. These rings are a particular case of Ore polynomial rings, introduced in \cite{Ore33}. Then, skew cyclic codes, or $\sigma$-cyclic codes when the twist map $\sigma$ is specified, are left ideals in the ring $F[x;\sigma]/F[x;\sigma](x^n - 1)$, where the length $n$ has to be such that the left ideal $F[x;\sigma](x^n - 1)$ is twosided within the skew polynomial ring.
Skew cyclic codes share many properties with cyclic codes, and many analogous constructions from the ones for cyclic codes have been described for skew cyclic codes. This includes skew Reed--Solomon codes, which are MDS codes (see e.g. \cite{GLN17s,GLN17pgz}, where nearest-neighbor decoding algorithms with respect to the Hamming metric are also described) and BCH-like, Hartmann--Tzeng-like, and Roos-like bounds (see respectively \cite[Proposition 2]{CLU09}, \cite{GLNN18} and \cite{ALN21}), since there is an analogous concept of a defining set for these codes. Again, it is straightforward to design the defining set of a code so that an instance of these bounds applies to the code.

Many properties related to the Hamming metric in cyclic codes apply to the rank metric, with respect to the field extension $F/F^\sigma$ where $F^\sigma$ is the fixed field of $\sigma$, in $\sigma$-cyclic codes. The rank metric with respect to $F/F^\sigma$ is such that the vector $(a_1, \dots, a_n)$ has as its rank weight the $F^\sigma$-dimension of the $F^\sigma$-span of $a_1, \dots, a_n$. In fact, \cite[Proposition 1]{CLU09} gives a BCH-like bound with respect to the rank metric, while the Roos-like bound given in \cite{ALN21} applies to the rank metric (in addition to the Hamming metric, since the Hamming distance is at least equal to the rank distance) and \cite{ALNW22} shows that the Hartmann--Tzeng-like bounds in \cite{GLNN18} also apply with respect to the rank metric. As a result, skew Reed--Solomon codes are MRD codes.
One of the main applications where the rank metric is relevant is random linear network coding, described in \cite{SKK08}, where the proposed codes are not skew cyclic codes but Gabidulin codes. These codes, first introduced in \cite{Delsarte78} and then independently in \cite{Gabidulin85}, are not defined as an ideal in an algebra, but as the kernel of a matrix, see e.g. their definition in \cite{KG05} or \cite[Section VI.A]{SKK08}. These codes are also MRD codes.
It is well-known how to perform nearest-neighbour error-correction for Gabidulin codes, see e.g. \cite{SJB11}, where it is also shown that these codes can be interleaved with themselves in order to get codes whose length is beyond the order of the field automorphism, which is a desirable property in the context of random linear network coding.

In Section \ref{section:definitionsandbounds}, a family of linear codes which generalizes Gabidulin codes is defined. Despite defining these codes as the left kernel of a matrix (hence, they do not take advantage of some algebraic structure as cyclic and skew cyclic codes do), it is shown that Hartmann--Tzeng-like and Roos-like bounds apply to these codes, and therefore it is straightforward to design them in such a way that a lower bound for their minimum rank distance is guaranteed.
In Section \ref{section:sc}, where the relationship of these codes with skew cyclic codes is discussed, it is shown that these codes generalize the family of skew cyclic codes whose length matches the order of the field automorphism, and, furthermore, each one of these codes (in particular, each Gabidulin code) is in fact rank equivalent to some skew cyclic code.
Following analogous approaches for Gabidulin, skew cyclic and cyclic codes, a pair of syndrome-based nearest-neighbor error-correcting algorithms up to the previously described bounds with respect to the rank metric for our codes (and therefore, for skew cyclic codes whose length matches the order of the automorphism) is developed.
This is done in Section \ref{section:sfsrsynthesisdecoding} by applying some known facts on skew-feedback shift registers, briefly recalled in Section \ref{section:sfsr}, to the decoding problem described in Section \ref{section:errorvalueslocators}.
Depending on the exact parameters, these algorithms might not be able to correct up to the considered bounds. Some conditions which guarantee them to reach the bounds are given in Section \ref{section:ontheassumption}, where criteria for determining the error-correcting capacity of the algorithms are discussed.
The possible failures that these algorithms might raise when attempting to decode beyond their error-correcting capacity are discussed in Section \ref{section:toomanyerrors}.
A generalized version of an algorithm by Gabidulin, which can be used in one of the steps of the decoding algorithms, is provided in Section \ref{section:gabidulinsalgorithm}.
An example using the results in this work in order to decode is shown in Section \ref{section:example}.
Finally, Section \ref{section:length} studies the possibility of using the codes introduced in Section \ref{section:definitionsandbounds} as a tool towards constructing codes whose length is beyond the order of the field automorphism through both subfield subcodes and interleaved codes, as this is important in the context of random linear network coding.

\section{Defining sets and lower bounds}
\label{section:definitionsandbounds}

\newcommand{\sord}{{|\sigma|}}
\newcommand{\thord}{{|\theta|}}

\newcommand{\bigO}{O}

\newcommand{\hv}{{\normalfont \textbf h}}
\newcommand{\cv}{{\normalfont \textbf c}}
\newcommand{\ev}{{\normalfont \textbf e}}
\newcommand{\yv}{{\normalfont \textbf y}}
\newcommand{\vv}{{\normalfont \textbf v}}
\newcommand{\bv}{{\normalfont \textbf b}}
\newcommand{\av}{{[\alpha]_\sigma}}
\newcommand{\zerov}{{(0, \dots, 0)}}

\newcommand{\ep}{\varepsilon}
\newcommand{\epv}{{\boldsymbol \upvarepsilon}}
\newcommand{\etav}{{\boldsymbol \upeta}}

Throughout this paper, we consider a field $F$, a field automorphism $\sigma : F \to F$ such that its order as an automorphism, $\sord$, is some positive integer, and a code length $n \le \sord$. See Section \ref{section:length} for constructions where the length is not limited by $\sord$.
Although the rank metric for block codes usually considers finite fields, and consequently most applications of the result of this work can be expected to be considered in the context of finite fields, in which case $\sigma$ is a power of the Frobenius endomorphism in $F = \FF_{q^\sord}$ such that $F^\sigma = \FF_q$, the results in this paper do not require the fields to be finite, so we will not make such a requirement.
It follows from elementary Galois theory that $F/F^\sigma$ is a Galois extension of degree $\sord$.
These items are enough for introducing the main construction in our work:

\begin{defn}\label{defn:ccode}
	Given an automorphism $\sigma : F \to F$ of order $\sord \in \ZZ^+$, a vector $\hv = (h_1, \dots, h_n) \in F^n$ whose $1 \le n \le \sord$ entries are $F^\sigma$-linearly independent and a set $T \subseteq \ZZ$ the $F$-linear code $C_{(\sigma, \hv, T)}$ is the intersection of the left kernels of $\sigma^i(\hv)^T$ for each $i \in T$; that is, $C_{(\sigma, \hv, T)} = \{ \vv \in F^n ~|~ \vv \sigma^i(\hv)^T = 0 \text{ for all } i \in T \}$.
\end{defn}

The $T$ in $\sigma^i(\hv)^T$, or in general in $A^T$ for any vector or matrix $A$, denotes transposition and shall not be confused with the set $T$.

If $T$ is a finite set $\{ i_1 < i_2 < \dots < i_m \}$, then it can equivalently be defined as the left kernel of the matrix
	\begin{equation}\label{eq:H}
		H_{(\sigma, \hv, T)} =
		\begin{pmatrix}
			\sigma^{i_1}(h_1) & \sigma^{i_2}(h_1) & \dots & \sigma^{i_m}(h_1) \\
			\sigma^{i_1}(h_2) & \sigma^{i_2}(h_2) & \dots & \sigma^{i_m}(h_2) \\
			\vdots & \vdots &  & \vdots \\
			\sigma^{i_1}(h_n) & \sigma^{i_2}(h_n) & \dots & \sigma^{i_m}(h_n) \\
		\end{pmatrix}.
	\end{equation}
While the elements in $T$ can be taken without loss of generality as different elements in $\{ 0, \dots, \sord - 1 \}$ by replacing every element in $T$ with its remainder modulo $\sord$ (so $T$ is finite without loss of generality), for convenience we shall not assume that to be the case.

Let $k \ge 0$ be such that the elements in $T$ yield $n - k$ different remainders modulo $\sord$. Then, on account of \cite[Corollary 4.13]{LL88}, $H_{(\sigma, \hv, \{ 0, \dots, n - 1 \})}$ is invertible, so its columns are linearly independent. As a result, if the remainders of the elements in $T$ modulo $\sord$ are in $\{ 0, \dots, n-1 \}$, as it is the case for example for $n = \sord$, then $C_{(\sigma, \hv, T)}$ is an $[n, k]$ $F$-linear code. Furthermore, for any $i \in \ZZ$, $H_{(\sigma, \sigma^i(\hv), \{ 0, \dots, n - 1 \})} = H_{(\sigma, \hv, \{ i, \dots, i + n - 1 \})}$ is also invertible, so for the same reason if there are $\sord - n$ consecutive integers which are not in $T$ modulo $\sord$, $C_{(\sigma, \hv, T)}$ is an $[n, k]$ code. In general, the dimension of $C_{(\sigma, \hv, T)}$ is at least $k$ and at most $k + \sord - n$, depending on the number of linearly independent columns.

Definition \ref{defn:ccode} replicates the notation for $C_{(\varphi_u, \alpha, d)}$-codes as described in \cite[Definition 2]{GNS21}, which takes a more general map $\varphi_u : F \to F$ and is restricted to $T = \{ 0, \dots, d-2 \}$ for some $2 \le d \le n$ and $\hv = (\alpha, \varphi_u(\alpha), \dots, \varphi_u^{n-1}(\alpha))$, where $n$ is required to be the dimension of $F$ over $F^{\varphi_u}$ (which, for $\varphi_u = \sigma$, is $\sord$; see \cite[Section II]{GNS21} for the definition of $F^{\varphi_u}$).

\begin{remark}\label{remark:shortened}
	If $n < \sord$, then the $n$-length code $C_{(\sigma, \hv, T)}$ can be seen as a shortened code from the $\sord$-length code $C_{(\sigma, \hv^*, T)}$, where $\hv^*$ is any $F^\sigma$-basis of $F$ having $\hv$ as a subset. In general, if every entry in $\hv$ is in $\hv^*$, then $C_{(\sigma, \hv, T)}$ is a shortened code from $C_{(\sigma, \hv^*, T)}$.
\end{remark}

The rank distance function (or rank metric) and the rank weight function (always with respect to the extension $F/F^\sigma$) over elements in the $F$-vector space $F^n$ will be denoted by $\rkd$ and $\rkw$.
That is, $\rkd : F^n \times F^n \to \NN$ is the map $(\vv, \vv') \mapsto \rkw(\vv - \vv')$, where $\rkw : F^n \to \NN$ assigns to every $\vv \in F^n$ its rank weight, defined as the dimension of the $F^\sigma$-vector space spanned by its entries.
Analogously, the Hamming distance function and the Hamming weight function will be denoted by $\hd$ and $\hw$. We will also denote by $\rkd(\mathcal C)$ and $\hd(\mathcal C)$ the minimum nonzero rank and Hamming distance between elements in a code $\mathcal C \subseteq F^n$, which, if $\mathcal C$ is $F$-linear, match the minimum nonzero rank and Hamming weight of elements in $\mathcal C$ respectively.

The following equivalence criterion for linear rank codes, which is a coarser version from the result of extending into general fields the one implied by \cite[Proposition 1]{Morrison14}, will be enough for the purposes in this work.

\begin{defn}\label{defn:fsigmalinearrankequivalence}
	Two linear codes $\mathcal C_1, \mathcal C_2$ of length $n$ over a field $F$ are said to be \emph{$F^\sigma$-linearly rank equivalent} (with respect to the rank metric derived from the field extension $F/F^\sigma$) if there exists an invertible $n \times n$ matrix $P$ with entries in $F^\sigma$ such that $\mathcal C_2 = \mathcal C_1 P$.
\end{defn}

Since the map $F^n \to F^n$, $\vv \mapsto \vv P$ is an isometry for the rank metric, two $F^\sigma$-linearly rank equivalent codes have the same dimension and the same rank weight distribution.
For finer criteria for rank equivalences between linear codes, see the linear and semilinear rank equivalences, introduced in \cite{Berger03} for finite fields, also studied in \cite[Section III.B]{Morrison14}. The semilinear rank equivalence has been considered in the context of general fields in, for example, \cite[Definition 3.4]{ALNW22}.

\begin{lemma}\label{lemma:rankequivalence}
	All codes of length $n$ satisfying Definition \ref{defn:ccode} for the same $\sigma, T$ and such that their respective $\hv \in F^n$ span the same subspace of $F$ are $F^\sigma$-linearly rank equivalent.
\end{lemma}
\begin{proof}
	Let $\hv_1, \hv_2$ be any vectors of length $n \le \sord$ with $n$ $F^\sigma$-linearly independent entries that span the same $F^\sigma$-vector space $V \subseteq F$. Then, these vectors are $F^\sigma$-bases of $V$, so $\hv_2^T = P \hv_1^T$ for some invertible $n \times n$ matrix $P$ with entries in $F^\sigma$. Then, $\sigma^d(\hv_2^T) = P\sigma^d(\hv_1^T)$ for any $d \in \ZZ$ since $\sigma^d(P) = P$.
	As a consequence, $H_{(\sigma, \hv_2, T)} = P H_{(\sigma, \hv_1, T)}$, so their left kernels, which are $C_{(\sigma, \hv_2, T)}$ and $C_{(\sigma, \hv_1, T)} P^{-1}$ respectively, are equal.
\end{proof}

If $n = \sord$, $\hv$ will be an $F^\sigma$-basis of $F$. Hence, Lemma \ref{lemma:rankequivalence} and Remark \ref{remark:shortened} give the next result, which generalizes \cite[Proposition 1]{Berger03} and a subsequent remark, which give the analogous result from Gabidulin codes.

\begin{corollary}\label{cor:equivalencensord}
	All codes according to Definition \ref{defn:ccode} of some common length $n$ for some common $\sigma$ and some common $T$ are $(\sord - n)$-times shortened from $F^\sigma$-linearly rank equivalent codes.
	If $n = \sord$, they are $F^\sigma$-linearly rank equivalent.
\end{corollary}

The minimum distance for a code $\mathcal D$ is kept, if not increased, for any subcode $\mathcal C \subseteq \mathcal D$.
Note that, when we state that $\mathcal C$ is a subcode of $\mathcal D$, $\mathcal C$ is not required to be an $F$-vector space (that is, if may not be an $F$-linear code, potentially being any subset of $\mathcal D$).
Thus, any error-correcting algorithm that successfully corrects up to $\tau$ errors for $\mathcal D$ will also correct up to $\tau$ errors for any $\mathcal C \subseteq \mathcal D$. Consequently, from an error-correcting perspective, for any code $\mathcal C \subseteq F^n$ we may only focus on its $\hv$-defining set for some $\hv$, as we shall define now.

\newcommand{\T}[2]{T_{#1}^{#2}}

\begin{defn}\label{defn:ds}
	For a field automorphism $\sigma : F \to F$ and a given $\hv = (h_1, \dots, h_n) \in F^n$ with $F^\sigma$-linearly independent elements,
	the \emph{$\hv$-defining set (with respect to $\sigma$) of a code $\mathcal C \subseteq F^n$} is
	\begin{equation}\label{eq:ds}
		\T \hv \sigma(\mathcal C) = \{ i \in \ZZ ~|~ \mathcal \cv \sigma^i(\hv)^T = 0 \text{ for all } \cv \in \mathcal C \}.
	\end{equation}
\end{defn}

The defining set, as defined here, is closed under addition or subtraction of multiples of $\sord$, so $i \in \ZZ$ is in $\T \hv \sigma(\mathcal C)$ if and only if $i$ mod $\sord$ is, and $\T \hv \sigma(\mathcal C)$ is completely determined by its subset within $\{ 0, \dots, \sord - 1 \}$, or by any subset of representatives modulo $\sord$. It is immediate that $C_{(\sigma, \hv, T)}$ is the largest code in $F^n$ such that its $\hv$-defining set with respect to $\sigma$ contains $T$; in particular, any code $\mathcal C \subseteq F^n$ is a subcode of $C_{(\sigma, \hv, \T \hv \sigma(\mathcal C))}$.

Following \cite[Theorem 3.3]{GLNN18} and \cite[Theorem 22]{ALN21}, we will now prove that, with respect to the rank metric (and therefore the Hamming metric), a Hartmann--Tzeng-like bound and a Roos-like bound result from the existence of patterns in the $\hv$-defining set of a code, or equivalently, in the set $T$ of a $C_{(\sigma, \hv, T)}$-code.
For brevity, $(a, b)$ will denote the greatest common divisor of the integers $a$ and $b$.

\begin{theorem}[Hartmann--Tzeng bound]\label{thm:ht}
	Assume that, for a code $\mathcal C \subseteq F^n$ and some $\hv \in F^n$ with $n$ $F^\sigma$-linearly independent entries where $n \le \sord$,
	there exist some integers $\delta \ge 2, r \ge 0, b, t_1, t_2$ such that $(\sord, t_1) = 1$, $(\sord, t_2) < \delta$ and the set $T = b + t_1 \{ 0, 1, \dots, \delta - 2 \} + t_2\{ 0, 1, \dots, r \}$ is a subset of $\T \hv \sigma(\mathcal C)$. Then $\hd(\mathcal C) \ge \rkd(\mathcal C) \ge \delta + r$.
\end{theorem}

\begin{proof}
	Let $\cv = (c_1, \dots, c_n) \in \mathcal C_{(\sigma, \hv, T)}$ be such that $\rkw(\cv) = \nu$ for some $\nu \le w = \delta + r - 1$.
	This means that $\cv = \cv' M_\cv$ for some $M_\cv \in M_{w \times n}(F^\sigma)$ of rank $w$ and some $\cv' = (c_1', \dots, c_\nu', 0, \dots, 0) \in F^w$ with $\nu$ $F^\sigma$-linearly independent entries.
	For every $0 \le i \le \delta - 2$ and every $0 \le j \le r$,
	\begin{equation}\label{eq:htds}
	0 = \cv \sigma^{b + i t_1 + j t_2}(\hv)^T = \cv' M_\cv \sigma^{b + i t_1 + j t_2}(\hv)^T = \cv' \sigma^{b + i t_1 + j t_2}\left(M_\cv \hv^T\right),
	\end{equation}
	since the entries of $M_\cv$ are in the fixed field of $\sigma$. Consider $(a_1, \dots, a_w)^T = M_\cv \hv^T$. For each $1 \le l \le w$, the $l$-th row of $M_\cv$ has the coordinates of $a_l$ with respect to $\hv$, which is a basis for some $n$-dimensional $F^\sigma$-subspace of $F$. As a result, the $F^\sigma$-linear independence of the rows of $M_\cv$ implies the one of $a_1, \dots, a_w$.

	Equation \eqref{eq:htds} means that $\cv'$ is in the left kernel of the matrix
	\[ B = \sigma^b\left( \left[ \begin{array}{c|c|c|c|c}
A & \sigma^{t_2}(A) & \sigma^{2 t_2}(A) & \dots & \sigma^{r t_2}(A) \\
\end{array} \right]_{w \times (\delta - 1)(r + 1)} \right), \]
	
	where
	\[ A = \left( \begin{array}{cccc}
		a_1 & \sigma^{t_1}(a_1) & \dots & \sigma^{(\delta - 2)t_1}(a_1) \\
		a_2 & \sigma^{t_1}(a_2) & \dots & \sigma^{(\delta - 2)t_1}(a_2) \\
		\vdots & \vdots &  & \vdots \\
		a_w & \sigma^{t_1}(a_w) & \dots & \sigma^{(\delta - 2)t_1}(a_w) \\
	\end{array}\right)_{w \times (\delta - 1)}. \]

	Such left kernel is trivial as \cite[Lemma 3.2]{GLNN18}, which can be applied by noting that $w = (\delta - 1) + r$, yields that $\rk(B) = w$.

	Thus, $\cv'$ is a zero vector and so is $\cv = \cv' M_\cv$: there are no nonzero elements in $C_{(\sigma, \hv, T)} \supseteq \mathcal C$ of rank weight at most $w = \delta + r - 1$ and, by the linearity of $C_{(\sigma, \hv, T)}$, the minimum rank distance between elements of $C_{(\sigma, \hv, T)}$, and therefore $\mathcal C$, is at least $\delta + r$. The well-known inequality $\hd \ge \rkd$ completes the proof.
\end{proof}

\begin{theorem}[Roos bound]\label{thm:roos}
	Assume that, for a code $\mathcal C \subseteq F^n$ and some $\hv \in F^n$ with $n$ $F^\sigma$-linearly independent entries where $n \le \sord$,
	there exist some integers $\delta \ge 2, r \ge 0, b, t_1, t_2, k_0,\dots,k_r$ such that $(\sord, t_1) = 1 = (\sord, t_2)$, $k_0 < \dots < k_r$, $k_r - k_0 \le \delta + r - 2$ and the set $T = b + t_1 \{ 0, 1, \dots, \delta - 2 \} + t_2 \{ k_0,\dots,k_r \}$ is a subset of $\T \hv \sigma(\mathcal C)$. Then $\hd(\mathcal C) \ge \rkd(\mathcal C) \ge \delta + r$.
\end{theorem}

\begin{proof}
	As in the previous proof, any $\cv = (c_1, \dots, c_n) \in \mathcal C_{(\sigma, \hv, T)}$ such that $\rkw(\cv) = \nu$ for some $\nu \le w = \delta + r - 1$ is such that $\cv = \cv' M_\cv$ where $M_\cv$ and $\cv'$ are as in the proof above and
	\begin{equation}\label{eq:rds}
	0 = \cv \sigma^{b + i t_1 + k_j t_2}(\hv)^T = \cv' M_\cv \sigma^{b + i t_1 + k_j t_2}(\hv)^T = \cv' \sigma^{b + i t_1 + k_j t_2}\left(M_\cv \hv^T\right)
	\end{equation}
	for every $0 \le i \le \delta - 2$ and every $0 \le j \le r$.
	If we define $\theta = \sigma^{t_2}$ and take as $s$ any integer such that $s t_2 \equiv t_1 \mod \sord$ (which exists since $(\sord, t_2) = 1$), \eqref{eq:rds} means that $\cv'$ is in the left kernel of the matrix
	\[ B = \sigma^b\left( \left[ \begin{array}{c|c|c|c|c}
A & \theta^{s}(A) & \theta^{2s}(A) & \dots & \theta^{(\delta - 2) s}(A) \\
\end{array} \right]_{w \times (\delta - 1)(r + 1)} \right), \]
	
	where
	\[ A = \left( \begin{array}{cccc}
		\theta^{k_0}(a_1) & \theta^{k_1}(a_1) & \dots & \theta^{k_r}(a_1) \\
		\theta^{k_0}(a_2) & \theta^{k_1}(a_2) & \dots & \theta^{k_r}(a_2) \\
		\vdots & \vdots &  & \vdots \\
		\theta^{k_0}(a_w) & \theta^{k_1}(a_w) & \dots & \theta^{k_r}(a_w) \\
	\end{array}\right)_{w \times (r + 1)}. \]
	The left kernel is again trivial as $\sigma^{-b}(B)$, and therefore $B$, has rank $w$ by \cite[Lemma 12]{ALN21}. Since $(\sord, t_2) = 1$, $F^\sigma = F^\theta$, hence $a_1, \dots, a_w$ are $F^\theta$-linearly independent, as required for that lemma. In addition, $\thord = \sord$, so $(\thord, s) = (\thord, s t_2) = (\thord, t_1) = 1$.
	
	Hence, $\cv'$ is a zero vector and so is $\cv = \cv' M_\cv$. The asserted statements follow from this as at the end of the previous proof.
\end{proof}

\begin{remark}
	Theorem \ref{thm:roos} does not generalize Theorem \ref{thm:ht}, since Theorem \ref{thm:ht} is based on the bound in \cite{Roos82}, which generalizes the original Hartmann--Tzeng bound introduced in \cite{HT72}.
	For example, for $\sord = 22$ and a defining set \(
	T = \{
	0,   1,  2,\allowbreak
	4,   5,  6,\allowbreak
	8,   9, 10,\allowbreak
	12, 13, 14
	\}
	\), Theorem \ref{thm:ht} can yield a lower bound of $7$ through $\delta = 4$, $r = 3$, $b = 0$, $t_1 = 1$, $t_2 = 4$, while similar parameters are not possible through Theorem \ref{thm:roos}, which requires $(\sord, t_2) = 1$.
	Theorem \ref{thm:roos} would generalize Theorem \ref{thm:ht} if in addition $(\sord, t_2)$ is required to be $1$.
\end{remark}

\begin{remark}\label{remark:htrgenerality}
	In Theorem \ref{thm:roos}, $k_0$ can be taken as $0$ without loss of generality, as the set $T$ is kept the same by replacing $k_j$ with $k_j - k_0$ for each $0 \le j \le r$ and then substituting $b + t_2 k_0$ for $b$. Alternatively, for the same reason, $b$ can be taken as $0$.
	In both Theorem \ref{thm:ht} and Theorem \ref{thm:roos}, $t_1$ may be considered to be $1$ without loss of generality, by replacing $\sigma$ with a power thereof with the same order and the same fixed field. Alternatively, for the same reason, $t_2$ may be considered to be $1$ in Theorem \ref{thm:roos}.
\end{remark}

Generalized Gabidulin codes of designed rank distance $d$, as described in \cite[Section IV.A]{KG05}, can be seen as $C_{(\sigma, \hv, T)}$, as in Definition \ref{defn:ccode}, where $T = \{ 0, t_1, 2t_1, \dots, (d-2)t_1 \}$ for some $t_1$ coprime with $\sord$, and $\sigma$ is the minimum nonnegative power of the Frobenius endomorphism of a finite field $F$ which fixes some given subfield. Stated otherwise, if a subfield of $F$ with cardinality $q = p^r$, for $p$ the characteristic of $F$, is the subfield to be considered for the rank metric, then $\sigma$ is the endomorphism $a \mapsto a^q$, the $r$-th power of the Frobenius endomorphism.
	Alternatively, generalized Gabidulin codes may be described as $C_{(\sigma, \hv, T)}$ for $T = \{ 0, 1, \dots, d-2 \}$ and for $\sigma$ any automorphism in a finite field $F$, as every automorphism in a finite field is some power of the Frobenius endomorphism.
	For $t_1 = 1$ and $\sigma$ the minimum nonnegative power of the Frobenius endomorphism fixing a given subfield, the result is a Gabidulin code in the classic sense (see \cite[Section 4]{Gabidulin85} or \cite[Section III]{KG05}).
	Both Theorem \ref{thm:ht} and Theorem \ref{thm:roos} show that the minimum distance of these codes is $d$, as their dimension is $k = n - d + 1$ and the Singleton bound for block codes (any row of the $k \times n$ generator matrix of the code in standard form, which will itself be a codeword, has a Hamming, and therefore rank, weight of at most $n - k + 1$) prevents it from being greater than $n - (n - d + 1) + 1 = d$.
	Hence, these results generalize \cite[Theorem 1]{KG05}, which proves the result for $F$ finite, $b = r = 0$, $t_2 = 1$ and, in the case of Theorem \ref{thm:roos}, $k_0 = 0$. In particular, maximum rank distance (MRD) codes (with respect to the rank metric as considered here; see e.g. \cite[Section II.B]{GY06} or \cite[Section II.C]{SKK08} for their definition) over infinite fields can be described by taking an infinite $F$ and any $\sigma : F \to F$.
	Note that, in \cite{KG05}, the parity check matrix is considered as the transpose of the matrix described in Definition \ref{defn:ccode}.

While Theorem \ref{thm:ht} and Theorem \ref{thm:roos} can be used to prove a lower bound in an already constructed code through its defining set if it happens to fit the required properties, for practical purposes the most straightforward approach is to use Definition \ref{defn:ccode} and either Theorem \ref{thm:ht} or Theorem \ref{thm:roos} to construct a code $C_{(\sigma, \hv, T)}$ satisfying a minimum lower bound on its minimum rank (or, therefore, Hamming) distance through choosing a suitable set $T$ for the considered $\sord$.

The existence of a lower bound $\delta + r$ for the minimum distance of a code $\mathcal C$ satisfying Theorem \ref{thm:ht} or Theorem \ref{thm:roos} implies an error-correcting capacity of at least $\tau = \lfloor \frac {\delta + r - 1} 2 \rfloor$: for any codeword $\cv \in \mathcal C \subseteq F^n$ and any error vector $\ev \in F^n$ of rank (or Hamming) weight at most $\tau$, there is a single codeword $\cv' \in \mathcal C$ such that the rank (respectively, Hamming) distance between $\cv'$ and $\yv = \cv + \ev$ is at most $\tau$, which is $\cv' = \cv$. We shall construct a procedure for, under conditions on the code parameters studied in Section \ref{section:ontheassumption}, successfully computing the only, if any, codeword $\cv \in \mathcal C$ such that $\rkd(\cv, \yv) \le \tau$ for a given $\yv \in F^n$, which can be used to decode any codeword after it has been added an error of rank (or Hamming) weight at most $\tau$; that is, a nearest-neighbor decoding algorithm up to these lower bounds.

\section{Skew polynomials and skew cyclic codes}
\label{section:sc}

As noted in the previous section, Definition \ref{defn:ccode} generalizes Gabidulin codes (and an extension thereof described in \cite{KG05}) into a family of codes for which, depending on the set $T$, a Hartmann--Tzeng-like bound and a Roos-like bound may apply, replicating the work in \cite{GLNN18} and \cite{ALN21}. As we shall see now, these codes also extend the ones in \cite{GLNN18} and \cite{ALN21}, so the nearest-neighbor error-correcting algorithm described in the coming sections will also apply to those skew cyclic codes, which we shall briefly describe now.

A field $F$ and a field automorphism $\sigma : F \to F$ of order $n$ define the skew polynomial ring $R = F[x;\sigma]$, which is the set of polynomials written with its coefficients on the left of the variable $x$ with the usual sum and a product such that $xa = \sigma(a) x$ for any $a \in F$. 
Then, a \emph{$\sigma$-cyclic} (or, if $\sigma$ can be inferred from the context, \emph{skew cyclic}) code of length $n$ is any left ideal of $\mathcal R = R/R(x^n - 1)$. Since $R$ is a left and right Euclidean domain (see e.g. \cite{Jacobson96}), every left ideal and every right ideal in $R$ is principal. As a result, any such code $\mathcal C$ can be identified with the only monic right divisor $g \in R$ of $x^n - 1$ such that the left ideal $\mathcal Rg$ equals $\mathcal C$.
$\mathcal R$ inherits both the Hamming metric and the rank metric from $F^n$ through the canonical $F$-vector space isomorphism
\begin{equation}\label{eq:vsisomorphism}
	F^n \to \mathcal R; \qquad (a_0, \dots, a_{n-1}) \mapsto a_0 + a_1 x + \dots + a_{n-1} x^{n-1} + R(x^n - 1),
\end{equation}
so the Hamming and rank weight in $\mathcal R$ are defined as the corresponding weight of the coefficient vector for the only representative of degree less than $n$. Further details on skew cyclic codes can be found in e.g. \cite[Section 0]{BU09}, \cite[Section 2]{CLU09} or \cite[Section 2]{GLNN18}.

While skew cyclic codes are commonly defined in such a way that the order of $\sigma$ is only required to divide $n$, here we will only consider the case where its order $\sord$ equals $n$, since this is necessary in order to use the concept of defining sets.
This limitation is overcome by considering subfield subcodes: in \cite[Section 2]{GLNN18}, it is shown that $\sigma$-cyclic codes of length $n$ such that $n = s \sord$ for $s > 1$ can be considered a subfield subcode from a $\theta$-cyclic code, for some $\theta : L \to L$ in an extension $L$ of $F$, where $\thord = n$ whenever $F$ is a finite field, as well as when $\sord$ is coprime with $s$ and $F$ is a rational function field $\mathbb F(z)$ in the variable $z$ for some finite field $\mathbb F$.
The given construction results in the fixed fields of $\theta$ and $\sigma$ being equal, so these subfield subcodes inherit any lower bounds for the minimum rank distance. Furthermore, $\theta$ is chosen so that its restriction to $F$ is $\sigma$, allowing decoding to be performed over the larger field if necessary. Some practical considerations and examples on choosing $\theta$-cyclic codes of length $\thord$ so that they have a suitable subfield subcode over a given subfield are reviewed in \cite[Section 4]{GLNN18}. Analogous remarks apply to codes following Definition \ref{defn:ccode}, as later described in Section \ref{section:length}.

As shown in \cite{GLNN18}, skew cyclic codes can be described from a BCH-like point of view in the following sense.
Consider some $\alpha \in F$ such that $\{ \alpha, \sigma(\alpha), \dots, \sigma^{n-1}(\alpha) \}$ is an $F^\sigma$-basis for $F$, and then define $\beta = \alpha^{-1}\sigma(\alpha)$. In this context, for any $g \in R$ which right divides $x^n - 1$ (we shall denote this by $g \rdiv x^n - 1$), the $\beta$-defining set of $g$, $T_\beta(g)$, and the $\beta$-defining set of the corresponding code $\mathcal Rg$, $T_\beta(\mathcal Rg)$, are defined as
\begin{equation}\label{eq:betads}
T_\beta(g) = T_\beta(\mathcal Rg) = \{ i \in \ZZ ~|~ x - \sigma^i(\beta) \rdiv g \}.
\end{equation}
The definition of the $\beta$-defining set in \cite{GLNN18} is this one restricted to the range $0 \le i \le n - 1$, which is enough to characterize the $\beta$-defining set as defined in \eqref{eq:betads}. Since for any finite set of skew polynomials $f_1, \dots, f_k$ there exists a least common left multiple $\lclm{f_1, \dots, f_k}$, which is the only monic generator of the left ideal $R f_1 \cap \dots \cap R f_k$, skew cyclic codes may directly be constructed as matching some set $T$ by choosing $g$ as the least common left multiple of $x - \sigma^i(\beta)$ for $i \in T$, whose minimum rank (and therefore Hamming) distance can be conveniently chosen to satisfy a suitable lower bound by \cite[Theorem 3.3]{GLNN18} or \cite[Theorem 22]{ALN21}. The least common left multiple can be computed through the extended left Euclidean algorithm; see e.g. \cite[Ch. I, Theorem 4.33]{Bueso/alt:2003}.

We will make use of the following elementary result.
\begin{lemma}\label{lemma:lrdiv}
	A skew polynomial $\sum_k c_k x^k \in F[x;\sigma]$ is right divided by $x - a^{-1}\sigma(a)$ for some $0 \ne a \in F$ if and only if $\sum_k c_k \sigma^k(a) = 0$.
\end{lemma}
\begin{proof}
	This follows from \cite[Lemma 2.4]{LL88} and the fact that the $k$-partial norm of $a^{-1} \sigma(a)$ relative to $\sigma$ is $a^{-1}\sigma^k(a)$.
\end{proof}

Since $\beta = \alpha^{-1}\sigma(\alpha)$ and therefore $\sigma^i(\beta) = \sigma^i(\alpha)^{-1} \sigma(\sigma^i(\alpha))$, by Lemma \ref{lemma:lrdiv} $\sum_j c_j x^j \in F[x;\sigma]$ is right divided by $x - \sigma^i(\beta)$ if and only if $\sum_j c_j \sigma^{i+j}(\alpha) = 0$. This can be used to show that $\beta$-defining sets can be considered a particular case of $\hv$-defining sets. The skew cyclic code generated by $g$ is $\mathcal Rg$, which means that any codeword in $\mathcal Rg$ can be lifted into a skew polynomial $c = \sum_{i=0}^{n-1} c_k x^k \in F[x;\sigma]$ which is a left multiple of $g$, which is in turn a left multiple of $x - \sigma^i(\beta)$ for all $i \in T_\beta(g)$. This means that
\begin{align*}
	T_\beta(g)
	& = \{ i \in \ZZ ~|~ x - \sigma^i(\beta) \rdiv g \} \\
	& = \{ i \in \ZZ ~|~ x - \sigma^i(\beta) \rdiv c \text{ for all } c \in \mathcal Rg \} \\
	& = \{ i \in \ZZ ~|~ \sum_{j = 0}^{n-1} c_j \sigma^{i+j}(\alpha) = 0 \text{ for all } c \in \mathcal Rg \} \\
	& = \{ i \in \ZZ ~|~ \cv \sigma^i(\av)^T = 0 \text{ for all } \cv \in \mathcal C \}
\end{align*}
where $\av = \left( \alpha, \sigma(\alpha), \dots, \sigma^{n-1}(\alpha) \right)$ and $\mathcal C$ is $\mathcal R g$ seen as an $F$-vector space through the inverse map of \eqref{eq:vsisomorphism}. This yields the following statements.

\begin{prop}\label{prop:skewcyclicdefiningsets}
	Consider any field automorphism $\sigma : F \to F$ of order $n$ and any $\alpha \in F$ such that $\av = \left( \alpha, \sigma(\alpha), \dots, \sigma^{n-1}(\alpha) \right)$ is an $F^\sigma$-basis of $F$. Define $\beta = \alpha^{-1}\sigma(\alpha)$ and $\mathcal R = R/R(x^n - 1)$.
	\begin{enumerate}
		\item For any skew cyclic code $\mathcal C = \mathcal R g$, $T_\beta(g)$ equals $\T \av \sigma(\mathcal C)$ as in Definition \ref{defn:ds}.
	
		\item For any set $T \subset \ZZ$, the $\sigma$-cyclic code $\mathcal R g$ where $g = \lclm{x - \sigma^i(\beta) ~|~ i \in T}$ equals the image under the isomorphism \eqref{eq:vsisomorphism} of $C_{(\sigma, \av, T)}$ as in Definition \ref{defn:ccode}, or equivalently, it is the image of the left kernel of $H_{(\sigma, \av, T)}$.
	\end{enumerate}
\end{prop}

Consequently, \cite[Theorem 3.3]{GLNN18} is generalized by Theorem \ref{thm:ht}, which in addition shows that the bound applies for the rank metric, and \cite[Theorem 22]{ALN21} is generalized by Theorem \ref{thm:roos}.
In the latter, in contrast with Theorem \ref{thm:roos}, it is required that $t_2 = 1$; however, as noted in Remark \ref{remark:htrgenerality}, no generality is lost by doing so.

\begin{remark}\label{remark:equivalentorshortenedfromskewcyclic}
	By Corollary \ref{cor:equivalencensord}, all codes of length $n = \sord$ satisfying Definition \ref{defn:ccode} for some $\sigma$ and some $T$ are equivalent, with respect to the rank metric (with respect to $F/F^\sigma$), to any given $\sigma$-cyclic code with the same length and defining set.
	Notably, such $\sigma$-cyclic codes exist as, by \cite[Proposition 11]{GNS21}, there always exists an $\alpha \in F$ such that $\av$ is an $F^\sigma$-basis of $F$.
	Furthermore, codes of shorter length according to Definition \ref{defn:ccode} are shortened from $\sigma$-cyclic codes with the same defining set.
\end{remark}

\section{The multisequence skew-feedback shift-register synthesis problem}
\label{section:sfsr}

A key step for decoding up to the bounds given by Theorem \ref{thm:ht} and Theorem \ref{thm:roos} will be shown to be equivalent to solving a particular instance of the skew-feedback shift-register synthesis problem, which was first described in \cite{SJB11} by generalizing the linear-feedback shift-register synthesis problem.

\begin{defn}\label{defn:SFSR}
	Let $s_0, \dots, s_{N-1}$ be a sequence of length $N$ of elements in a field $F$ with a field automorphism $\theta : F \to F$. A vector $\vv = (v_0, v_1, \dots, v_\ell)$ with entries in $F$ (or a field extension thereof), where $v_0 \ne 0$, is a \emph{$\theta$-skew-feedback shift register} (or $\theta$-SFSR for short) of length $\ell$ for the sequence if
	\begin{equation} \label{eq:SFSR}
\sum_{i=0}^\ell v_i \theta^i\left(s_{n-i} \right) = 0 \qquad \text{for all } \ell \le n < N.
	\end{equation}
\end{defn}

	We are using $\theta$ instead of $\sigma$ for the field automorphism since it will not necessarily match $\sigma$ as considered throughout the rest of this work. If $\theta$ is chosen as the identity map in $F$, this is a classical linear-feedback shift register, or LFSR. By SFSR we may denote a $\theta$-SFSR for some unspecified $\theta$.

	There is a notion of equivalence between SFSRs. Any vector $\vv$ whose first entry $v_0$ is nonzero can be multiplied by any nonzero element in $F$ (or an extension thereof), including $v_0^{-1}$. The resulting vector is a $\theta$-SFSR for any given sequence if and only if so is the initial $\vv$. Hence, two $\theta$-SFSRs which are proportional (that is, they are equal once divided by their respective first entry) can be considered to be equivalent. Additionally, a $\theta$-SFSR of length $\ell$ can also be defined from its last $\ell$ components, excluding $v_0$ (which is assumed to be $1$ without loss of generality), as in \cite{SJB11}. While, with such a definition, this equivalence is replaced with an identity, we will adjust to our definition since the SFSRs we will find include their first entry, which will mostly not be normalized to $1$.

	The existence of an SFSR of length $\ell$ for a given sequence implies that any subset of $\ell$ consecutive items, including the first $\ell$ ones, determine the next terms in the sequence, since $v_0 \ne 0$ and therefore an expression for $s_n$ results from \eqref{eq:SFSR}. Hence, an SFSR can be understood as an abstract machine that yields the full sequence when receiving the first $\ell$ elements as input; see \cite{SJB11}. Note that multiplying $\vv$ by any nonzero element in $F$ (or an extension thereof), and therefore considering an equivalent SFSR, has no effect in the generated sequence. In the case that $v_\ell$ is also nonzero, which is not required by our definition, any subset of $\ell$ consecutive terms will also determine the previous elements in the sequence.
	
	We are mainly interested in the following version of the SFSR synthesis problem for multiple sequences.
	
	\begin{problem}\label{problem:SFSRsynthesis}
		Consider some $N,L \in \ZZ^+$, some field $F$ with an automorphism $\theta$ and, for each $j \in \{ 0, \dots, {L-1} \}$, a sequence $s^{(j)} = \sq 0 j, \sq 1 j, \dots, \sq {N-1} j$ of length $N$ with elements in $F$.
		Find the smallest $\ell \in \ZZ^+$ such that there exists a vector $\vv \in F^{\ell+1}$ which is a $\theta$-SFSR for each one of those sequences, as well as one of such vectors $\vv$.
	\end{problem}

	For our work, the problem as stated above will suffice. In \cite{SJB11}, a solution, requiring $\bigO(N^2)$ operations in the field $F$, for a generalization of this problem, where the sequences are allowed to have different lengths, is provided, and asymptotically faster algorithms also exist, see \cite{SB14}. We will make use of the fact that this problem is, thus, solvable within reasonable complexity with respect to the number of operations in the field.

	Allowing $\vv$ to belong to a field extension from $F$ in Problem \ref{problem:SFSRsynthesis} would not result in essentially different SFSRs, in the sense that any SFSR can be projected (assuming the axiom of choice if the field extension is not finite) into one whose entries are in $F$:
	
	\newcommand{\p}[2]{\pi_{#1}^{(#2)}}
	
	\begin{lemma}\label{lemma:SFSRsubfield}
		Let $\vv = (v_0, v_1, \dots, v_\ell)$ be a $\theta$-SFSR, where $\theta : F \to F$, with entries in any field extension $E$ of $F$, of length $\ell$ for $L$ sequences of length $N$ $s^{(0)}, \dots, s^{(L-1)}$ whose entries belong to a subfield $K$ of $F$ such that $\theta$ restricted to $K$ is a field automorphism $\tau : K \to K$.
		Consider any $K$-basis $\mathcal B = \{ b_i ~|~ i \in I\}$ of $E$, where $I$ is an arbitrary set of indices, such that $b_\iota = 1$ for some $\iota \in I$ and the projection to the $\iota$-th coordinate $\p {\mathcal B} \iota : E \to K$, $\p {\mathcal B} \iota \left(\sum_{i \in I} c_i b_i \right) = c_\iota$.
		Then, $\p {\mathcal B} \iota (v_0^{-1}\vv) = (\p {\mathcal B} \iota (1), \p {\mathcal B} \iota (v_0^{-1}v_1), \dots, \p {\mathcal B} \iota (v_0^{-1}v_\ell))$, whose entries are in $K$, is a $\tau$-SFSR for the sequences.
	\end{lemma}
	\begin{proof}
		The result of applying $\p {\mathcal B} \iota $ to both sides of each one of the identities
		$0 = \sum_{i=0}^\ell v_0^{-1}v_i \tau^i\left(\sq {n-i} j \right)$ for each $\ell \le n < N$ and each $0 \le j < L$ meets the definition of $\p {\mathcal B} \iota (v_0^{-1}\vv)$ being a $\tau$-SFSR for the sequences. Note that the first entry of $\p {\mathcal B} \iota (v_0^{-1}\vv)$ is nonzero, since it is $\p {\mathcal B} \iota (1) = 1$.
	\end{proof}
	
	Furthermore, the length $\ell$ and the uniqueness (up to a proportionality factor) of the solutions $\vv$ of Problem \ref{problem:SFSRsynthesis} do not depend on the considered field $F$ as long as it contains the elements of the sequences:

	\begin{prop}\label{prop:SFSRsubfield}
		Let $\vv = (v_0, v_1, \dots, v_\ell) \in K^{\ell + 1}$ be a $\tau$-SFSR for the $L$ sequences of length $N$ $s^{(0)}, \dots, s^{(L-1)}$ whose entries belong to $K$, for some automorphism $\tau : K \to K$. Let $F$ be a field extension of $K$.
		Then, $\vv$ is the only shortest $\tau$-SFSR with entries in $K$ for the given sequences if and only if $\vv$ is the only shortest $\tau$-SFSR with entries in $F$ for the sequences.
	\end{prop}

	By the only shortest $\tau$-SFSR with entries in a field $L$ for the sequences, we mean that the only elements in $L^{\ell + 1}$ which are a $\tau$-SFSR for the sequences are the nonzero $L$-multiples of $\vv$ and the sequence admits no shorter $\tau$-SFSRs with entries in $L$.

	\begin{proof}
		In order to prove the nontrivial implication, assume $\vv = (v_0, \dots, v_\ell)$, with $v_0 = 1$, is the only shortest $\tau$-SFSR with entries in $K$ for the sequences. Consider any $K$-basis $\mathcal B$ of $F$ with $1$ as its first entry and any $\vv' = (v_0', \dots, v_{\ell'}') \in F^{\ell' + 1}$ which is a $\tau$-SFSR of length $\ell' \le \ell$ for the sequences. By Lemma \ref{lemma:SFSRsubfield}, $\p {\mathcal B} 1 ((v_0')^{-1} \vv')$, whose first entry is $1$, is a $\tau$-SFSR in $K^{\ell' + 1}$ for the sequences. By the hypothesis on $\vv$, $\ell' = \ell$ and $\p {\mathcal B} 1 ((v_0')^{-1} \vv') = \vv$; that is, $\p {\mathcal B} 1 ((v_0')^{-1} v_i') = v_i$ for each $0 \le i \le \ell$ for every $K$-basis $\mathcal B$ of $F$ whose first entry is $1$.

		We now have to show that $(v_0')^{-1} \vv' = \vv$ and therefore $\vv'$ is an $F$-multiple of $\vv$.
		If $(v_0')^{-1} v_i'$ was not an element in $K$ for some $i$, then there would exist at least one basis of the form $\mathcal B = \{ 1, (v_0')^{-1} v_i, \dots \}$, which results in $v_i = \p {\mathcal B} 1 ((v_0')^{-1} v_i') = 0$, as well as another basis of the form $\mathcal B' = \{ 1, 1 + (v_0')^{-1} v_i', \dots \}$, and therefore $0 = v_i = \p {\mathcal B'} 1 ((v_0')^{-1} v_i') = 1$. This contradiction shows that $(v_0')^{-1} v_i' \in K$ and therefore $(v_0')^{-1} v_i' = \p {\mathcal B} 1 ((v_0')^{-1} v_i') = v_i$ for any basis $\mathcal B$ whose first entry is $1$.
	\end{proof}

\section{Error values and error locators}
\label{section:errorvalueslocators}

The usual approach for syndrome-based error-correcting decoding in the context of Gabidulin codes (see e.g. \cite[Sections V.A and VI.A3]{SKK08} and \cite[Section VII]{SJB11}) as well as skew codes and generalizations thereof (see \cite{GLN17s,GLN17pgz,GNS21}) involves, at least implicitly, the concepts of error values and error locators, as well as the equivalence between the error-correcting problem and the problem of finding those error values and locators.
We shall describe this in our context of some $\sigma : F \to F$, some $n \le \sord$ and some defining set $\T \hv \sigma(\mathcal C)$ for some code $\mathcal C \subseteq F^n$.
We will mostly follow \cite{SKK08}, although the parallelism among the cited works is apparent. We are considering a reasonably general scenario, where the fields may or may not be finite and the defining set can be any set; we are not yet requiring the defining set to satisfy Theorem \ref{thm:ht} or Theorem \ref{thm:roos}. Further generalizations (e.g. for error-erasure decoding or for list decoding) should also be immediate from the following description.

For any $d \in \ZZ$, we define the $d$-th syndrome (with respect to \( \hv \)) of any $\vv \in F^n$ as $\vv \sigma^d(\hv)^T$.
Consider any decomposition of some received $\yv \in F^n$ as $\yv = \cv + \ev$ where $\cv \in \mathcal C$ and $\ev = (e_0, \dots, e_{n-1}) \in F^n$.
If $d \in \T \hv \sigma(\mathcal C)$, then the $d$-th syndromes of $\yv$ (which we shall denote by $S_d$) and $\ev$ are equal, as the one of $\cv$ is $0$ by Definition \ref{defn:ds}:
\begin{equation}\label{eq:sdeh}
	S_d = \yv \sigma^d(\hv)^T = \ev \sigma^d(\hv)^T \qquad \text{for all } d \in \T \hv \sigma(\mathcal C).
\end{equation}

This means that the $d$-th syndrome of $\ev$, if $d \in \T \hv \sigma(\mathcal C)$, can be computed from $\yv$.

Now, consider the $F^\sigma$-vector space spanned by $\{ e_0, \dots, e_{n-1}\}$. Let $\nu \le n$ be its $F^\sigma$-dimension, which is also $\rkw(\ev)$, and let $\epv = (\ep_1, \dots, \ep_\nu) \in F^\nu$ be any ordered basis for the space. Then, there exists a $\nu$-rank matrix $B \in M_{\nu \times n}(F^\sigma)$ such that
\begin{equation}\label{eq:epsilon}
	\ev = \epv B.
\end{equation}
Plugging this into \eqref{eq:sdeh},
$S_d = \ev \sigma^d(\hv)^T = \epv B \sigma^d(\hv)^T = \epv \sigma^d(\etav)^T$
for all $d \in \T \hv \sigma(\mathcal C)$, where
\begin{equation}\label{eq:eta}
	\etav^T = (\eta_1, \dots, \eta_\nu)^T = B \hv^T.
\end{equation}

\begin{defn}
	Following \cite{SKK08}, we will refer to the entries of $\epv$ and $\etav$ respectively as \emph{error values} and \emph{error locators} (associated to $\epv$ and $\etav$ respectively, and ultimately to the error vector $\ev$). These are not unique, since $\epv$ is any basis of a space.
\end{defn}

We used the fact that the entries of $B$ are in $F^\sigma$ and therefore $B \sigma^d(\hv)^T = \sigma^d(B) \sigma^d(\hv)^T = \sigma^d(B\hv^T) = \sigma^d(\etav)^T$. From \eqref{eq:eta}, $B$ can be seen as the matrix whose $i$-th row contains the coordinates of the $i$-th entry of $\etav$, $\eta_i$, with respect to $\hv$, which is a $F^\sigma$-basis for some $n$-dimensional subspace of $F$ (equal to $F$ if and only if $n = \sord$). Since the $\nu$ rows of $B$ are linearly independent, the error locators are $F^\sigma$-linearly independent. 

Hence, the syndromes, error values and error locators satisfy the equations
\begin{equation} \label{eq:syndromeeq}
	S_d = \epv \sigma^d(\etav)^T = \sum_{k=1}^\nu \ep_k \sigma^d(\eta_k) \qquad \text{for all } d \in \T \hv \sigma(\mathcal C),
\end{equation}
or equivalently, applying $\sigma^{-d}$ to both sides,
\begin{equation} \label{eq:syndromeeq2}
	\sigma^{-d}\left(S_d\right) = \etav \sigma^{-d}(\epv)^T = \sum_{k=1}^\nu \eta_k \sigma^{-d}(\ep_k) \qquad \text{for all } d \in \T \hv \sigma(\mathcal C).
\end{equation}

\begin{remark}\label{remark:decompositionfromsolution}
	From each solution $\epv, \etav$ for \eqref{eq:syndromeeq} or the equivalent \eqref{eq:syndromeeq2} such that the error locators are in the $F^\sigma$-span of $\hv$, there is a single matrix $B \in M_{\nu \times n}(F^\sigma)$ satisfying \eqref{eq:eta}, and then \eqref{eq:epsilon} returns some $\ev$ such that, by \eqref{eq:sdeh}, the $d$-th syndromes of $\cv = \yv - \ev$ are zero for each $d \in \T \hv \sigma(\mathcal C)$ and therefore $\cv \in C_{(\sigma, \hv, \T \hv \sigma(\mathcal C))}$.
\end{remark}

The following observation, which leads to the possibility of modifying solutions into other, possibly shorter solutions, is straightforward from \eqref{eq:syndromeeq} (or \eqref{eq:syndromeeq2}).

\begin{lemma}\label{lemma:epetalc}
	Let $\epv, \etav \in F^\nu$ constitute a solution for \eqref{eq:syndromeeq} or the equivalent \eqref{eq:syndromeeq2} and let $M$ be any invertible matrix in $M_{\nu \times \nu}(F^\sigma)$.
	Then $\epv M, \etav (M^{-1})^T$ is another solution.
\end{lemma}

That is, if elementary column operations over $F^\sigma$ are performed on $\epv$ and the elemantary row operations whose matrices are the inverse of the previous column operations are applied to $\etav^T$ in the same order, we get another solution.
In particular, if $\epv, \etav$ was a solution for the above equations such that there existed some $F^\sigma$-linear dependence among the corresponding error values or the error locators, a shorter solution can be constructed.
For example, if $\ep_1 = \ep_2 - \ep_3$, then $\sum_{k=1}^3 \ep_k \sigma^d(\eta_k) = (\ep_2 - \ep_3) \sigma^d(\eta_1) + \ep_2 \sigma^d(\eta_2) + \ep_3 \sigma^d(\eta_3) = \ep_2 \sigma^d(\eta_1 + \eta_2) + \ep_3 \sigma^d(-\eta_1 + \eta_3)$, so $\ep_1$ can be removed from $\epv$ and the segment $(\eta_1, \eta_2, \eta_3)$ of $\etav$ can be replaced by $(\eta_1 + \eta_2, -\eta_1 + \eta_3)$.

Starting from any solution such that the error locators are in the $F^\sigma$-span of $\hv$, this removal of $F^\sigma$-linear dependent entries can be done until we arrive at what we are defining as \textit{relevant solution}.

\begin{defn}\label{defn:relevantsolution}
	In the context of Equation \eqref{eq:syndromeeq} or the equivalent Equation \eqref{eq:syndromeeq2}, we define a \emph{relevant solution of length $\nu \ge 0$} as a pair $\epv, \etav$ solving those equations such that $\epv = (\ep_1, \dots, \ep_\nu)$ and $\etav = (\eta_1, \dots, \eta_\nu)$ are $F^\sigma$-linearly independent and the error locators are in the $F^\sigma$-span of $\hv$.
\end{defn}

The following is immediate from Definition \ref{defn:relevantsolution} and Lemma \ref{lemma:epetalc}.

\begin{lemma}\label{lemma:minrelevant}
	The shortest $\nu$ such that there are solutions $\epv, \etav$ for \eqref{eq:syndromeeq} or \eqref{eq:syndromeeq2} such that $\eta_1, \dots, \eta_\nu$ are in the $F^\sigma$-span of $\hv$ only admits relevant solutions.
\end{lemma}

For each relevant solution $\epv, \etav$, when applying Remark \ref{remark:decompositionfromsolution}, the rank of the matrix $B$ is $\nu$, and so is the rank weight of $\ev$. This results in the following equivalence:

\begin{prop}\label{prop:solutionequivalence}
	For every $\nu \in \NN$ and every $\yv \in F^n$, finding a decomposition $\yv = \cv + \ev$ where $\cv \in C_{(\sigma, \hv, \T \hv \sigma(\mathcal C))}$ (equivalently, $\cv \sigma^d(\hv)^T = 0$ for all $d \in \T \hv \sigma(\mathcal C)$) and $\rkw(\ev) = \nu$ is equivalent to finding a relevant solution $\epv, \etav$ of length $\nu$.
\end{prop}

Consequently, the nearest-neighbor error-correcting problem (with respect to the code $C_{(\sigma, \hv, \T \hv \sigma(\mathcal C))} \supseteq \mathcal C$) and the problem of finding the only decomposition, if any, $\yv = \cv + \ev$ such that the weight of $\ev$ is not greater than a given known error-correcting capacity $\tau$ can analogously be stated as finding a solution $\epv, \etav$, where the error locators are in the $F^\sigma$-span of $\hv$, of, respectively, minimal length $\nu$ or the minimal if any length $\nu \le \tau$. By Lemma \ref{lemma:minrelevant}, the length minimality requirement implies that the solution will be relevant.

Furthermore, each one of \eqref{eq:syndromeeq} and the equivalent \eqref{eq:syndromeeq2} is a linear equation system with the entries of $\epv$ and, respectively, $\etav$ as the unknowns. This suggests the possibility of getting $\epv$ from $\etav$ or the other way around, once that one of them which is part of a relevant solution is known, by solving the corresponding system.

\begin{lemma}\label{lemma:epsilonoreta}
	Let $\epv, \etav \in F^\nu$ constitute a solution for \eqref{eq:syndromeeq} and the equivalent \eqref{eq:syndromeeq2}.
	If $\rkw(\etav) = \nu$ and the code $C_{(\sigma, \etav, \T \hv \sigma(\mathcal C))}$ has dimension zero, then the only solution for $\epv$ in \eqref{eq:syndromeeq} when choosing $\etav$ as in the solution above is $\epv$.
	If $\rkw(\epv) = \nu$ and the code $C_{(\sigma^{-1}, \epv, \T \hv \sigma(\mathcal C))}$ has dimension zero, then the only solution for $\etav$ in \eqref{eq:syndromeeq2} when choosing $\epv$ as in the solution above is $\etav$.
\end{lemma}
\begin{proof}
	By hypothesis, there is at least one solution.
	If $\epv$ and $\epv'$ are a solution for some given $\etav$, then by \eqref{eq:syndromeeq} the result of subtracting each element in $\epv'$ to the corresponding element in $\epv$ is a codeword in $C_{(\sigma, \etav, \T \hv \sigma(\mathcal C))}$. If this code only has the zero codeword, then necessarily $\epv = \epv'$. The same argument applies to $\etav$ and $\etav'$ being solutions for some given $\epv$, their difference element by element being a codeword in $C_{(\sigma^{-1}, \epv, \T \hv \sigma(\mathcal C))}$ by \eqref{eq:syndromeeq2}.
\end{proof}

If there is a known bound as in Theorem \ref{thm:ht} or Theorem \ref{thm:roos} with respect to $\T \hv \sigma(\mathcal C)$ greater than $\nu$, which is to be expected if we want to correct an error of rank weight $\nu$, then the codes considered in Lemma \ref{lemma:epsilonoreta} are guaranteed to only have the zero codeword.
This gives the following result.

\begin{prop}\label{prop:epsilonoreta}
	Let $\epv, \etav \in F^\nu$ constitute a solution for \eqref{eq:syndromeeq} and the equivalent \eqref{eq:syndromeeq2}, and assume that either Theorem \ref{thm:ht} or Theorem \ref{thm:roos} directly gives a lower bound greater than $\nu$ for $\mathcal C$ (or equivalently, for any code with defining set $\T \hv \sigma(\mathcal C)$).
	If $\rkw(\etav) = \nu$, then the only solution for $\epv$ in \eqref{eq:syndromeeq} when choosing $\etav$ as in the solution above is $\epv$.
	If $\rkw(\epv) = \nu$, then the only solution for $\etav$ in \eqref{eq:syndromeeq2} when choosing $\epv$ as in the solution above is $\etav$.
	In particular, if the solution $\epv, \etav$ is a relevant solution, then both of the above consequences are true.
\end{prop}

\section{The implicit SFSR synthesis problems in the error-correcting problem}
\label{section:sfsrsynthesisdecoding}

From Proposition \ref{prop:solutionequivalence} and Proposition \ref{prop:epsilonoreta}, the remaining step in order to decode is to find either some $\epv$ or some $\etav$ belonging to a solution of \eqref{eq:syndromeeq} or the equivalent \eqref{eq:syndromeeq2}. Since our goal is nearest-neighbor decoding, we may take advantage of the emergence of a pattern in the syndromes $S_d$ as a consequence of the structure of the defining set $\T \hv \sigma(\mathcal C)$ when the rank weight for the error $\ev = \yv - \cv$ is small enough for some codeword $\cv$. Hence, for this section, we will assume that the defining set $\T \hv \sigma(\mathcal C)$ of the considered code $\mathcal C \subseteq F^n$ has a subset of the form
\begin{equation}\label{eq:T}
	T = b + t_1 \{ 0, 1, \dots, \delta - 2 \} + t_2\{ k_0, \dots, k_r \} \quad \text{where} \quad \delta \ge 2 \text{ and } (\sord, t_1) = 1
\end{equation}
for some integers $b, t_1, t_2, \delta, r, k_0 < \dots < k_r$.
Equivalently, $\mathcal C \subseteq C_{(\sigma, \hv, \T \hv \sigma(\mathcal C))} \subseteq C_{(\sigma, \hv, T)}$. We are still not requiring this set $T$ to have the exact structure required by either Theorem \ref{thm:ht} or Theorem \ref{thm:roos}.

\begin{remark}\label{remark:sigmat1}
	Since $(\sord, t_1) = 1$, the fixed fields of $\sigma$, $\sigma^{t_1}$, $\sigma^{-1}$ and $\sigma^{-t_1}$ are all the same; thus, the concepts of $F^\sigma$-, $F^{\sigma^{-1}}$-, $F^{\sigma^{t_1}}$- and $F^{\sigma^{-t_1}}$-linear independence (or dependence) over $F$ are all equivalent. The error values and the error locators are linearly independent over any of these four equal fixed fields, and so are the entries of $\sigma^k(\epv)$ and $\sigma^k(\etav)$ for any $k \in \ZZ$, as any $F^\sigma$-linear combination is kept after applying the $F^\sigma$-linear map $\sigma$.
\end{remark}

As noted in \cite[Section VI.D]{SKK08}, which works in the context of Gabidulin codes, there are two dual approaches to rank-metric decoding, depending on how the error values and the error locators are used, as their roles in \eqref{eq:syndromeeq2} are reversed with respect to \eqref{eq:syndromeeq}.
As a result of the pattern in $T$, $(\delta-1)(r+1)$ identities are given by either \eqref{eq:syndromeeq} or \eqref{eq:syndromeeq2}, so the syndromes can be arranged into $r+1$ sequences of length $\delta - 1$, where the $i$-th element in the $j$-th sequence is
\newcommand{\forallij}{\text{ for all } 0 \le i \le \delta - 2,\,0 \le j \le r}
\begin{equation}\label{eq:s}
	\s i j = S_{b + t_1 i + t_2 k_j} \qquad \forallij
\end{equation}
or, in order to apply \eqref{eq:syndromeeq2} instead of \eqref{eq:syndromeeq},
\begin{equation}\label{eq:st}
	\st i j = \sigma^{-b - t_1 i - t_2 k_j}\left(S_{b + t_1 i + t_2 k_j}\right) \quad \forallij.
\end{equation}
Using this notation, which has been taken from \cite[Section VI]{FT91}, \eqref{eq:syndromeeq}, whose solutions $\epv, \etav$ lead to decoding as shown in the previous section, is written as
\begin{equation}\label{eq:see}
	\s i j = \sum_{k=1}^\nu \ep_k \sigma^{b + t_1 i + t_2 k_j}(\eta_k) \qquad \forallij,
\end{equation}
while the equivalent \eqref{eq:syndromeeq2} is
\begin{equation}\label{eq:stee}
	\st i j = \sum_{k=1}^\nu \eta_k \sigma^{- b - t_1 i - t_2 k_j}(\ep_k) \qquad \forallij.
\end{equation}
For the analogous case in non-skew codes, it was shown in \cite[Section VI]{FT91} that the corresponding sequences $\s i j$ satisfy a recurrence given by an LFSR related to the error values.
We will show that, for any $\epv, \etav$ which are a solution for the equations above, the sequences given by $\s i j$ satisfy a recurrence given by a $\sigma^{t_1}$-SFSR related to $\epv$, and the sequences given by $\st i j$ satisfy a recurrence from a $\sigma^{-t_1}$-SFSR related to $\etav$. The first recurrence is a general case from the approach in \cite[Section VI.A3]{SKK08} and \cite[Section VII]{SJB11} for Gabidulin codes, while the second one extends \cite[Section VI.D]{SKK08} as well as the work for skew codes in \cite{GLN17s,GLN17pgz}, being also analogous to the work in \cite{GNS21}.
The terms for these recurrences result from two relevant skew polynomials.

\newcommand{\sh}{\varsigma}
\newcommand{\ld}{\lambda}
\newcommand{\shv}{\boldsymbol \varsigma}
\newcommand{\ldv}{\boldsymbol \uplambda}
\begin{defn}[Error span polynomial]\label{defn:span}
	Consider the skew polynomial ring $R' = F[z;\sigma^{t_1}]$.
	The \emph{error span (skew) polynomial} $\sh = \sh_0 + \sh_1 z + \dots + \sh_\nu z^\nu \in R'$ associated to $\epv = (\ep_1, \dots, \ep_\nu)$ is $\lclm{z - \ep_k^{-1} \sigma^{t_1}(\ep_k) ~|~ 1 \le k \le \nu}$.

	Equivalently (see Lemma \ref{lemma:lrdiv}), $\sh$ is the only element in $R'$ of degree $\nu$ such that $\sh_\nu = 1$ and
	\begin{equation}\label{eq:span}
		\sum_{l=0}^\nu \sh_l \sigma^{t_1 l}(\ep_k) = 0 \qquad \text{for all } 1 \le k \le \nu.
	\end{equation}
	
	The \emph{error span vector} is $\shv = (\sh_0, \dots, \sh_\nu) \in F^{\nu + 1}$.
\end{defn}

\begin{defn}[Error locator polynomial]\label{defn:locator}
	Consider the skew polynomial ring $\bar R' = F[\bar z;\sigma^{-t_1}]$.
	The \emph{error locator (skew) polynomial} $\ld = \ld_0 + \ld_1 \bar z + \dots + \ld_\nu \bar z^\nu \in \bar R'$ associated to $\etav = (\eta_1, \dots, \eta_\nu)$ is $\lclm{\bar z - \eta_k^{-1}\sigma^{-t_1}(\eta_k) ~|~ 1 \le k \le \nu}$.

	Equivalently (see Lemma \ref{lemma:lrdiv}), $\ld$ is the only element in $\bar R'$ of degree $\nu$ such that $\ld_\nu = 1$ and
	\begin{equation}\label{eq:locator}
		\sum_{l=0}^\nu \ld_l \sigma^{-t_1 l}(\eta_k) = 0 \qquad \text{for all } 1 \le k \le \nu.
	\end{equation}
	
	The \emph{error locator vector} is $\ldv = (\ld_0, \dots, \ld_\nu) \in F^{\nu + 1}$.
\end{defn}

These polynomials can be defined as their analogous linearized polynomials in the context of rank codes over finite fields, see e.g. \cite{SKK08,SJB11}, while they are taken as commutative polynomials for cyclic codes, see e.g. \cite{FT91}.
Since the Hamming metric is the one considered in \cite{GLN17s, GLN17pgz, GNS21}, their error locator skew polynomial is the equivalent in an opposite ring (that is, in the ring with the same elements and the same sum operation where the product is the one where the arguments are swapped) from a left multiple from the one in Definition \ref{defn:locator}; in fact, in those works the analogous (in an opposite ring) for the error locator skew polynomial as defined here is computed first, being then completed into one suitable for decoding with respect to the Hamming metric.

While defining the error span and locator vectors $\shv, \ldv$ might be enough, and in fact only these vectors will be considered in the algorithms, working with skew polynomials (or linearized polynomials) allows us to readily apply many properties that would otherwise have to be explicitly described in case of working with vectors.

The following result shows that $\shv$ and $\ldv$ might be an SFSR of length $\nu$ for some sequences, as their first entry is nonzero, as required in Definition \ref{defn:SFSR}.

\begin{lemma}\label{lemma:degreeandconstantcoefficient}
	$\sh_\nu = \ld_\nu = 1$, and $\sh_0$ and $\ld_0$ are nonzero. That is, both the error span polynomial $\sh$ and the error locator polynomial $\ld$ have degree $\nu$ and a nonzero constant coefficient.
\end{lemma}
\begin{proof}
	The degree of both $\sh$ and $\ld$, which are monic, is $\nu$ by \cite[Theorem 5.3]{DL07} and the fact that both $\epv$ and $\etav$ are vectors with $F^\sigma$-linearly independent entries (recall Remark \ref{remark:sigmat1}).
	Both $\sh_0$ and $\ld_0$ are the product of nonzero elements in $F$ by \cite[Section 2]{DL07}.
\end{proof}

In general, the constant coefficients might not be equal to $1$.

\begin{remark}\label{remark:shld2epeta}
	By \cite[Corollary 5.4]{DL07}, the $\nu$-dimensional $F^\sigma$-span (note Remark \ref{remark:sigmat1}) of the error values is $\ker(\sh(\sigma^{t_1}))$, where $\sh(\sigma^{t_1}) : F \to F$ is the map sending any $\gamma \in F$ into $\sh_0 \gamma + \sh_1 \sigma^{t_1}(\gamma) + \sh_2 \sigma^{2 t_1}(\gamma) + \dots + \sh_\nu \sigma^{\nu t_1}(\gamma)$.
	Analogously, the $\nu$-dimensional $F^\sigma$-span of the error locators is $\ker(\ld(\sigma^{-t_1}))$.
	Thus, $\epv$ or $\etav$ can be computed from $\sh$ and $\ld$ respectively as an $F^\sigma$-basis of $\ker(\sh(\sigma^{t_1}))$ or, respectively, $\ker(\ld(\sigma^{-t_1}))$. If $F$ is a finite field, this is equivalent to finding an $F^\sigma$-basis for the roots of the corresponding linearized polynomial, which is the approach taken in \cite{SJB11}.

	Recall that, by Lemma \ref{lemma:epetalc}, any basis of the space will work as long as Equation \eqref{eq:syndromeeq} or the equivalent Equation \eqref{eq:syndromeeq2} holds. Thus, as described in Proposition \ref{prop:epsilonoreta}, one has to choose to compute either $\epv$ from the error span polynomial or $\etav$ from the error locator polynomial, and then the other one is computed through solving Equation \eqref{eq:syndromeeq} or Equation \eqref{eq:syndromeeq2}. That is, $\epv$ and $\etav$ should not be independently computed from $\sh$ and $\ld$. Hence, exactly one of those polynomials has to be found.
\end{remark}

One of the main properties of these skew polynomials that we must observe is that they satisfy the following \emph{key equations} with respect to the syndromes given by $T$:

\begin{prop}[Key equations with an error span polynomial]\label{prop:keyspan}
	Let $\epv, \etav$ be a solution of \eqref{eq:see} where $\s i j$ is as in \eqref{eq:s} and let $\sh_0 + \sh_1 z + \dots + \sh_\nu z^\nu$ be the error span polynomial associated to $\epv$ as in Definition \ref{defn:span}.
	For every $i \in \{ \nu, \dots, \delta-2 \}$ and for every $j \in \{ 0, \dots, r\}$,
	\begin{equation}\label{eq:keyspan}
		\sum_{l=0}^\nu \sh_l \sigma^{t_1 l}\left(\s {i-l} j\right) = 0.
	\end{equation}
	Stated otherwise, $\shv = (\sh_0, \dots, \sh_\nu)$ is a $\sigma^{t_1}$-skew-feedback shift register for the $r+1$ sequences of length $\delta - 1$ given by $\s i j$ (as noted in Lemma \ref{lemma:degreeandconstantcoefficient}, $\sh_0 \ne 0$).
\end{prop}
\begin{proof}
	By applying \eqref{eq:see} and \eqref{eq:span},
	\begin{align*}
		\sum_{l=0}^\nu \sh_l \sigma^{t_1 l}\left(\s {i-l} j\right)
		& = \sum_{l=0}^\nu \sh_l \sigma^{t_1 l}\left( \sum_{k=1}^\nu \ep_k \sigma^{b + (i-l) t_1 + k_j t_2}(\eta_k) \right) \\
		& = \sum_{k=1}^\nu \sigma^{b + i t_1 + k_j t_2}(\eta_k) \sum_{l=0}^\nu \sh_l \sigma^{t_1 l}( \ep_k ) \stepcounter{equation}\tag{\theequation}\label{eqline:eespan}  \\
		& = 0
	\end{align*}
	for each considered $i,j$.
\end{proof}

\begin{prop}[Key equations with an error locator polynomial]\label{prop:keylocator}
	Let $\epv, \etav$ be a solution of \eqref{eq:stee} where $\st i j$ is as in \eqref{eq:st} and let $\ld_0 + \ld_1 \bar z + \dots + \ld_\nu \bar z^\nu$ be the error locator polynomial associated to $\etav$ as in Definition \ref{defn:locator}.
	For every $i \in \{ \nu, \dots, \delta-2 \}$ and for every $j \in \{ 0, \dots, r\}$,
	\begin{equation}\label{eq:keylocator}
		\sum_{l=0}^\nu \ld_l \sigma^{-t_1 l}\left( \st {i-l} j \right) = 0.
	\end{equation}
	Stated otherwise, $\ldv = (\ld_0, \dots, \ld_\nu)$ is a $\sigma^{-t_1}$-skew-feedback shift register for the $r+1$ sequences of length $\delta - 1$ given by $\st i j$ (as noted in Lemma \ref{lemma:degreeandconstantcoefficient}, $\ld_0 \ne 0$).
\end{prop}

\begin{proof}
	By applying \eqref{eq:stee} and \eqref{eq:locator},
	\begin{align*}
		\sum_{l=0}^\nu \ld_l \sigma^{-t_1 l}\left( \st {i-l} j \right)
		& = \sum_{l=0}^\nu \ld_l \sigma^{-t_1 l}\left( \sum_{k=1}^\nu \eta_k \sigma^{-b -t_1(i-l) - t_2 k_j}(\ep_k) \right) \\
		& = \sum_{k=1}^\nu \sigma^{-b - i t_1 - k_j t_2}(\ep_k) \sum_{l=0}^\nu \ld_l \sigma^{-t_1 l}( \eta_k ) \stepcounter{equation}\tag{\theequation}\label{eqline:eelocator} \\
		& = 0
	\end{align*}
	for each considered $i,j$.
\end{proof}

\begin{remark}
	Equation \eqref{eq:keyspan} and Equation \eqref{eq:keylocator} can equivalently be stated using skew polynomials in $R'$ and $\bar R'$ respectively. Equation \eqref{eq:keyspan} is equivalent to the $i$-th coefficient of $\sh S^{(j)}$, for $\nu \le i \le \delta - 2$, being equal to zero, where $S^{(j)} = \s 0 j + \s 1 j z +  \dots + \s {\delta - 2} j z^{\delta - 2} \in R'$. Equation \eqref{eq:keylocator} is equivalent to the $i$-th coefficient of $\ld \tilde S^{(j)}$, for $\nu \le i \le \delta - 2$, being equal to zero, where $\tilde S^{(j)} = \st 0 j + \st 1 j \bar z +  \dots + \st {\delta - 2} j \bar z^{\delta - 2} \in \bar R'$. Equivalently, those products modulo $z^{\delta - 1}$ or, respectively, $\bar z^{\delta - 1}$, have degree less than $\nu$. Since, for a single sequence, it is immediate to recognize a congruence identity from this ($\sh S \equiv \omega \mod z^{\delta - 1}$ or $\ld \tilde S \equiv \omega \mod \bar z^{\delta - 1}$, where the degree of $\omega$ is less than $\nu$), this is how the key equations are stated in \cite{GLN17s}, as well as, using linearized polynomials, in \cite[Section VI.A3]{SKK08} and \cite[Theorem 4]{KG05}.
\end{remark}

If the error span polynomial $\sh$ associated to some $\epv$ from some decomposition $\yv = \cv + \ev$ where $\cv \in \mathcal C$ and $\rkw(\ev) = \nu$ is the shortest (up to nonzero $F$-multiples) $\sigma^{t_1}$-SFSR solving \eqref{eq:keyspan}, then $\sh$ can be retrieved by any algorithm solving Problem \ref{problem:SFSRsynthesis}. The same applies to the error locator polynomial $\ld$ associated to $\etav$: if it is the shortest $\sigma^{-t_1}$-SFSR solving \eqref{eq:keylocator}, then it can be found by solving the SFSR synthesis problem.
In order to have uniqueness, up to $F$-multiples, for the SFSR synthesis problem, $\nu$ should be at most $\delta - 2$ so that both \eqref{eq:keyspan} and \eqref{eq:keylocator} result in $\delta - 1 - \nu > 0$ identities, as otherwise they give no information on the error span and locator polynomials, since any skew polynomial of the same degree would then satisfy both sets of zero key equations.
Furthermore, if $\nu \le \delta - 2$, we may define the $(\delta - 1 - \nu)(r+1) \times \nu$ matrices

\begin{equation}\label{eq:HE}
	\def\arraystretch{1.6}
	H_\nu = \left[
		\begin{array}{c}
			\sigma^{k_0 t_2}(\overline H_\nu) \\ \hline
			\sigma^{k_1 t_2}(\overline H_\nu) \\ \hline
			\vdots \\ \hline
			\sigma^{k_r t_2}(\overline H_\nu)
		\end{array}
	\right],
	\qquad
	E_\nu = \left[
		\begin{array}{c}
			\sigma^{-k_0 t_2}(\overline E_\nu) \\ \hline
			\sigma^{-k_1 t_2}(\overline E_\nu) \\ \hline
			\vdots \\ \hline
			\sigma^{-k_r t_2}(\overline E_\nu)
		\end{array}
	\right],
\end{equation}
where
\[
	\overline H_\nu = \left( \sigma^{b + i t_1}(\eta_k) \right)_{\substack{\nu \le i \le \delta - 2 \\ 1 \le k \le \nu}},
	\qquad
	\overline E_\nu = \left( \sigma^{-b - i t_1}(\ep_k) \right)_{\substack{\nu \le i \le \delta - 2 \\ 1 \le k \le \nu}}.
\]

\begin{lemma}\label{lemma:nurank}
	Consider any decomposition $\yv = \cv + \ev$ where $\cv \in \mathcal C$ and $\rkw(\ev) = \nu$, with some associated error values and locators $\epv, \etav$ satisfying \eqref{eq:syndromeeq} and their corresponding error span and locator vectors $\shv, \ldv$.
	If $\nu \le \delta - 2$ and the rank of $H_\nu$ is $\nu$, then the shortest $\sigma^{t_1}$-skew-feedback shift registers generating the $r+1$ sequences of length $\delta - 1$ given by $\s i j$ are the nonzero $F$-multiples of $\shv$.
	If $\nu \le \delta - 2$ and the rank of $E_\nu$ is $\nu$, then the shortest $\sigma^{-t_1}$-skew-feedback shift registers generating the $r+1$ sequences of length $\delta - 1$ given by $\st i j$ are the nonzero $F$-multiples of $\ldv$.
\end{lemma}
\begin{proof}
	The term in line \eqref{eqline:eespan} can be seen as each one of the entries, for each $i$ and each $j$, of the result of right multiplying $H_\nu$ by the column vector whose $\nu$ entries are the left-hand side of \eqref{eq:span} for each $k$. Any $\sigma^{t_1}$-SFSR $\shv'$ of length at most $\nu$ generating $\s i j$ will also satisfy \eqref{eq:keyspan} when substituting $\shv$ with it, and therefore line \eqref{eqline:eespan} will be zero, for all $\nu \le i \le \delta - 2$ and all $0 \le j \le r$, so either $\shv'$ satisfies \eqref{eq:span} and is therefore a multiple of $\shv$, or the rank of $H_\nu$ (respectively, $E_\nu$) is less than $\nu$.
From the point of view of the error locator polynomial, analogous remarks apply to line \eqref{eqline:eelocator}, so any $\sigma^{-t_1}$-SFSR $\ldv'$ of length at most $\nu$ generating $\st i j$ either satisfies \eqref{eq:locator} and hence is a multiple of $\ldv$, or the rank of $E_\nu$ is less than $\nu$.
\end{proof}

\begin{remark}\label{remark:epsilonoreta2}
	The codes $C_{(\sigma, \etav, \T \hv \sigma(\mathcal C))}$ and $C_{(\sigma^{-1}, \epv, \T \hv \sigma(\mathcal C))}$ described in Lemma \ref{lemma:epsilonoreta} are in the left kernel of, respectively, $H_\nu$ and $E_\nu$, as every row in those matrices gives a condition which has to be satisfied by the corresponding code.
	As a result, if the rank of any of those matrices is $\nu$, then the corresponding code has dimension zero and, by the lemma, the step of solving $\epv$ from $\etav$ or vice versa will yield at most one solution.
\end{remark}

Under some conditions on $\nu$ and the parameters $\delta, r, t_2, k_0, \dots, k_r$ that define the set $T$ as given in \eqref{eq:T}, we can guarantee that the rank of both $H_\nu$ and $E_\nu$ as defined in \eqref{eq:HE} will be $\nu$. In order to do so, consider $T_\nu$ for each $0 \le \nu \le \delta - 2$ defined as
\begin{equation}\label{eq:Tnu}
	T_\nu = b + t_1 \nu + t_1 \{ 0, \dots, \delta - 2 - \nu \} + t_2 \{ k_0, \dots, k_r \} \subseteq T.
\end{equation}
$T_\nu$ is a subset of $T$ such that
\[
	\begin{aligned}
		T_\nu - t_1 \{ 0, \dots, \nu \}
		& = \bigcup_{i=0}^\nu \left( b + t_1 i + t_1 \{ 0, \dots, \delta - 2 - \nu \} + t_2 \{ k_0, \dots, k_r \} \right) \\
		& = \bigcup_{i=0}^\nu \left( b + t_1 \{ i, \dots, \delta - 2 - \nu + i \} + t_2 \{ k_0, \dots, k_r \} \right) \\
		& = T,
	\end{aligned}
\]
where $A - B$ is $A + (-1)B$, that is, for two sets $A, B \subset \ZZ$, $A - B = \{ a - b ~|~ a \in A, b \in B \}$.

\begin{lemma}\label{lemma:zerocodes}
	The rank of $H_\nu$ (respectively, $E_\nu$) equals $\nu$ if and only if the code $C_{(\sigma, \etav, T_\nu)}$ (respectively, $C_{(\sigma^{-1}, \epv, T_\nu)}$) only contains the zero codeword.
\end{lemma}
\begin{proof}
	This is immediate from Definition \ref{defn:ccode} observing that $H_\nu$ and $E_\nu$ are, respectively, the transpose matrices of $H_{(\sigma, \etav, T_\nu)}$ and $H_{(\sigma^{-1}, \epv, T_\nu)}$ as defined in \eqref{eq:H}.
\end{proof}

Since $\nu$ is the length of the codes $C_{(\sigma, \etav, T_\nu)}$ and $C_{(\sigma^{-1}, \epv, T_\nu)}$, if the minimum Hamming distances of these codes can be shown to be at least $\nu + 1$, then these codes are trivial, in the sense of only containing the zero codeword.
Any subset in $T_\nu$ implying, through either Theorem \ref{thm:ht} or Theorem \ref{thm:roos}, a minimum distance of at least $\nu + 1$ will result in both codes being trivial.
Hence, the following assumption guarantees that we will be able to correct up to $\nu$ errors.

\begin{assumption}[Error-correcting capacity]\label{assumption:Tnusubset}
	The rank weight $\nu$ of the error vector $\ev$ is such that $\nu \le \delta - 2$ and some subset of $T_\nu$ as defined in \eqref{eq:Tnu} is equal, modulo $\sord$, to a set of one of the following forms:
	\begin{enumerate}
		\item \emph{(Hartmann--Tzeng-like subset}). $b' + t_1' \{ 0, \dots, \delta' - 2 \} + t_2' \{ 0, \dots, r' \}$ where $\delta' \ge 2, r' \ge 0$, $(\sord, t_1') = 1$, $(\sord, t_2') < \delta'$ and $\delta' + r' > \nu$.
		\item \emph{(Roos-like subset}). $b' + t_1' \{ 0, \dots, \delta' - 2 \} + t_2' \{ k_0', \dots, k_{r'}' \}$ where $\delta' \ge 2, r' \ge 0$, $(\sord, t_1') = 1 = (\sord, t_2')$, $k_0' < \dots < k_{r'}'$, $k_{r'}' - k_0' \le \delta' + r' - 2$ and $\delta' + r' > \nu$.
	\end{enumerate}
\end{assumption}

\begin{remark}
	Since $T_\nu \subseteq T_{\nu - 1} \subseteq \dots \subseteq T_1$, if the conditions in Assumption \ref{assumption:Tnusubset} hold for some $\nu > 0$, then they also hold for every integer from $1$ to $\nu$.
\end{remark}

We will review some details on this assumption, including simple citeria that guarantee that it is satisfied for a given $\nu$ (including whether $\nu$ can be $\tau = \lfloor \frac{\delta + r - 1} 2 \rfloor$) whenever the set $T$ already has a Hartmann--Tzeng-like, or Roos-like, structure, in Section \ref{section:ontheassumption}.

The next results follow from Lemma \ref{lemma:nurank} and Lemma \ref{lemma:zerocodes}.

\begin{theorem}\label{thm:keyspan}
	Consider any decomposition $\yv = \cv + \ev$ where $\cv \in \mathcal C$ and $\rkw(\ev) = \nu$, with its associated $\shv$ as in Definition \ref{defn:span}.
	If $\nu$ satisfies Assumption \ref{assumption:Tnusubset}, the shortest $\sigma^{t_1}$-skew-feedback shift registers generating the $r+1$ sequences of length $\delta - 1$ given by $\s i j$ are the nonzero $F$-multiples of $\shv$.
\end{theorem}

\begin{theorem}\label{thm:keylocator}
	Consider any decomposition $\yv = \cv + \ev$ where $\cv \in \mathcal C$ and $\rkw(\ev) = \nu$, with its associated $\ldv$ as in Definition \ref{defn:locator}.
	If $\nu$ satisfies Assumption \ref{assumption:Tnusubset}, the shortest $\sigma^{-t_1}$-skew-feedback shift registers generating the $r+1$ sequences of length $\delta - 1$ given by $\st i j$ are the nonzero $F$-multiples of $\ldv$.
\end{theorem}

Thus, if there is some $\cv \in \mathcal C$ such that its distance to the received $\yv$ satisfies Assumption \ref{assumption:Tnusubset}, both $\shv$ and $\ldv$ can be retrieved by solving a skew-feedback shift-register synthesis problem. Then, as discussed in the previous section, one can retrieve $\epv$ and $\etav$, and finally $\ev$ and $\cv$. Note that, as a consequence of Lemma \ref{lemma:epsilonoreta}, Remark \ref{remark:epsilonoreta2} and Assumption \ref{assumption:Tnusubset}, when solving for $\etav$ in Equation \eqref{eq:stee} or for $\epv$ in Equation \eqref{eq:see}, the only solution will lead to the only codeword at minimum rank distance to $\yv$; in particular, the resulting error locators after solving $\etav$ from \eqref{eq:stee} will be in the $F^\sigma$-span of $\hv$.

Hence, Algorithm \ref{alg:span} and Algorithm \ref{alg:locator} are able to correct errors of rank weight $\nu$ for any $\nu$ satisfying Assumption \ref{assumption:Tnusubset}.

\begin{algorithm}
	\caption{Decoding algorithm through the error span polynomial}\label{alg:span}
	\begin{algorithmic}[1]
		\Input A word $\yv \in F^n$ and parameters $\hv, \sigma, b, \delta \ge 2, t_1, t_2, {k_0, \dots, k_r}$.
		\Require $b + t_1 \{ 0, \dots, \delta - 2 \} + t_2 \{ k_0, \dots, k_r \} \subseteq \T \hv \sigma(\mathcal C)$ and $(\sord, t_1) = 1$.
		\Require There exists $\cv \in \mathcal C$ such that $\rkd(\cv, \yv) = \nu$ satisfies Assumption \ref{assumption:Tnusubset}.
		\Output $\cv \in F^n$.
		\Ensure $\cv$ is the nearest codeword in $\mathcal C$ to $\yv$.
		\Statex % blank line

		\State $\s i j \leftarrow \yv \sigma^{b + t_1 i + t_2 k_j}\left( \hv \right)^T$ for $0 \le i \le \delta - 2, 0 \le j \le r$ \Comment{see \eqref{eq:s}, \eqref{eq:sdeh}}
		\If{$\s i j = 0$ for all computed $i, j$}
			\State \Return $\yv$
		\EndIf
		\State $\shv = (\sh_0, \dots, \sh_\nu) \leftarrow$ the shortest $\sigma^{t_1}$-SFSR generating $\s i j$ \Comment{see Theorem \ref{thm:keyspan}}
		\State Find $\epv = (\ep_1, \dots, \ep_\nu)$ from $\shv$ \Comment{see Remark \ref{remark:shld2epeta}}
		\State $\st i j \leftarrow \sigma^{-b - t_1 i - t_2 k_j}\left(\s i j\right)$ for $0 \le i \le \delta - 2, 0 \le j \le r$ \Comment{see \eqref{eq:st}}
		\State Find $\etav$ from $\epv$ and $\st i j$ by solving \eqref{eq:stee}
		\State Find $B$ from $\etav$ and $\hv$ by solving \eqref{eq:eta}
		\State $\ev \leftarrow \epv B$ \Comment{see \eqref{eq:epsilon}}
		\State \Return $\yv - \ev$
	\end{algorithmic}
\end{algorithm}

\begin{algorithm}
	\caption{Decoding algorithm through the error locator polynomial}\label{alg:locator}
	\begin{algorithmic}[1]
		\Input A word $\yv \in F^n$ and parameters $\hv, \sigma, b, \delta \ge 2, t_1, t_2, {k_0, \dots, k_r}$.
		\Require $b + t_1 \{ 0, \dots, \delta - 2 \} + t_2 \{ k_0, \dots, k_r \} \subseteq \T \hv \sigma(\mathcal C)$ and $(\sord, t_1) = 1$.
		\Require There exists $\cv \in \mathcal C$ such that $\rkd(\cv, \yv) = \nu$ satisfies Assumption \ref{assumption:Tnusubset}.
		\Output $\cv \in F^n$.
		\Ensure $\cv$ is the nearest codeword in $\mathcal C$ to $\yv$.
		\Statex % blank line

		\State $\s i j \leftarrow \yv \sigma^{b + t_1 i + t_2 k_j}\left( \hv \right)^T$ for $0 \le i \le \delta - 2, 0 \le j \le r$ \Comment{see \eqref{eq:s}, \eqref{eq:sdeh}}
		\If{$\s i j = 0$ for all computed $i, j$}
			\State \Return $\yv$
		\EndIf
		\State $\st i j \leftarrow \sigma^{-b - t_1 i - t_2 k_j}\left(\s i j\right)$ for $0 \le i \le \delta - 2, 0 \le j \le r$  \Comment{see \eqref{eq:st}}
		\State $\ldv = (\ld_0, \dots, \ld_\nu) \leftarrow$ the shortest $\sigma^{-t_1}$-SFSR generating $\st i j$ \Comment{see Theorem \ref{thm:keylocator}}
		\State Find $\etav = (\eta_1, \dots, \eta_\nu)$ from $\ldv$ \Comment{see Remark \ref{remark:shld2epeta}}
		\State Find $\epv$ from $\etav$ and $\s i j$ by solving \eqref{eq:see} \label{algstep:epeta}
		\State Find $B$ from $\etav$ and $\hv$ by solving \eqref{eq:eta} \label{algstep:Blocator}
		\State $\ev \leftarrow \epv B$ \Comment{see \eqref{eq:epsilon}}
		\State \Return $\yv - \ev$
	\end{algorithmic}
\end{algorithm}

\begin{remark}\label{remark:Bbeforevalues}
	In Algorithm \ref{alg:locator}, the order of step \ref{algstep:Blocator} and step \ref{algstep:epeta} can be swapped.
\end{remark}

\begin{remark}\label{remark:precomputeH}
	The step of computing $\s i j$ amounts to left multiplying the $n \times (\delta - 1)(r+1)$ matrix $H_{(\sigma, \hv, T)}$ (which has the code $C_{(\sigma,\hv,T)}$, which contains $\mathcal C$, as its left kernel; remember Equation \eqref{eq:H}) by the vector $\yv$. The matrix $H_{(\sigma, \hv, T)}$ might be computed once and then used every time this algorithm is run with the same parameters.
\end{remark}

\begin{remark}\label{remark:algspangabidulin}
	For $b = r = 0$ and $\sigma^{t_1}$ the Frobenius endomorphism of $F$ a finite field (that is, when considering a Gabidulin code or a subcode thereof), Algorithm \ref{alg:span} is equivalent to \cite[Algorithm 3]{SJB11}.
	In fact, under these conditions, if the syndromes are computed using the matrix as described in Remark \ref{remark:precomputeH}, the check of whether the syndromes are zero in Algorithm \ref{alg:span} is dropped, the SFSR synthesis algorithm is the same, a basis of the kernel of the map is computed as roots of a linearized polynomial, and \eqref{eq:stee} is solved as described later in Section \ref{section:gabidulinsalgorithm} (as opposed to solving it as a general linear equation system), then they are the same algorithm up to decoding failure handling.
	Assumption \ref{assumption:Tnusubset} is always satisfied in this case, as we shall note later, shortly after Lemma \ref{lemma:assumptionTnuimpliesassumptionTnusubset}.
\end{remark}

\begin{remark}\label{remark:sfsrredundancy}
	$\shv$ and $\ldv$ are SFSRs generating respectively $\s i j$ and $\st i j$, and therefore they generate the right-hand sides of \eqref{eq:see} and \eqref{eq:stee} respectively. As a result, if $\epv, \etav$ are such that they satisfy \eqref{eq:see} or, equivalently, \eqref{eq:stee}, for all $0 \le i < \nu$ and all $0 \le j \le r$, then they do so for the remaining $i$ and all $j$.
	For example, if we got $\shv$ as a $\sigma^{t_1}$-SFSR of length $\nu$ for $\s i j$ and $\etav$ is such that \eqref{eq:stee}, and therefore the equivalent \eqref{eq:see}, holds for some $j$ and for all $0 \le i < \nu$, then
	\[
		\begin{aligned}
			\s \nu j & = -\sh_0^{-1}\sum_{l=1}^\nu \sh_l \sigma^{t_1 l}\left( \s {\nu-l} j \right) \\
			         & = -\sh_0^{-1}\sum_{l=1}^\nu \sh_l \sigma^{t_1 l}\left( \sum_{k=1}^\nu \ep_k \sigma^{b + t_1 (\nu - l) + t_2 k_j} (\eta_k) \right) \\
			         & = -\sh_0^{-1} \sum_{k=1}^\nu \sigma^{b + t_1 \nu + t_2 k_j}(\eta_k) \left( \sum_{l=1}^\nu \sh_l \sigma^{t_1 l}( \ep_k ) \right) \\
			         & = -\sh_0^{-1} \sum_{k=1}^\nu \sigma^{b + t_1 \nu + t_2 k_j}(\eta_k) (- \sh_0 \ep_k) \\
			         & =  \sum_{k=1}^\nu \ep_k \sigma^{b + t_1 \nu + t_2 k_j}(\eta_k),
		\end{aligned}
	\]
	which is the expression to be satisfied for $\s \nu j$, equivalent to the one for $\st \nu j$. This means that, in the step of solving a linear system, the equations for $\nu \le i \le \delta - 2$, $0 \le j \le r$ are guaranteed to be redundant. Furthermore, if $\nu \le \delta - 1$, which is implied by Assumption \ref{assumption:Tnusubset}, any set of $\nu$ consecutive equations for some fixed $0 \le j \le r$ from \eqref{eq:see} or \eqref{eq:stee} will have a unique solution since \cite[Corollary 4.13]{LL88} applies to the $\nu \times \nu$ matrix associated to those equations; hence, if $\nu \le \delta - 1$ and there exists a solution $\epv, \etav$, the first $\nu$ equations from \eqref{eq:see} or \eqref{eq:stee} will yield that solution.
\end{remark}

These algorithms may fail, either being unable to complete some step or returning an element not in the code $\mathcal C$, if no codeword $\cv$ is at a distance $\nu$ from $\yv$ such that $\nu$ satisfies Assumption \ref{assumption:Tnusubset}. However, in Section \ref{section:toomanyerrors} we will show that, if any of these algorithms finishes and returns a valid codeword, it is guaranteed to be the one at minimum rank distance to the input $\yv$. Stated otherwise, any decoding failure either results in some step of the algorithm not being completed, or can be detected by checking whether the output is a codeword: the algorithms will not return a valid codeword which is not the nearest one to $\yv$. Section \ref{section:toomanyerrors} explores this in further detail.

\section{On the nature of Assumption \ref{assumption:Tnusubset}}\label{section:ontheassumption}

We have shown that we can correct up to $\nu$ errors providing that Assumption \ref{assumption:Tnusubset} holds for $\nu$. This requires finding some suitable subset of $T_\nu$, as defined in the assumption.
If $T$ is already defined as some set that guarantees a Hartmann--Tzeng-like or Roos-like bound for the minimum distance of the code through Theorem \ref{thm:ht} or Theorem \ref{thm:roos} respectively, then the most immediate subset of $T_\nu$ that can be considered in Assumption \ref{assumption:Tnusubset} is the very set $T_\nu$ with  the same structure as $T$ except replacing $b + t_1 \nu$ for $b$ and $\delta - \nu$ for $\delta$.
If $\nu$ is low enough, we may also directly find a BCH-like subset of $T_\nu$.
This results in the following conditions:

\begin{assumption}\label{assumption:Tnu}
	The rank weight $\nu$ of the error vector $\ev$ is such that at least one of the following statements is true:
	\begin{enumerate}
		\item \emph{(BCH-like set}). $\nu \le \lfloor \frac{\delta - 1}2 \rfloor$ (equivalently, $2\nu \le \delta - 1$).
		\item \emph{(Hartmann--Tzeng-like set}). $\nu \le \tau = \lfloor \frac{\delta + r - 1} 2 \rfloor$, $(\sord, t_2) < \delta - \nu$ (equivalently, $\nu \le \delta - (\sord, t_2) - 1$) and $\{ k_0, \dots, k_r \} = \{ 0, \dots, r\}$.
		\item \emph{(Roos-like set}). $\nu \le \tau = \lfloor \frac{\delta + r - 1} 2 \rfloor$, $(\sord, t_2) = 1$ and $k_r - k_0 \le \delta + r - \nu - 2$.
	\end{enumerate}
\end{assumption}

\begin{remark}\label{remark:uptotau}
	Assumption \ref{assumption:Tnu} implies that $\nu \le \tau = \lfloor \frac{\delta + r - 1} 2 \rfloor$.
	It also implies that $\nu \le \delta - 2$, as $\delta \ge 2$ (hence $\lfloor \frac{\delta-1}2\rfloor \le \lfloor \frac{2\delta - 3} 2 \rfloor = \delta - 2$), $(\sord, t_2) \ge 1$, and $r \le k_r - k_0 \le \delta + r - \nu - 2$.
	Combining this two observations, $\nu \le \nu + 2(\tau - \nu) = 2 \tau - \nu \le (\delta - 1 - \nu) + r \le (\delta - 1 - \nu)(r+1)$.
\end{remark}

\begin{remark}
	The first statement in Assumption \ref{assumption:Tnu} is not redundant. It does not impose conditions on $r$ or $t_2$: when it is satisfied, then only one block is actually necessary for decoding. Thus, this first statement describes a BCH-like scenario, as if $r = 0$. It is not hard to find circumstances where this first statement allows $\nu$ to take greater values than the ones covered by the second and the third statements: for example, for any case where $(\sord, t_2) = \delta - 1 > 1$, as well as any case such that $k_r - k_0 = \delta + r - 2$ and $\delta > 2$, only this first statement will make $\nu = 1$ work.
\end{remark}

\begin{lemma}\label{lemma:assumptionTnuimpliesassumptionTnusubset}
	Assumption \ref{assumption:Tnu} implies Assumption \ref{assumption:Tnusubset}. In addition, each statement in Assumption \ref{assumption:Tnu} implies that the minimum distance of the code is at least, respectively, $\delta$, $\delta + r$ and $\delta + r$, and therefore $\nu$ is within the error-correcting capacity of the code.
\end{lemma}
\begin{proof}
	Remark \ref{remark:uptotau} has shown that Assumption \ref{assumption:Tnu} implies that $\nu \le \delta - 2$.
	If the first statement in Assumption \ref{assumption:Tnu} is true, then $b + t_1 \nu + t_1 \{ 0, \dots, \delta - \nu - 2 \} + 1\{ 0 \} \subseteq T_\nu$ fits both of the conditions in Assumption \ref{assumption:Tnusubset}, since $\delta' = \delta - \nu \ge \delta - \lfloor \frac{\delta - 1} 2 \rfloor > \lfloor \frac{\delta - 1} 2 \rfloor \ge \nu$, while $b + t_1 \{ 0, \dots, \delta - 2 \} + 1\{ 0 \} \subseteq T$ implies a lower bound of $\delta$ by both Theorem \ref{thm:ht} and Theorem \ref{thm:roos}.
	For the second and third statements, the set $T_\nu$ sastisfies respectively the first and the second condition in Assumption \ref{assumption:Tnusubset}, since $\nu$ being at most $\tau$ implies that $\delta' + r' = \delta - \nu + r \ge \delta + r - \lfloor \frac{\delta + r - 1} 2 \rfloor > \lfloor \frac{\delta + r - 1} 2 \rfloor \ge \nu$. Theorem \ref{thm:ht} and Theorem \ref{thm:roos} apply to $T$ respectively, for each statement, returning a lower bound of $\delta + r$.
\end{proof}

When the entries of $T$ are consecutive (and therefore $r = 0$), as it is the case for Gabidulin codes (\cite{SKK08,SJB11}) as well as the skew codes considered in \cite{GLN17s,GLN17pgz}, Assumption \ref{assumption:Tnu} is simply equivalent to $\nu$ being at most the error-correcting capacity $\tau$, which is usually already assumed when considering nearest-neighbor error correction. Hence, no assumption has to be considered in those works.
In general, Assumption \ref{assumption:Tnu} might not allow $\nu$ to reach the error-correcting capacity $\tau$ implied by those bounds. Firstly, $\tau$ might be above $\delta - 2$ if $r \ge \delta - 1$. Additionally, it might be the case that $(\sord, t_2)$ is less than $\delta$ but not less than $\delta - \tau$, or $k_r - k_0$ is at most $\delta + r - 2$ but greater than $\delta + r - \tau - 2$. In these situations, the two last statements in Assumption \ref{assumption:Tnu} will be false for $\nu = \tau$, even if a Hartmann--Tzeng, or Roos, bound applies to $T$ guaranteeing an error-correcting capacity of $\tau$.

We shall now characterize the situations where $\nu = \tau$ satisfies Assumption \ref{assumption:Tnu} and therefore we can reach the error-correcting capacity given by Theorem \ref{thm:ht} or Theorem \ref{thm:roos}.

\begin{assumption}\label{assumption:sigmat2dr}
	At least one of the following statements holds:
	\begin{enumerate}
		\item $r = 0$ (that is, a BCH-like case).
		\item $\delta$ is odd and $r = 1$.
		\item $(\sord, t_2) < \delta - \tau$ (i.e. $(\sord, t_2) \le \lfloor \frac{\delta - r} 2 \rfloor$) and $\{ k_0, \dots, k_r \} = \{ 0, \dots, r \}$.
		\item $(\sord, t_2) = 1$ and $k_r - k_0 \le \delta + r - \tau - 2$ (equivalently, $k_r - k_0 < \lfloor \frac {\delta + r} 2 \rfloor$; also equivalent to the number of integers missing in the range from $k_0$ to $k_r$, which is $k_r - k_0 - r$, being less than $\delta - \tau - 1 = \lfloor \frac{\delta - r} 2 \rfloor$).
	\end{enumerate}
\end{assumption}

The equivalences are a consequence of the following identity:
\[
	1 + \left \lfloor \frac{\delta - r} 2 \right \rfloor
	= \left \lfloor 1 + \delta - \frac{\delta+r}2 \right \rfloor
	= \delta + \left \lfloor - \frac{\delta + r - 1}{2} + \frac 1 2 \right \rfloor
	= \delta + \left \lceil - \frac{\delta + r - 1}{2} \right \rceil
	= \delta - \tau.
\]

\begin{lemma}\label{lemma:assumptionsigmat2drimpliesuptotau}
	Assumption \ref{assumption:Tnu} holds for $\nu = \tau$ if and only if Assumption \ref{assumption:sigmat2dr} holds.
\end{lemma}
\begin{proof}
	When $r = 0$, the first statement in Assumption \ref{assumption:Tnu} is simply $\nu \le \tau$, as in this case $\tau = \lfloor \frac{\delta + r - 1} 2 \rfloor = \lfloor \frac{\delta - 1} 2 \rfloor$. For the same reason, this first statement also reduces to $\nu \le \tau$ when $\delta$ is odd and $r = 1$. Thus, in these cases, Assumption \ref{assumption:Tnu} is satisfied for $\nu = \tau$.
	It is immediate that the third and fourth statement imply the second and third statement, respectively, of Assumption \ref{assumption:Tnu}.

	Now assume that none of these four statements holds. Then, either $r \ge 2$ or $\delta$ is even, so $\tau = \lfloor \frac{\delta + r - 1}2 \rfloor > \lfloor \frac{\delta - 1}{2} \rfloor$, therefore the first statement in Assumption \ref{assumption:Tnu} is false. It is straightforward that the remaining statements in Assumption \ref{assumption:Tnu} are also false.
\end{proof}

\begin{remark}\label{remark:rdelta}
	Any of the statements in Assumption \ref{assumption:sigmat2dr} implies that $r \le \delta - 2$. Therefore, by the previous result, Assumption \ref{assumption:Tnu} will not hold for $\nu = \tau$ if $r > \delta - 2$. In general, an increase of $1$ in either $(\sord, t_2)$ or $k_r - k_0$ increases in $1$ the maximum value for $r$ in the third and, respectively, fourth statement in Assumption \ref{assumption:sigmat2dr}.
\end{remark}
	
\begin{remark}
	In \cite[Section VI.B]{FT91}, which deals with decoding up to the Roos bound for cyclic codes, it is assumed (in addition to a fact analogous to $(\sord, t_2)$ being $1$) that $r \le \delta - 2$ and $k_r - k_0 < \tau$; however, since $\tau \le \lfloor \frac {\delta + r} 2 \rfloor$, the latter condition implies that $k_r - k_0 < \lfloor \frac {\delta + r} 2 \rfloor$, which itself implies that $r \le \delta - 2$, making this condition redundant. As a result, the assumptions required in \cite[Section VI.B]{FT91} are the cyclic analogous ones from a particular case of the fourth statement in Assumption \ref{assumption:sigmat2dr}. The proof in \cite[Section VI.B]{FT91} can still be made to work for $k_r - k_0 = \tau < \lfloor \frac {\delta + r} 2 \rfloor$: the condition $k_r - k_0 < \tau$ is used in order to show twice that $\tau + k_r - k_0 < \delta + r - 1$ (written firstly as $\delta - 2 - \tau + \tau - (k_r - r - k_0) > \tau$ and then as $\delta - 2 - (k_r - r - k_0) + 1 > \tau$); this is achieved from the fact that $\tau + k_r - k_0 < 2\lfloor \frac{\delta + r - 1} 2 \rfloor \le \delta + r - 1$, but $\lfloor \frac{\delta + r - 1} 2 \rfloor + \lfloor \frac{\delta + r} 2 \rfloor$ is also at most (in fact, equals) $\delta + r - 1$.
\end{remark}

One can also consider the original Hartmann--Tzeng bound, as the analogue of the one described in \cite{HT72}, which is the one given in Theorem \ref{thm:ht} requiring, additionally, that $(\sord, t_2) = 1$. This also results in a case of the bound in Theorem \ref{thm:roos}. The interest of this case rests in the fact that, once $(\sord, t_1) = 1 = (\sord, t_2)$, if $r > \delta - 2$, which is required to be false as noted in Remark \ref{remark:rdelta}, one can rearrange the parameters $\delta, r, t_1, t_2$ so that the set $T$ is seen as $b + t_2 \{ 0, \dots, \delta' - 2 = r \} + t_1 \{ 0, \dots, r' = \delta - 2 \}$ and then $r' < \delta' - 2$. Furthermore, in this case, Assumption \ref{assumption:sigmat2dr} is reduced to $r$ being at most $\delta - 2$, which can be achieved by rearranging the parameters if necessary:

\begin{assumption}\label{assumption:1rd}
	$(\sord, t_2) = 1$, $\{ k_0, \dots, k_r \} = \{ 0, \dots, r \}$ and $r \le \delta - 2$.
\end{assumption}

\begin{lemma}
	Assume $(\sord, t_2) = 1$ and $\{ k_0, \dots, k_r \} = \{ 0, \dots, r \}$. Then, Assumption \ref{assumption:sigmat2dr} holds if and only if $r \le \delta - 2$. In particular, if Assumption \ref{assumption:1rd} holds, then Assumption \ref{assumption:Tnu} holds for $\nu = \tau$.
\end{lemma}
\begin{proof}
	Remark \ref{remark:rdelta} has shown that Assumption \ref{assumption:sigmat2dr} implies that $r \le \delta - 2$. Under Assumption \ref{assumption:1rd}, the fourth statement in Assumption \ref{assumption:sigmat2dr} is true since $k_r - k_0 = r$ and therefore $k_r - k_0 - r = 0 < 1 \le \lfloor \frac{\delta - r} 2 \rfloor$.
\end{proof}

\begin{remark}
	As
	$\lceil \frac{\delta + r - 2}2 \rceil \le \tau = \lfloor \frac{\delta + r - 1}2 \rfloor$,
	if $r \le \delta - 2$ we get that $\tau \le \lfloor \frac{2\delta - 3}2 \rfloor = \delta - 2$, while $r \ge \delta - 1$ would imply that $\tau \ge \lceil \frac{2 \delta - 3}2 \rceil = \delta - 1$. This means that $r$ might be replaced with $\tau$ in the last condition of Assumption \ref{assumption:1rd}.
\end{remark}

\begin{remark}
	An analogous version of Assumption \ref{assumption:1rd} is taken in \cite[Section VI.A]{FT91}, which discusses decoding up to the original Hartmann--Tzeng bound in the cyclic case, and also notes the possibility of rearranging the set $T$ when $r > \delta - 2$ since both $t_1$ and $t_2$ are assumed to be coprime with the analogous of $\sord$.
\end{remark}

Another relevant question is whether the error correction capacity allowed by Assumption \ref{assumption:Tnu} is larger than the one covered by the equivalent BCH scenario (that is, replacing $r$ by $0$ and leaving all the other parameters the same).
In the Hartmann--Tzeng case, where $k_0, \dots, k_r$ are consecutive, this happens exactly when $\lfloor \frac{\delta - 1}2 \rfloor$ is less than both $\tau = \lfloor \frac{\delta + r - 1} 2 \rfloor$ and $\delta - (\sord, t_2) - 1$, as this means that the second statement of Assumption \ref{assumption:Tnu} allows for a greater $\nu$ than the first one. The condition $\lfloor \frac{\delta - 1}2 \rfloor < \tau$ is satisfied when either $r \ge 2$, or $r = 1$ and $\delta$ is even, while the other one is satisfied when $(\sord, t_2)$ is less than $\delta - 1 - \lfloor \frac{\delta - 1}2 \rfloor = - \lfloor \frac {1 - \delta} 2 \rfloor = \lceil \frac {\delta - 1} 2 \rceil = \lfloor \frac \delta 2 \rfloor$.
In the Roos case, where $(\sord, t_2) = 1$, this requires that $\lfloor \frac{\delta - 1}2 \rfloor$ is less than both $\tau$ and $\delta + r - 2 - k_r + k_0$. Again, the first condition on $\lfloor \frac{\delta - 1}2 \rfloor$ occurs when either $r \ge 2$, or $r = 1$ and $\delta$ is even, while the second one is fulfilled if and only if $k_r - k_0$ is less than $\delta + r - 2 - \lfloor \frac{\delta - 1} 2 \rfloor = \lfloor \frac \delta 2 \rfloor + r - 1$.

\section{Identifying and handling decoding failures} \label{section:toomanyerrors}

	Algorithm \ref{alg:span} and Algorithm \ref{alg:locator} consider as a precondition that there exists a codeword $\cv$ such that $\rkd(\cv, \yv) = \nu$ satisfies Assumption \ref{assumption:Tnusubset}. Hence, they may fail if the rank distance from $\yv$ to any codeword, including the one that was probably meant to be received, is above the error-correcting capacity implied by Assumption \ref{assumption:Tnusubset}. In this section, we shall describe the errors that might arise in such a situation.
	
	For this purpose, we now must distinguish between $\nu$ the minimum distance between $\yv$ and any codeword in $\mathcal C$, and $\bar \nu$ the length of the SFSR that is returned in the algorithm, that is, of the shortest $\sigma^{t_1}$-SFSR generating $\s i j$ in the case of Algorithm \ref{alg:span} or the shortest $\sigma^{-t_1}$-SFSR generating $\st i j$ for Algorithm \ref{alg:locator}. Since there is at least one codeword $\cv$ such that $\rkd(\cv, \yv) = \nu$, and therefore some SFSR of length $\nu$ generating the sequences (see Proposition \ref{prop:keyspan} and Proposition \ref{prop:keylocator}), $\bar \nu \le \nu$. This prevents the algorithms from returning a valid codeword that is not at the minimum possible distance to $\yv$, since the resulting error vector $\ev$ is computed as $\epv B$ and, as a consequence, its rank weight is at most the length of $\epv$, which is $\bar \nu$. Thus, if the algorithm returns a valid codeword $\cv \in \mathcal C$, then $\rkd(\yv, \cv)$ is the minimum possible, even if $\nu$ does not satisfy Assumption \ref{assumption:Tnusubset}. Hence, there is a decoding failure (that is, the algorithm cannot retrieve a codeword at minimum distance to $\yv$) if and only if either the algorithm is unable to complete some step, or it returns an invalid codeword, in the sense that its output is not in $\mathcal C$.

	Assumption \ref{assumption:Tnusubset} is required in order to guarantee that the shortest SFSR contains the coefficients for an error span polynomial in Algorithm \ref{alg:span} or an error locator polynomial in Algorithm \ref{alg:locator}. If the SFSR synthesis algorithm returns a correct error span polynomial $\shv$ or error locator polynomial $\ldv$ for some codeword (whose distance to $\yv$, as noted above, will be minimal), then the remaining steps will succeed in returning a codeword, even if the observed $\bar \nu$ does not satisfy Assumption \ref{assumption:Tnusubset}. There is, therefore, a decoding failure if and only if the obtained $\shv$ is not associated to an error span polynomial (in Algorithm \ref{alg:span}) or $\ldv$ is not associated to an error locator polynomial (in Algorithm \ref{alg:locator}).
	This will be the case if there exists some SFSR generating the sequences which is shorter than the ones associated to correct error span or locator polynomials, and can also be expected if there is such an SFSR of the same length, since the SFSR synthesis algorithms might not select the correct SFSR over the others.
	The condition $\nu \le \delta - 2$ required by Assumption \ref{assumption:Tnusubset} prevents the shortest SFSR from having to satisfy zero equations, while the rest of the assumption prevents, through Lemmas \ref{lemma:nurank} and \ref{lemma:zerocodes}, the emergence of other solutions of equal, or shorter, length through exploiting linear dependencies between the rows of $H_{\nu}$ or $E_{\nu}$.
	Hence, such a situation results if $\nu > \delta - 2$ or if the matrices $H_{\nu}$ and $E_{\nu}$ insert solutions into \eqref{eqline:eespan} and, respectively, \eqref{eqline:eelocator}. This could result in the following errors while running the algorithm.
	
	Once $\shv$ or $\ldv$ are not related to an error span or, respectively, locator polynomial, the most immediate error that could be found is the $F^\sigma$-dimension of, respectively, $\ker(\sh(\sigma^{t_1}))$ and $\ker(\ld(\sigma^{-t_1}))$, not reaching the observed $\bar \nu$. This can be the case, for example, if the SFSR is $(1, 0, \dots, 0)$, which might be a minimum length SFSR if $\nu > \delta - 2$ as noted above. This would prevent the algorithm from getting $\epv$ from $\shv$ or $\etav$ from $\ldv$.

	If $r > 0$ and $\epv$ or $\etav$ is invalid (in the sense that the algorithm is able to construct it from $\shv$ or $\ldv$ respectively, but it is not part of a valid solution), the linear system \eqref{eq:see} or, respectively, \eqref{eq:stee} might have no solutions, so the algorithm might fail at the step of computing such a solution in order to have both $\etav$ and $\epv$. If $r = 0$, this step will not fail:
	by Remark \ref{remark:sfsrredundancy}, which applies even if $\etav$, $\epv$ are not from a valid solution, the linear system \eqref{eq:see} or \eqref{eq:stee} is equivalent to the one formed from its first $\min(\bar \nu, \delta - 1)$ equations, whose associated matrix, by \cite[Corollary 4.13]{LL88} and the fact that the error values and the error locators are computed as the basis of an $F^\sigma$-space, has $\min(\bar \nu, \min(\bar \nu, \delta - 1)) = \min(\bar \nu, \delta - 1)$ as its rank. Since this rank matches the number of equations, the linear system has at least one solution.	

	The step of computing $B$ from $\etav$ and $\hv$ amounts to finding the coordinates of the error locators with respect to $\hv$, which is a basis for a $F^\sigma$-subspace of $F$. This step fails if and only if some entry in the computed $\etav$ is not in the $F^\sigma$-span of the elements in $\hv$. In particular, if $n = \sord$, this step will not fail.

	The output $\cv$ is computed as $\yv - \ev = \yv - \epv B$, where $B$ is chosen so that \eqref{eq:eta} is satisfied, and $\epv$ and $\etav$ are such that \eqref{eq:syndromeeq} holds (for $d \in T \subset \T \hv \sigma(\mathcal C)$, where $T$ is as in \eqref{eq:T}). Hence, the $d$-th syndrome, for $d \in T$, of $\ev$ equals the one of $\yv$, and therefore $\cv$ is a codeword in $C_{(\sigma,\hv,T)}$. If the considered code $\mathcal C$ equals $C_{(\sigma,\hv,T)}$, instead of a proper subcode thereof, any output for the algorithm will be correct, so in that case any decoding failure will consist in the algorithm not being able to finish.
	
	As a result, when no codeword exists at a distance $\nu$ from $\yv$ satisfying Assumption \ref{assumption:Tnusubset}, the following errors might (or might not) occur:

	\begin{enumerate}
		\item The dimension of $\ker(\sh(\sigma^{t_1}))$ or $\ker(\ld(\sigma^{-t_1}))$ might be less than the length $\bar \nu$ of the obtained SFSR. This would prevent $\epv$ or $\etav$ from being constructed from $\shv$ or, respectively, $\ldv$.
		\item \label{failure:nosolution} If $r > 0$, there might be no solutions for \eqref{eq:see} or \eqref{eq:stee}. This causes the step of obtaining $\etav$ from $\epv$ or vice versa to fail.
		\item \label{failure:notinhspan} If $n < \sord$, some error locators might not be in the $F^\sigma$-span of $\hv$, so $B$ cannot be computed.
		\item \label{failure:notincode} If $\mathcal C \neq C_{(\sigma,\hv,T)}$, the retrieved codeword might not be in $\mathcal C$. Hence, the probable last step of getting the message from the codeword would fail.
	\end{enumerate}
	Handling a decoding failure is equivalent to checking for the possible failures above in Algorithm \ref{alg:span} or Algorithm \ref{alg:locator}. By handling these errors, returning `decoding failure' once any of them is found, said precondition can be omitted; instead, we may add the following postcondition: if `decoding failure' is returned, then no codeword is at a distance satisfying Assumption \ref{assumption:Tnusubset}.

	If $\bar \nu \le \delta - 1$ and, as suggested in Remark \ref{remark:sfsrredundancy}, only $\bar \nu$ equations are considered when solving the linear system that relates $\etav$ and $\epv$, then this linear system will be guaranteed to have a single solution, skipping failure \ref{failure:nosolution}. This is equivalent to finding a solution as if the set $T$ was reduced to $b + t_1 \{ 0, \dots, \delta - 2 \}$, so if the solution is not a solution for the full system \eqref{eq:see} or \eqref{eq:stee}, the failure \ref{failure:notincode} is bound to happen even if $\mathcal C$ matches $C_{(\sigma,\hv,T)}$, unless the failure \ref{failure:notinhspan} is triggered beforehand.

	Note that observing some $\bar \nu$ satisfying Assumption \ref{assumption:Tnusubset} does not guarantee a successful decoding, since it might be the case that $\bar \nu < \nu$, which would in fact lead to a failure since no codewords are at distance $\bar \nu$ to $\yv$. Similarly, a failure is guaranteed if $\rkw(\ev)$ is less than $\bar \nu$ and therefore less than $\nu$, which can happen if either $\rkw(\epv) < \bar \nu$ or $\rk(B) < \bar \nu$; however, these need not be checked, since a decoding failure will be raised when attempting to retrieve the message from $\cv$.

\section{Gabidulin's algorithm}
\label{section:gabidulinsalgorithm}

One step of the decoding algorithms involves solving either \eqref{eq:see} for $\epv$, when $\etav$ has been obtained from an error locator vector, or the equivalent \eqref{eq:stee} for $\etav$, if an error span vector led to $\epv$. By Remark \ref{remark:sfsrredundancy}, if $\nu \le \delta - 1$, in order to solve, these equations only need to be considered for the range $0 \le i \le \nu - 1$ and for any fixed $j$ in $\{ 0, \dots, r\}$, for example $j = 0$.

Given the structure of the coefficient matrices associated to these equation systems, there is an algorithm for computing the solution which is more efficient than a general linear equation system solving algorithm, described in \cite[Section 6]{Gabidulin85} in the context of a finite field, $t_1 = 1$ and $b = r = 0$ which is considered in Gabidulin codes. In addition, by using this approach, one may skip the computation of $\st i j$ from $\s i j$ in Algorithm \ref{alg:span} or $\s i j$ from $\st i j$ in Algorithm \ref{alg:locator}, since only the coefficients $\s i j$ are needed for the former algorithm, while the latter only requires $\st i j$, which can directly be computed as $\sigma^{-b - t_1 i - t_2 k_j}(\yv) \hv^T$ (although this requires applying a power of $\sigma$ $n$ times, which under most circumstances is less than the required $(\delta - 1)(r+1)$ times it has to be done at the beginning of Algorithm \ref{alg:locator}, note that by Remark \ref{remark:precomputeH} the $n$ times it is applied for computing $\s i j$ might not be repeated after the first execution). We will now describe this procedure in our context by applying the ideas in \cite[Section 6]{Gabidulin85}.

Provided that $\nu \le \delta - 1$, the first $\nu$ equations of both \eqref{eq:see} and \eqref{eq:stee} have the form
\begin{equation}\label{eq:gabidulinsystem}
	b_i = \sum_{k=1}^\nu a_k \theta^{\bar b + t_1 i} \left( X_k \right) \qquad \text{ for all } 0 \le i \le \nu - 1
\end{equation}
(cf. \cite[Equation (42)]{Gabidulin85}) where $\bar b = b + t_2 k_0$ and the unknowns are $X_1, \dots, X_\nu$. This is \eqref{eq:see} for $a_k = \ep_k$, $b_i = \s i 0$, $X_k = \eta_k$ and $\theta = \sigma$, while in the case of \eqref{eq:stee}, $a_k = \eta_k$, $b_i = \st i 0$, $X_k = \ep_k$ and $\theta = \sigma^{-1}$. Note that the coefficients $a_k$ are $F^{\theta^{t_1}}$-linearly independent (remember Remark \ref{remark:sigmat1}); in particular, they are nonzero.

If $\nu = 1$, this is readily solved as $X_k = \theta^{- \bar b}\left( a_1^{-1} b_0 \right)$. Otherwise, let $E_i$ be the $(i+1)$-th equation in \eqref{eq:gabidulinsystem}, i.e. the left-hand side of $E_i$ is $b_i$ for each $0 \le i \le \nu - 1$.
Then, for any $0 \le i \le \nu - 2$ and any $1 \le k \le \nu$, $\theta^{\bar b + t_1 i} \left( X_k \right)$ is present in both $E_i$ and $\theta^{-t_1}\left( E_{i+1} \right)$, multiplied respectively by $a_k$ and $\theta^{-t_1}(a_k)$. Hence, the equation $E_i - \frac{a_k}{\theta^{-t_1}(a_k)} \theta^{-t_1}\left( E_{i+1} \right)$ does not depend on $X_k$. Therefore,
one might replace the system \eqref{eq:gabidulinsystem} with the one given by the union of $\{ E_i - \frac{a_1}{\theta^{-t_1}(a_1)} \theta^{-t_1}\left( E_{i+1} \right) ~|~ 0 \le i \le \nu - 2 \}$, which has $\nu - 1$ equations and $\nu - 1$ unknowns since it does not depend on $X_1$, and $\{ E_0 \}$. The first $\nu - 1$ equations have the form
\begin{equation}\label{eq:gabidulinneweq}
	b_i - \frac{a_1 \theta^{-t_1}(b_{i+1})}{\theta^{-t_1}(a_1)} = \sum_{k=2}^\nu \left( a_k - \frac{a_1 \theta^{-t_1}(a_k )}{\theta^{-t_1}(a_1)}  \right) \theta^{\bar b + t_1 i} \left( X_k \right) \quad \text{ for all } 0 \le i \le \nu - 2,
\end{equation}
which matches the structure in \eqref{eq:gabidulinsystem} by taking $b_i - a_1 \theta^{-t_1}\left( a_1^{-1} b_{i+1} \right)$ as $b_i$, $a_k - a_1 \theta^{-t_1}\left( a_1^{-1} a_k \right)$ as $a_k$, and $\nu - 1$ as $\nu$.
The new values for $a_k$ are also $F^{\theta^{t_1}}$-linearly independent, hence nonzero: if $\sum_{k=2}^\nu c_k \left( a_k - a_1 \theta^{-t_1}\left( a_1^{-1} a_k \right) \right) = 0$ for some $c_2, \dots, c_\nu \in F^{\sigma^{t_1}}$, then for $C = \sum_{k=2}^\nu c_k a_k$ we get $C = a_1 \theta^{-t_1}\left( a_1^{-1} C \right)$ and therefore $a_1^{-1} C = \theta^{-t_1}\left( a_1^{-1} C \right)$, so $a_1^{-1}C \in F^{\theta^{t_1}}$ and there is some $c_1 \in F^{\theta^{t_1}}$ such that $0 = c_1 a_1 + C = \sum_{k=1}^\nu c_k a_k$, which by the $F^{\theta^{t_1}}$-linear independence of $a_1, \dots, a_\nu$ means that $c_2, \dots, c_\nu$ are all zero.
Therefore, the same can be done for these $\nu - 1$ equations. This leads to recursively solving a system of the form \eqref{eq:gabidulinsystem} for any size $\nu > 1$ by constructing and solving a system of size $\nu - 1$ and then solving $X_1$ as $\theta^{- \bar b}\left( a_1^{-1} \left( b_0 - \sum_{k=2}^\nu a_k \theta^{\bar b}\left( X_k \right) \right) \right)$. 
Each recursion step consists of constructing the coefficients for the new system of equations (which requires $\bigO(\nu)$ operations in the field $F$ including subtractions, products, divisions and applications of powers of the map $\theta$), solving this system of equations, and then solving for the first variable (which again requires $\bigO(\nu)$ operations in $F$). This means that this procedure is done in $\bigO(\nu^2)$ operations in $F$.

We may describe this recursion explicitly, as in \cite[Eqs. (46)--(48)]{Gabidulin85}, by defining
\begin{align}
A_k^{(1)} & = a_k & \text{ for each } & 1 \le k \le \nu, \\
B_i^{(1)} & = b_i & \text{ for each } & 0 \le i \le \nu - 1, \\
\label{eq:Gabidulinakj}
A_k^{(j)} & = A_k^{(j-1)} - A_{j-1}^{(j-1)} \theta^{-t_1}\left( \frac { A_k^{(j-1)} }{ A_{j-1}^{(j-1)} } \right) & \text{ for each } & 2 \le j \le k \le \nu, \\
\label{eq:Gabidulinbij}
B_i^{(j)} & = B_i^{(j-1)} - A_{j-1}^{(j-1)} \theta^{-t_1}\left( \frac { B_{i+1}^{(j-1)} }{ A_{j-1}^{(j-1)} } \right) & \text{ for each } & 2 \le j \le \nu, 0 \le i \le \nu - j.
\end{align}

Then, $X_\nu, X_{\nu - 1}, \dots, X_2, X_1$ are solved in this order as
\begin{equation}\label{eq:GabidulinXk}
	X_k = \theta^{-\bar b}\left( \frac{ B_0^{(k)} - \sum_{l=k+1}^{\nu} A_{l}^{(k)} \theta^{\bar b}\left( X_l \right) }{ A_k^{(k)} } \right).
\end{equation}

\section{An example}\label{section:example}

In this section, the usage of Algorithm \ref{alg:span} and Algorithm \ref{alg:locator} will be illustrated in the context of an example.
First, Theorem \ref{thm:ht} and Theorem \ref{thm:roos} will be applied in order to get lower bounds for the minimum rank distance of a given code.
This will be followed by an overview on how to use Assumption \ref{assumption:Tnusubset} and the other possible assumptions discussed in Section \ref{section:ontheassumption} to guarantee that decoding will be successful up to some error-correcting capacity.
Once that an error-correcting capacity has been established, the steps for running Algorithm \ref{alg:span} and Algorithm \ref{alg:locator} with a corrupted codeword will be summarized.
By tweaking a successful decoding example, the possible decoding failures discussed in Section \ref{section:toomanyerrors} will be encountered.

\subsection{Applying the bounds}

Let $F$ be the finite field with $2^{14}$ elements $\FF_{2^{14}} = \FF_2(a)$ where $a \in F$ is such that $a^{14} + a^7 + a^5 + a^3 + 1 = 0$, and take as $\sigma : F \to F$ the Frobenius endomorphism, $\gamma \mapsto \gamma^2$ for all $\gamma \in F$, whose order is $\sord = 14$ and whose fixed field $F^\sigma$ has two elements.
The nonzero elements in $F$ will be denoted throughout this section as powers of $a$ in order to reduce the length of the expressions involving elements in $F$. They will be shown as polynomials on $a$ of degree less than $14$ whenever the sum of some such elements has to be considered, as the sum is immediate with this representation.

For $\alpha = a^7$, $\hv = \left( \alpha, \sigma(\alpha), \dots, \sigma^{13}(\alpha) \right)$, which was denoted by $\av$ in Proposition \ref{prop:skewcyclicdefiningsets}, is an $F^\sigma$-basis of $F$. Therefore, we may consider $F$-linear codes of length $14$ of the form $C_{(\sigma, \hv, T)}$ by choosing some set $T$.
By Proposition \ref{prop:skewcyclicdefiningsets}, such codes will be skew cyclic codes; furthermore, every possible option for $\hv$ is $\av M$ for $M$ some $n$-rank $14 \times n$ matrix with entries in $F^\sigma$, for some $n \le 14$. These codes will not be skew cyclic if $n < 14$ and might not if $n = 14$. As noted in Remark \ref{remark:equivalentorshortenedfromskewcyclic}, any such code will be a $(14 - n)$-times shortened code from a code equivalent to the corresponding skew cyclic code with $\hv = \av$ and the same $T$.
For our example, we will consider $\hv = \av$, getting a skew cyclic code.

If $T$ is chosen as $\{ 0,1,2,3,4,\allowbreak 8,9,10,11,12 \} = \{ 0,1,2,3,4 \} + \{ 0, 8 \}$, then the code $\mathcal C = C_{(\sigma, \hv, T)}$, which is the left kernel of
\[
	H_{(\sigma, \hv, T)}
	=
	\begin{pmatrix}
		\alpha & \alpha^{2} & \alpha^{2^2} & \alpha^{2^3} & \alpha^{2^4} & \alpha^{2^8} & \alpha^{2^9} & \alpha^{2^{10}} & \alpha^{2^{11}} & \alpha^{2^{12}} \\
		\alpha^{2} & \alpha^{2^2} & \alpha^{2^3} & \alpha^{2^4} & \alpha^{2^5} & \alpha^{2^9} & \alpha^{2^{10}} & \alpha^{2^{11}} & \alpha^{2^{12}} & \alpha^{2^{13}} \\
		\alpha^{2^2} & \alpha^{2^3} & \alpha^{2^4} & \alpha^{2^5} & \alpha^{2^6} & \alpha^{2^{10}} & \alpha^{2^{11}} & \alpha^{2^{12}} & \alpha^{2^{13}} & \alpha \\
		\alpha^{2^3} & \alpha^{2^4} & \alpha^{2^5} & \alpha^{2^6} & \alpha^{2^7} & \alpha^{2^{11}} & \alpha^{2^{12}} & \alpha^{2^{13}} & \alpha & \alpha^{2} \\
		\alpha^{2^4} & \alpha^{2^5} & \alpha^{2^6} & \alpha^{2^7} & \alpha^{2^8} & \alpha^{2^{12}} & \alpha^{2^{13}} & \alpha & \alpha^{2} & \alpha^{2^2} \\
		\alpha^{2^5} & \alpha^{2^6} & \alpha^{2^7} & \alpha^{2^8} & \alpha^{2^9} & \alpha^{2^{13}} & \alpha & \alpha^{2} & \alpha^{2^2} & \alpha^{2^3} \\
		\alpha^{2^6} & \alpha^{2^7} & \alpha^{2^8} & \alpha^{2^9} & \alpha^{2^{10}} & \alpha & \alpha^{2} & \alpha^{2^2} & \alpha^{2^3} & \alpha^{2^4} \\
		\alpha^{2^7} & \alpha^{2^8} & \alpha^{2^9} & \alpha^{2^{10}} & \alpha^{2^{11}} & \alpha^{2} & \alpha^{2^2} & \alpha^{2^3} & \alpha^{2^4} & \alpha^{2^5} \\
		\alpha^{2^8} & \alpha^{2^9} & \alpha^{2^{10}} & \alpha^{2^{11}} & \alpha^{2^{12}} & \alpha^{2^2} & \alpha^{2^3} & \alpha^{2^4} & \alpha^{2^5} & \alpha^{2^6} \\
		\alpha^{2^9} & \alpha^{2^{10}} & \alpha^{2^{11}} & \alpha^{2^{12}} & \alpha^{2^{13}} & \alpha^{2^3} & \alpha^{2^4} & \alpha^{2^5} & \alpha^{2^6} & \alpha^{2^7} \\
		\alpha^{2^{10}} & \alpha^{2^{11}} & \alpha^{2^{12}} & \alpha^{2^{13}} & \alpha & \alpha^{2^4} & \alpha^{2^5} & \alpha^{2^6} & \alpha^{2^7} & \alpha^{2^8} \\
		\alpha^{2^{11}} & \alpha^{2^{12}} & \alpha^{2^{13}} & \alpha & \alpha^{2} & \alpha^{2^5} & \alpha^{2^6} & \alpha^{2^7} & \alpha^{2^8} & \alpha^{2^9} \\
		\alpha^{2^{12}} & \alpha^{2^{13}} & \alpha & \alpha^{2} & \alpha^{2^2} & \alpha^{2^6} & \alpha^{2^7} & \alpha^{2^8} & \alpha^{2^9} & \alpha^{2^{10}} \\
		\alpha^{2^{13}} & \alpha & \alpha^{2} & \alpha^{2^2} & \alpha^{2^3} & \alpha^{2^7} & \alpha^{2^8} & \alpha^{2^9} & \alpha^{2^{10}} & \alpha^{2^{11}}
	\end{pmatrix},
\]
has a minimum rank distance of at least $7$, and therefore an error-correcting capacity of at least $3$, by virtue of Theorem \ref{thm:ht}, applied for $b = 0$, $t_1 = 1$, $t_2 = 8$, $\delta = 6$ and $r = 1$. Note that $(t_2, 14) = 2$, which, as required by our theorem, is less than $\delta$.
The set $T$ is also equal, modulo $\sord = 14$, to $8 + \{ 0, 1, 2, 3, 4 \} + 6 \{ 0, 1 \}$, so a bound of $7$ also follows from Theorem \ref{thm:ht} for $b = 8$, $t_2 = 6$ and the other parameters as above; again, $(t_2, 14) = 2$.
This set can also be described as $8 + \{ 0, 1, 2, 3, 4 \} + 3 \{ 0, 2 \}$, which corresponds, by Theorem \ref{thm:roos}, to a lower bound of $7$ for the minimum rank distance of $\mathcal C$, where $b = 8$, $t_1 = 1$, $t_2 = 3$, $\delta = 6$, $r = 1$, $k_0 = 0$ and $k_1 = 2$. Observe that $k_r - k_0 = 2 \le 5 = \delta + r - 2$, as needed for the theorem.
These lower bounds for the minimum rank distance also apply if $\hv$ is not chosen as $\av$. 
Since the length of $\mathcal C$, $n$, is equal to $\sord = 14$, the dimension of $\mathcal C$ can be determined to be $14-10 = 4$, as discussed after Definition \ref{defn:ccode}, and therefore $H_{(\sigma, \hv, T)}$ is a parity check matrix for $\mathcal C$.

There are other options for showing that the minimum rank distance of $\mathcal C$ (and even of some codes of greater dimension, having $\mathcal C$ as a subcode) is at least $7$ through Theorem \ref{thm:roos}.
For example, the parameters $b = 8$, $t_1 = 1$, $t_2 = 5$, $\delta = 6$, $r = 1$ and $\{ k_0, k_1 \} = \{ 0, 4 \}$ also give a lower bound of $7$.
This is also the case for the parameters $b = 0$, $t_1 = 1$, $t_2 = 3$, $\delta = 5$, $r = 2$ and $\{ k_0, k_1, k_2 \} = \{ 0, 3, 5 \}$, which give the set $\{ 0, 1, 2, 3 \} \cup \{ 9, 10, 11, 12 \} \cup \{ 15, 16, 17, 18 \}$, which modulo $\sord = 14$ is the set $T \setminus \{ 8 \}$. This means that the $5$-dimensional code $\mathcal C' = C_{(\sigma, \hv, \{ 0, 1, 2, 3, 4, 9, 10, 11, 12 \})}$ has $\mathcal C$ as a subcode and also has a minimum rank distance of at least $7$.
Using $b = 9$, $t_1 = 1$, $t_2 = 5$, $\delta = 4$, $r = 3$ and $\{ k_0, k_1, k_2, k_3 \} = \{ 0, 1, 3, 4 \}$, the corresponding set is $\{ 9, 10, 11,\allowbreak 14, 15, 16,\allowbreak 24, 25, 26,\allowbreak 29, 30, 31 \}$, which modulo $14$ is $T \setminus \{ 4, 8 \}$. As a result, the $6$-dimensional code $\mathcal C'' = C_{(\sigma, \hv, \{ 0, 1, 2, 3, 9, 10, 11, 12 \})}$ has $\mathcal C'$, and therefore $\mathcal C$, as a subcode and has a minimum rank distance of at least $7$.
The suitability of these parameters for decoding will be discussed later.

A lower bound of at least $7$ cannot be proven for the code through Theorem \ref{thm:ht} or Theorem \ref{thm:roos} if $r$ is forced to be $0$. That is, there is no BCH-like bound of $7$ that applies to the code (see \cite[Corollary 3.4]{GLNN18} for a definition for this bound for skew cyclic codes, such as the one currently being considered). Through BCH-like bounds, it is only possible to get a minimum distance of at least $6$ through the straightforward subset $\{ 0, \dots, 4 \}$ (or the also straightforward $\{ 8, \dots, 12 \}$) of $T$.
In fact, a lower bound of at least $7$ cannot be shown through any bound analogous to the original Hartmann--Tzeng bound as described in \cite{HT72}, that is, through Theorem \ref{thm:ht} if $t_2$ is forced to be relatively prime with $14$, or, equivalently, through Theorem \ref{thm:roos} if $k_0, \dots, k_r$ are restricted to be $0, \dots, r$.

The code $C_{(\sigma, \hv, T)}$ is an MDS code, since its minimum Hamming distance equals $11$, the maximum possible for a $[14, 4]$ code. This follows from the absence of singular $10 \times 10$ submatrices in the $14 \times 10$ matrix $H_{(\sigma, \hv, T)}$, so the minimum Hamming distance is greater than $10$, and the Singleton bound, which prevents its minimum distance from being greater than $11$.
In contrast, the code is not a maximum rank distance (MRD) code for the rank metric with respect to the field extension $F/F^\sigma$. An example of a codeword in $\mathcal C$ of rank weight $10$ is
\begin{equation}\label{eq:examplec}
	\cv = \left(1, a^{4851}, a^{13201}, a^{10}, a^{11714}, a^{5336}, a^{15691}, 0, a^{6387}, a^{7195}, a^{5026}, 0, a^{14643}, 0\right).
\end{equation}
It turns out that $(0,0,0,0,1,0,1,0,0,1,1,0,1,0)\cv^T = 0$. That is, the sum of the fifth, seventh, tenth, eleventh and thirteenth entries of $\cv$ is zero. These entries are
\[
	\begin{aligned}
		a^{11714} & = a^{12} + a^{11} + a^{10} + a^6 + a^3 + 1, \\
		a^{15691} & = a^{13} + a^{12} + a^9 + a^6 + a^5 + a^4 + a^2 + a, \\
		a^{7195} & = a^{12} + a^{10} + a^9 + a^8 + a^7 + a^6 + a^5 + a^3 + 1, \\
		a^{5026} & = a^{13} + a^9 + a^6 + a^4 + 1, \\
		a^{14643} & = a^{12} + a^{11} + a^9 + a^8 + a^7 + a^2 + a + 1.
	\end{aligned}
\]
This, together with the observation that there are $11$ nonzero entries in $\cv$, implies that the rank weight of $\cv$ is at most $10$. It can be computed to be $10$ by checking that this is the dimension of the $F^\sigma$-vector space spanned by its entries.
The minimum rank distance of $\mathcal C$ is, therefore, at least $7$ by the bounds given by the structure of $T$ and either Theorem \ref{thm:ht} or Theorem \ref{thm:roos}, and at most $10$ since $\cv \in \mathcal C$ and $\rkw(\cv) = 10$. This also proves that $\mathcal C$ is not rank equivalent to a generalized Gabidulin code, as these are MRD codes by \cite[Theorem 1]{KG05}.

It will be shown later, in Example \ref{example:subfieldsubcodeasinterleaved}, that $7$ is in fact the minimum rank distance of $\mathcal C$, since it has nonzero codewords whose entries are all in the subfield of $F$ with $2^7$ elements, and such codewords cannot have a rank weight beyond $7$. An example of such a codeword is
\[
	\left(0, 1, a^{3612}, a^{1290}, a^{2709}, a^{3612}, a^{4773}, a^{15222}, a^{16125}, a^{1419}, a^{2193}, a^{4644}, a^{8256}, a^{4902}\right).
\]

\subsection{Applying the assumptions}

As the first step to be able to correct errors of rank weight up to $\tau = \lfloor \frac{7-1}2 \rfloor = 3$ with respect to $\mathcal C = C_{(\sigma, \hv, T)}$, some parameters $b, \delta, t_1, t_2, k_0, \dots, k_r$ such that $\nu = 3$ satisfies Assumption \ref{assumption:Tnusubset}, or any assumption that implies this one as studied in Section \ref{section:ontheassumption}, have to be chosen.
The most straightforward option is the one that resulted in a lower bound of $7$ for the minimum rank distance by Theorem \ref{thm:ht}: $b = 0$, $t_1 = 1$, $t_2 = 8$, $\delta = 6$, $r = 1$, $k_0 = 0$ and $k_1 = 1$. We may also consider the parameters for a lower bound of $7$ by Theorem \ref{thm:roos}: $b = 8$, $t_1 = 1$, $t_2 = 3$, $\delta = 6$, $r = 1$, $k_0 = 0$ and $k_1 = 2$.
Assumption \ref{assumption:sigmat2dr} is true for both sets of parameters, as they respectively satisfy the third and the fourth statement.
Hence, by Lemma \ref{lemma:assumptionsigmat2drimpliesuptotau} and Lemma \ref{lemma:assumptionTnuimpliesassumptionTnusubset}, Assumption \ref{assumption:Tnusubset} is satisfied for $\nu = \tau = 3$ for both sets of parameters.

In other cases, it might not be possible to use Assumption \ref{assumption:sigmat2dr}, so we will briefly comment the possibility of using other assumptions.
For the considered parameters, Assumption \ref{assumption:1rd}, which would otherwise be easier to check than Assumption \ref{assumption:sigmat2dr}, is not satisfied.
For the considered parameters, it is straightforward that the second and, respectively, third statements of Assumption \ref{assumption:Tnu} are satisfied for $\nu = 3$.
If this was not the case, or (equivalently, by Lemma \ref{lemma:assumptionsigmat2drimpliesuptotau}) if Assumption \ref{assumption:sigmat2dr} did not apply, we would have to check if Assumption \ref{assumption:Tnu} applies for some $\nu < \tau$; in such a case, we would get a lower guaranteed error-correcting capacity.
If the capacity given by Assumption \ref{assumption:Tnu} is unsatisfactory, there is the option of directly testing Assumption \ref{assumption:Tnusubset}.
In this example, it is evaluated as follows: firstly, as required, $\nu = 3$ is not greater than than $\delta - 2 = 4$; and the set $T_3$ as defined in \eqref{eq:Tnu}, which is $\{ 3, 4, 11, 12 \}$ for both sets of parameters, can be described, modulo $\sord = 14$, as both the Hartmann--Tzeng-like set $3 + \{ 0, 1 \} + 8\{ 0, 1 \}$ and the Roos-like set $8 + \{ 0, 1 \} + 3\{ 0, 2 \}$.
Following the notation in Assumption \ref{assumption:Tnusubset} explicitly, $b' = 3$, $t_1' = 1$, $t_2' = 8$, $\delta' = 3$, $r' = 1$ gives the Hartmann--Tzeng-like subset, and $b' = 8$, $t_1' = 1$, $t_2' = 3$, $\delta' = 3$, $r' = 1$, $\{ k'_0, k'_1 \} = \{ 0, 2 \}$ results in the Roos-like subset. In both cases, $\delta' + r' = 4 > \nu = 3$ as required.
However, note that designing the code such that Assumption \ref{assumption:sigmat2dr} will be satisfied is more straightforward than designing a code such that both a lower bound applies and Assumption \ref{assumption:Tnusubset} is satisfied for the capacity implied by the bound.

Other parameters for proving a lower bound of $7$ were mentioned earlier.
For the parameters $b = 8$, $t_1 = 1$, $t_2 = 5$, $\delta = 6$, $r = 1$ and $\{ k_0, k_1 \} = \{ 0, 4 \}$, as well as for the parameters for $\mathcal C'$ ($b = 0$, $t_1 = 1$, $t_2 = 3$, $\delta = 5$, $r = 2$, $\{ k_0, k_1, k_2 \} = \{ 0, 3, 5 \}$), Assumption \ref{assumption:sigmat2dr} is not satisfied (equivalently, Assumption \ref{assumption:Tnu} is not satisfied for $\nu = 3$), as $k_r - k_0 = 4$ is not less than $\lfloor \frac{\delta + r} 2 \rfloor = 3$.
However, for the former parameters, it is possible to show that Assumption \ref{assumption:Tnusubset} holds for $\nu = 3$. In fact, the set $T_3 = \{ 3, 4, 11, 12 \}$ (recall \eqref{eq:Tnu}) is the same for these parameters as it was for the parameters considered above, so the same argument can be used to prove that Assumption \ref{assumption:Tnusubset} is satisfied (or, if no subset of $T_3$ proving the Assumption was known, it could be stated that the Assumption is satisfied for these parameters if and only if it is satisfied for any of the former parameters).
For the parameters for $\mathcal C'$, the set $T_3$ takes the form $\{ 3, 4, 12 \}$. This set does not guarantee a minimum distance of at least $4$; for example, the code $C_{(\sigma, (1, a^{15777}, a^{5023}), \{ 3, 4, 12 \})}$ has dimension greater than zero. Therefore, there is no need to attempt to use Assumption \ref{assumption:Tnusubset} in order to show that an error of rank weight $3$ can be corrected with those parameters.
It is more straightforward to show that it is not possible to correct errors of rank weight $\nu = 3$ with respect to the parameters considered for $\mathcal C''$, as in that case $\delta = 4$, so $\nu > \delta - 2$ and the first requirement in Assumption \ref{assumption:Tnusubset} is unsatisfied.

\subsection{Running the algorithms}
\setcounter{MaxMatrixCols}{14}

Consider $\ev = \epv B$ for
\begin{equation}\label{eq:exampleepvB}
	\epv = \left( 1, a, a^{11} \right),
	\quad
	B =
	\begin{pmatrix}
		1 & 1 & 1 & 1 & 1 & 1 & 1 & 1 & 1 & 1 & 1 & 1 & 1 & 1 \\
        0 & 0 & 1 & 1 & 1 & 1 & 1 & 1 & 1 & 1 & 1 & 1 & 1 & 1 \\
        0 & 1 & 1 & 1 & 1 & 1 & 0 & 0 & 0 & 0 & 0 & 0 & 0 & 0
	\end{pmatrix},
\end{equation}
and suppose that $\yv = \cv + \ev$, where $\cv$ is as in \eqref{eq:examplec}, is received. That is, an error of rank weight $\nu = 3$ was added to the codeword by adding $1$ to each coordinate, $a$ to each coordinate except the first two ones, and $a^{11}$ to the coordinates from the second to the sixth. This results in
\begin{multline}\label{eq:exampley}
	\yv = \left(0, a^{4170}, a^{4896}, a^{1514}, a^{2234}, a^{8214}, a^{11585},
	\right. \\ \left.
	a^{12897}, a^{1004}, a^{7930}, a^{15557}, a^{12897}, a^{13045}, a^{12897}\right).
\end{multline}

We will run both Algorithm \ref{alg:span} and Algorithm \ref{alg:locator} for $\yv$ as above, for $\hv, \sigma$ as considered for $\mathcal C$, and for the first parameters that gave a lower bound of $7$ for the minimum distance of $\mathcal C$ through Theorem \ref{thm:roos}: $b = 8$, $t_1 = 1$, $t_2 = 3$, $\delta = 6$, $r = 1$, $k_0 = 0$ and $k_1 = 2$. As discussed earlier, Assumption \ref{assumption:Tnusubset} is satisfied for $\nu = 3$ and these parameters, so both algorithms should be successful.

Since the error vector was designed as $\ev = \epv B$, it is already possible to determine the error span vector $\shv$ and the error locator vector $\ldv$ that the algorithms should compute. By Definition \ref{defn:span}, the error span polynomial is
\[
	\sh = \lclm{z - 1, z - a^{-1} \sigma(a), z - a^{-11} \sigma(a^{11})} = z^3 + a^{1871}z^2 + a^{6165}z + a^{14247} \in F[z;\sigma],
\]
so the error span vector can be expected to be any nonzero $F$-multiple of
\begin{equation}\label{eq:exampleshv}
	\shv = \left( a^{14247}, a^{6165}, a^{1871}, 1 \right).
\end{equation}
The error locator vector $\ldv$ also follows from
\begin{equation}\label{eq:exampleetav}
	\etav = B \hv^T = \left(1, a^5 + a^3, a^{10} + a^9 + a^6 + a^5 + a^3 + a \right) = \left(1, a^{9414}, a^{12430} \right).
\end{equation}
The $1$ as the first entry is a consequence of every entry in the first row of $B$ being $1$ and the sum of the entries of $\hv$ (which is the trace of $\alpha = a^7$ with respect to $F/F^\sigma$) being $1$.
This gives, by Definition \ref{defn:locator}, the error locator polynomial
\[
	\begin{aligned}
		\ld
		& = \lclm{\bar z - 1, \bar z - a^{-9414} \sigma(a^{9414}), \bar z - a^{-12430} \sigma(a^{12430})} \\
		& = \bar z^3 + a^{4634} \bar z^2 + a^{13397} \bar z + a^{11909} \in F[\bar z;\sigma^{-1}]
	\end{aligned}
\]
and the error locator vector
\begin{equation}\label{eq:exampleldv}
	\ldv = \left(a^{11909}, a^{13397}, a^{4634}, 1\right).
\end{equation}
The algorithms, if successful, should get a nonzero $F$-multiple of, respectively, $\shv$ and $\ldv$ at the step of solving the SFSR synthesis problem.
Note that there is no guarantee that they will get exactly $\epv$ and $\etav$ as above, recall that these may be modified as described in Lemma \ref{lemma:epetalc}.

The first step in both algorithms is the computation of the syndromes $\s i j$ for $0 \le i \le \delta - 2, 0 \le j \le r$. This results in
\begin{equation}\label{eq:sijexample}
	\left( \s i j \right)_{\substack{0 \le j \le 1 \\ 0 \le i \le 4}} =
	\begin{pmatrix}
		a^{3109}  & a^{9463} & a^{641}  & a^{14960} & a^{4892} \\
		a^{13234} & a^{4437} & a^{6745} & a^{12053} & a^{16376}
	\end{pmatrix}.
\end{equation}
That is, the $i$-th column and the $j$-th row, starting at zero, of the matrix above has $\s i j$. As they are not all zero, $\yv$ is not a codeword and a nonzero error has been inserted.

In Algorithm \ref{alg:span}, the next step is solving the $\sigma$-SFSR synthesis problem (Problem \ref{problem:SFSRsynthesis} for $\theta = \sigma$) for the two sequences given by the rows above. Theorem \ref{thm:keyspan} guarantees that the shortest $\sigma$-SFSR for those sequences is $\shv$ as in \eqref{eq:exampleshv}, with uniqueness up to nonzero $F$-multiples. Indeed, the result of normalizing the output of, for example, \cite[Algorithm 2]{SJB11} on its last entry (that is, dividing the vector by its last entry so that the last entry is $1$) is $\shv$; note that \cite[Algorithm 2]{SJB11} normalizes its output on its first entry, so it returns $a^{-14247}\shv$, but it does not matter whether the SFSR is normalized in any coordinate or not.
Now we compute the error values as any basis of the kernel of the $F^\sigma$-linear map $\sh(\sigma) = \sum_{k=0}^\nu \sh_k \sigma^k = a^{14247} + a^{6165} \sigma + a^{1871} \sigma^2 + \sigma^3$, as described in Remark \ref{remark:shld2epeta}.
Since $F$ is a finite field and $F^\sigma$ has two elements, this kernel is equal to the roots of the $2$-polynomial, or linearized polynomial, $a^{14247} x + a^{6165} x^2 + a^{1871} x^4 + x^8 \in F[x]$, so this step is equivalent to the computation of a basis for these roots.
Depending on the algorithm used for finding a basis for this space, the result may or may not be $\epv$ as in \eqref{eq:exampleepvB}. For this example, assume that we do get $\epv = \left(1, a, a^{11} \right)$.

Then, we have to compute the error locators. This can be done either by computing
\begin{equation}\label{eq:stijexample}
	\left( \st i j \right)_{\substack{0 \le j \le 1 \\ 0 \le i \le 4}} =
	\begin{pmatrix}
		a^{2380} & a^{7922} & a^{10256} & a^{4999} & a^{3185} \\
		a^{13234} & a^{10410} & a^{5782} & a^{11746} & a^{9215}
	\end{pmatrix}
\end{equation}
following \eqref{eq:st} and then solving for the entries of $\etav$ in the linear system given by Equation \eqref{eq:stee} (or a suitable subsystem with three equations, which would only require the computation of any three consecutive entries in one of the rows in \eqref{eq:stijexample}; see Remark \ref{remark:sfsrredundancy}), or by running the variant of Gabidulin's algorithm described in Section \ref{section:gabidulinsalgorithm} with input $a_k = \ep_k$ and $X_k = \eta_k$ for $1 \le k \le 3$, $b_i = \s i 0$ for $0 \le i \le 2$, $\theta = \sigma$ and $\bar b = 8$.
The latter consists of computing the three non-blank elements below the first row in each one of the following matrices (where the blank entries are undefined and not computed) according to \eqref{eq:Gabidulinakj} and \eqref{eq:Gabidulinbij}:
\[
	\left( A_k^{(j)} \right)_{\substack{1 \le j \le 3 \\ j \le k \le 3}} =
	\begin{pmatrix}
		1 & a & a^{11} \\
		  & a^{6449} & a^{675} \\
		  &  & a^{4433}
	\end{pmatrix},
	\ 
	\left( B_i^{(j)} \right)_{\substack{1 \le j \le 3 \\ 0 \le i \le 3 - j}} =
	\begin{pmatrix}
		a^{3109} & a^{9463} & a^{641} \\
		a^{15998} & a^{15308} & \\
		a^{8211} &  &
	\end{pmatrix},
\]
and then $\eta_3, \eta_2, \eta_1$ as $X_3, X_2, X_1$ in \eqref{eq:GabidulinXk}. Through both paths, the resulting $\etav$ is as in \eqref{eq:exampleetav}.
Now, $B$ is computed as the solution of \eqref{eq:eta}; that is, the matrix of the coordinates of the entries of $\etav$ with respect to $\hv$ and the field $F^\sigma$. This gives $B$ as in \eqref{eq:exampleepvB}. Then, $\ev = \epv B$ is the error that was designed, and $\yv - \ev$ is the codeword $\cv$ described in \eqref{eq:examplec}. Hence, Algorithm \ref{alg:span} has been able to reconstruct the sent codeword in spite of the inserted error of rank weight $\nu = 3$.

For Algorithm \ref{alg:locator}, after computing the syndromes $\s i j$ and checking that they are not all zero and therefore there was an error, the syndromes $\st i j$ are computed, getting the ones in \eqref{eq:stijexample} (or, alternatively, it is possible to compute directly $\st i j$ if Gabidulin's algorithm is to be used, as then $\s i j$ are not needed).
The next step is to solve the $\sigma^{-1}$-SFSR synthesis problem, or Problem \ref{problem:SFSRsynthesis} for $\theta = \sigma^{-1}$, for the sequences given by the two rows in \eqref{eq:stijexample}.
By Theorem \ref{thm:keyspan}, the shortest $\sigma^{-1}$-SFSR for those sequences is $\ldv$ as in \eqref{eq:exampleldv} with uniqueness up to nonzero $F$-multiples, so a nonzero $F$-multiple of $\ldv$ will be the result of running any suitable SFSR synthesis algorithm, such as \cite[Algorithm 2]{SJB11}.
The error locators are now, by Remark \ref{remark:shld2epeta}, any basis of the kernel of the $F^\sigma$-linear map $\ld(\sigma^{-1}) = \sum_{k=0}^\nu \ld_k \sigma^{-k} = a^{11909} + a^{13397} \sigma^{-1} + a^{4634} \sigma^{-2} + \sigma^{-3}$ or, equivalently, any basis of the roots of the $2$-polynomial $a^{11909} + a^{13397} x^{2^{13}} + a^{4634} x^{2^{12}} + x^{2^{11}}$. Note that the kernel of $\ld(\sigma^{-1})$ is the same as the kernel of $\sigma^3 \circ \ld(\sigma^{-1}) = \sigma^3(a^{11909})\sigma^3 + \sigma^3(a^{13397}) \sigma^2 + \sigma^3(a^{4634}) \sigma + \sigma^0 = a^{13357}\sigma^3 + a^{8878} \sigma^2 + a^{4306}\sigma + \sigma^0$, which also has the same kernel as the result of dividing it by $a^{13357}$, which is $\sigma^3 + a^{11904}\sigma^2 + a^{7332}\sigma + a^{3026}$. In particular, the basis can be computed as a basis for the roots of the $2$-polynomial $x^8 + a^{11904} x^4 + a^{7332}x^2 + a^{3026}x$, whose degree is $8$ instead of $2^{13}$ as in the previous $2$-polynomial.
Independently of the approach, the resulting $\etav$ might not necessarily be the one in \eqref{eq:exampleetav}. Assume that the computed basis turns out to be
\begin{equation}\label{eq:exampleetav2}
	\etav = \left(1, a^{10} + a^9 + a^6 + a, a^5 + a^3 \right) = \left(1, a^{16301}, a^{9414} \right).
\end{equation}
That is, instead of $\etav = (\eta_1, \eta_2, \eta_3)$ as in \eqref{eq:exampleetav}, the computed error locators are $\etav = (\eta_1, \eta_3, \eta_2 + \eta_3)$.
It is possible to compute at this moment, as noted in Remark \ref{remark:Bbeforevalues}, the matrix $B$ such that $\etav = B \hv^T$, getting
\[
	B = 
	\begin{pmatrix}
		1 & 1 & 1 & 1 & 1 & 1 & 1 & 1 & 1 & 1 & 1 & 1 & 1 & 1 \\
		0 & 1 & 0 & 0 & 0 & 0 & 1 & 1 & 1 & 1 & 1 & 1 & 1 & 1 \\
		0 & 0 & 1 & 1 & 1 & 1 & 1 & 1 & 1 & 1 & 1 & 1 & 1 & 1
	\end{pmatrix},
\]
whose rows are the first, the third, and the sum of the second and the third of the ones in \eqref{eq:exampleepvB}, as we got exactly that linear combination of $\eta_1, \eta_2, \eta_3$ for the error locators.

The next, or previous, step with respect to computing $B$ is solving for the error values, either by solving the linear system in Equation \eqref{eq:see} (or a subset of these equations, recall Remark \ref{remark:sfsrredundancy}), or by using the variant of Gabidulin's algorithm described in Section \ref{section:gabidulinsalgorithm} with input $a_k = \eta_k$ and $X_k = \ep_k$ for $1 \le k \le 3$, $b_i = \st i 0$ for $0 \le i \le 2$, $\theta = \sigma^{-1}$ and $\bar b = 8$.
We get
\[
	\left( A_k^{(j)} \right)_{\substack{1 \le j \le 3 \\ j \le k \le 3}} =
	\begin{pmatrix}
		1 & a^{16301} & a^{9414} \\
		  & a^{10332} & a^{11189} \\
		  &  & a^{15128}
	\end{pmatrix},
	\left( B_i^{(j)} \right)_{\substack{1 \le j \le 3 \\ 0 \le i \le 3 - j}} =
	\begin{pmatrix}
		a^{2380} & a^{7922} & a^{10256} \\
		a^{554} & a^{12832} &  \\
		a^{7532} &  &
	\end{pmatrix}
\]
when running this algorithm. Both approaches lead to
\begin{equation}\label{eq:exampleepv2}
	\epv = \left( 1, a^{11}, a^{11} + a \right) = \left( 1, a^{11}, a^{5001} \right),
\end{equation}
which does not match the one computed in \eqref{eq:exampleepvB} for the same reason that $\etav$ is not the one in \eqref{eq:exampleetav}. Anyway, the designed error $\ev$ is $\epv B$, and $\yv - \ev$ is the codeword $\cv$ defined in \eqref{eq:examplec}. Algorithm \ref{alg:locator}, therefore, also retrieves the sent codeword after adding an error of rank weight $\nu = 3$.

Note that, if the same error vector $\ev$ is added to any other codeword $\cv \in \mathcal C$, including the zero codeword, the computed syndromes $\s i j$ are the same. As a consequence, both algorithms would perform the exact same operations after computing the syndromes until the final step of returning $\yv - \ev$.

Additionally, replacing any of $k_0, \dots, k_r$ with other number which is the same modulo $\sord = 14$, or multiplying all of them by any $s$ relatively prime with $14$ and $t_2$ by the inverse of $s$ modulo $14$, has no effect on the computations done by the algorithms. Hence, the algorithms perform the exact same computations for the alternative parameters that gave a lower bound of $7$ through Theorem \ref{thm:ht} ($t_2 = 6$, $k_1 = 1$ and the other parameters as in the example).
Furthermore, permutating the order of the entries $k_0, \dots, k_r$ of the algorithm only has the effect of applying the same permutation to the order of the sequences and the order of the equations to be solved in \eqref{eq:see} or \eqref{eq:stee}, and therefore the algorithm performs essentially the same computations after such a rearrangement.
As a result, if the algorithms are run with the first parameters that gave a lower bound of $7$ through Theorem \ref{thm:ht} ($b = 0$, $t_1 = 1$, $t_2 = 8$, $\delta = 6$, $r = 1$, $k_0 = 0$ and $k_1 = 1$) or with the alternative parameters that gave a bound of $7$ through Theorem \ref{thm:roos} which also lead to satisfying Assumption \ref{assumption:Tnusubset} ($b = 8$, $t_1 = 1$, $t_2 = 5$, $\delta = 6$, $r = 1$, $k_0 = 0$ and $k_1 = 4$), the algorithms will do the same tasks, simply swapping $\s i 0$ with $\s i 1$ and $\st i 0$ with $\st i 1$ for each $i$.

\subsection{Forcing decoding failures}

In order to show how the algorithms handle decoding failures, some decoding failures can be forced by modifying their input and parameters in such a way that Assumption \ref{assumption:Tnusubset} is not true for $\nu$ the rank weight of the inserted error and the chosen parameters. While the most straightforward way to do so would be to repeat the example above with some $\ev$ of rank weight $4$, this would require doing some new computations. This will be done later, in Example \ref{example:interleavederrorcorrection}. Here, the decoding failures will be forced by modifying the parameters used for decoding, so that they would still work for $\mathcal C$, but not up to a decoding capacity of $3$. This approach will actually be required in order to be able to get the decoding failures that occur later in the algorithms, as they require a code of length less than $\sord = 14$ or a code not equal to $C_{(\sigma, \hv, T)}$.

Fistly, we will use the same parameters as in the working example except $r = 0$ and $k_0 = 0$. These parameters satisfy all the preconditions except the requirement on Assumption \ref{assumption:Tnusubset} for $\nu = 3$, so the algorithm may not be able to correct the previously considered $\yv$ into $\cv$. They correspond to decoding with respect to the Gabidulin code $C_{(\sigma, \hv, \{ 8, 9, 10, 11, 12 \})} = C_{(\sigma, \sigma^8(\hv), \{ 0, 1, 2, 3, 4 \})}$, which is a subcode of $\mathcal C$.
When running Algorithm \ref{alg:span} and Algorithm \ref{alg:locator} with these parameters and $\yv$, they will only compute the first row of syndromes in \eqref{eq:sijexample}. In each case, while the shortest SFSR which generates the sequence given by this row still has length $3$, it is not unique, and we may not get $\shv$. If the shortest SFSR is computed by using \cite[Algorithm 2]{SJB11} skipping the check performed by the algorithm to ensure that the problem does have a single solution, instead of $\shv$, the output is $\shv^* = a^{-2502}\left(a^{2502}, a^{15390}, a^{313}, 1\right)$.
The kernel of the corresponding function $\sh^*(\sigma)$ has dimension $1$, being generated by $a^{8977}$, so the algorithm fails to compute as much error values as the degree of the error span polynomial (or the length of $\shv^*$ minus $1$) and has to return a decoding failure.
For Algorithm \ref{alg:locator}, only the first row in \eqref{eq:stijexample} is computed and the shortest $\sigma^{-1}$-SFSR which generates this row is again also of length $3$ but not unique. Algorithm 2 in \cite{SJB11} computes $\ldv^* = a^{-4802}\left(a^{4802}, a^{10732}, a^{773}, 1\right)$, and the kernel of $\ld^*(\sigma^{-1})$ has dimension $1$, being generated by $a^{15003}$, returning a decoding failure again.
If $r$ is again $0$ and $k_0$ is chosen as $2$, then the algorithms will compute, and solve the SFSR synthesis problem for, the second row in \eqref{eq:sijexample} and \eqref{eq:stijexample}. The returned failure is essentially the same, except this time the kernel has dimension zero in both cases, that is, the functions are actually invertible.

Another decoding failure is bound to happen by using the parameters in the working example except $\delta = 3$. This means that the $\sigma$-SFSR synthesis problem considered in Algorithm \ref{alg:span} will involve the first two elements in each row in \eqref{eq:sijexample}.
There is no $\sigma$-SFSR of length $1$ for the two sequences. As the length of these sequences is $2$, every $\sigma$-SFSR of length $2$ will generate both sequences. Depending on the SFSR synthesis algorithm, we may get some $\shv^*$ whose corresponding function $\shv^*(\sigma)$ either has a kernel of dimension less than $2$, getting the same decoding failure as above (for example, if the algorithm computes $(1, 0, \dots, 0)$ every time it is a solution; note that this cannot be normalized on its last entry), or has a kernel of $F^\sigma$-dimension $2$, skipping that failure. The latter is the case, for example, for $\shv^* = (a^2 + a, a^2 + a + 1, 1)$, which has both $1$ and $a$ in its kernel, since $\lclm{z - 1, z - a^{-1}\sigma(a)} = z^2 + (a^2 + a + 1)z + a^2 + a$. For this $\shv^*$, the algorithm will get some $\epv$ such as $(1, a)$ as a basis of the kernel and therefore as error values, and then will fail at the stage of computing the error locators, as there will be no solutions for the \eqref{eq:stee} for the syndromes and the computed $\epv$.
The algorithm does go on if the solutions for \eqref{eq:stee} are computed using the variant of Gabidulin's algorithm described in Section \ref{section:gabidulinsalgorithm}, later computing some $\ev$ of rank weight $2$, and therefore returning some $\cv$ which is not a codeword, getting the last decoding failure. Note that, if Gabidulin's algorithm is run with the second sequence instead of the first, the computed error locators will be different; this implies that there are no solutions for the full equations in \eqref{eq:stee}.

As shown in \eqref{eq:exampley}, the first entry of $\yv$ is zero, while the first entry of both $\cv$ and $\ev$ is nonzero. In particular, the result of removing the first entry of $\cv$ is not a codeword in the result of shortening $\mathcal C$ at its first coordinate.
Therefore, if either Algorithm \ref{alg:span} or Algorithm \ref{alg:locator} is run with the input in the working example after removing the first entry of both $\yv$ and $\hv$, both sets of syndromes, \eqref{eq:sijexample} and \eqref{eq:stijexample}, are the same, so the same $\epv$ and $\etav$ will be computed. However, the algorithm will fail at the step of computing some $B$ such that $\etav = B \hv^T$, returning the third possible decoding failure, as the $F^\sigma$-span of the entries of $\hv$ does not contain the first entry of $\etav$; this corresponds to the first column in the matrix $B$ in \eqref{eq:exampleepvB} having a nonzero entry in its first row.

As shown earlier, the last possible decoding failure might be raised if the second one is skipped by using Gabidulin's algorithm (as described in Section \ref{section:gabidulinsalgorithm}) for solving \eqref{eq:see}.
This failure would also be returned if $\mathcal C$ is a proper subset of $C_{(\sigma, \hv, T)}$, where $T$ is the set implied by the used parameters.
For example, if $\mathcal C$ is defined as the result of removing $\cv$ from $C_{(\sigma, \hv, T)}$ for the $T$ considered in the working example, the working example would end in a decoding failure, as $\cv$ is not a codeword in this $\mathcal C$.

\section{Increasing the length of the codes}
\label{section:length}

The length of the codes considered so far, which follow Definition \ref{defn:ccode}, is bounded by the order of the considered field automorphism, which is itself bounded as a consequence of the properties of the field.
For example, if $F$ is a finite field with $p^m$ elements for some prime $p$ and some $m \ge 1$, then any automorphism $\sigma : F \to F$ is a power of, and therefore its order divides the one of, the Frobenius endomorphism, whose order is $m$.
However, for some applications, particularly in the context of the lifting construction for random linear network coding (see \cite[Section IV.A and Section VI.E]{SKK08}), a code length significantly higher than the order of the field automorphism is desirable.

\subsection{Subfield subcodes and interleaved codes}

There are two immediate options for getting codes of arbitrary length over a given field $F$: interleaving codes over $F$ until the resulting length is large enough, or considering the subfield subcode with respect to $F$ of a code over a larger field $L$ that can support the desired length. We shall denote by $[F_1:F_2]$ the $F_2$-dimension of $F_1$; that is, the degree of the field extension $F_1/F_2$.

If $\mathcal C_1, \dots, \mathcal C_\ell$ are linear codes over the same field $F$, then the interleaved code $\mathcal C_1 \times \dots \times \mathcal C_\ell$ has the concatenation of codewords in $\mathcal C_1, \dots, \mathcal C_\ell$ as its codewords. This code has as its length and dimension the sum of, respectively, the lengths and the dimensions of the considered codes, while its (Hamming or rank) minimum distance is the minimum of the minimum distances of the codes. The codes might be chosen to be identical: an $[n,k,d]$ code $\mathcal C$ over the field $F$ (where $d$ is with respect to either the Hamming metric or the rank metric) interleaved with itself $\ell$ times results in an $[\ell n, \ell k, d]$ code over $F$. By $\mathcal C^{\times \ell}$, we will denote the result of interleaving the code $\mathcal C$ with itself $\ell$ times.

\begin{example}\label{example:perblockrank}
	The immediate approach for nearest-neighbor decoding with respect to the Hamming metric for interleaved codes is to independently decode each one of the $\ell$ blocks. However, as we shall illustrate in the following example, this is not well suited for the rank metric.
	From a field $F$ with $a,b \in F$ such that $\{1, a, b\}$ are $K$-linearly independent for some $K$ subfield of $F$, consider the rank metric with respect to the extension $F/K$ and the $[3, 1, 3]$ (therefore MRD) code $\mathcal C = F(1, a, b)$.
	By subtracting the $1$-weight vectors $(0, 0, b)$, $(0, a, 0)$ and $(1, 0, 0)$ to $(1, a, b)$, the respective resulting vectors $(1, a, 0)$, $(1, 0, b)$ and $(0, a, b)$ are shown to be at a distance $1$ to $\mathcal C$, having $(1, a, b)$ as the only codeword at distance $1$, since $1$ is the error-correcting capacity of $\mathcal C$.
	As a result, for the vector $\yv = (1, a, 0, 1, 0, b, 0, a, b)$ and the code $\mathcal C^{\times 3}$, a per-block nearest-neighbor error correction algorithm would retrieve $(1, a, b, 1, a, b, 1, a, b)$, which is at a rank distance $3$, while there are codewords in $\mathcal C^{\times 3}$ at a distance $2$ to $\yv$, such as $\yv + (-1, -a, 0, 0, a, 0, 1, 0,0) = (0, 0, 0, 1, a, b, 1, a, b)$.
\end{example}

\begin{remark}
	When considering rank-metric block codes over finite fields, if the length is not assumed to be constrained by the degree of the extension, the Singleton bound is usually defined (see e.g. \cite[Section II.B]{GY06} or \cite[Section II.C]{SKK08}) in such a way that, for an $[n, k]$ $F$-linear code $\mathcal C$ and the rank metric with respect to $F/K$, once $\ell$ is such that $\ell n \ge [F:K]$, the minimum distance given by the bound for $[\ell n, \ell k]$ codes over the field $F$ (such as $\mathcal C^{\times \ell}$) does not depend on $\ell$. Since the concept of MRD codes is defined with respect to this bound, $d$ does not get indefinitely further from the minimum distance that an MRD $[\ell n, \ell k]$ code over $F$ would have as $\ell$ increases. In particular, if $n = [F:K]$ and $d = n - k + 1$, the resulting $[\ell n, \ell k, d]$ code will be an MRD code for any $\ell \ge 1$; this is noted in \cite[Section VI.E]{SKK08}.
\end{remark}

For the alternative approach, we may consider an extension in the following sense:

\begin{defn}[{\cite[Definition 2.2]{GLNN18}}]\label{defn:extension}
	A field automorphism $\theta : L \to L$ is said to be an extension of the field automorphism $\sigma : F \to F$ of degree $\ell$ if $L/F$ is a field extension of degree $\ell$, $L^\theta = F^\sigma$ and the restriction of $\theta$ to $F$ is $\sigma$.
\end{defn}

\begin{remark}
	Both $L/L^\theta$ and $F/F^\sigma = F/L^\theta$ are Galois extensions. By elementary facts on Galois extensions, so is $L/F$.
	Consequently, $[L:F]$ is equal to $\thord / \sord$ and the condition on the degree of the extension being $\ell$ may be replaced with $\thord = \ell \sord$ in the definition above; in fact, this is how this concept is stated in \cite[Definition 2.2]{GLNN18}. Furthermore, $F = L^{\theta^{\sord}}$: it is immediate that $F \subseteq L^{\theta^{\sord}}$, and $\left[L:L^{\theta^{\sord}}\right] = \thord / \sord = [L:F]$.
\end{remark}

See \cite[Section 2]{GLNN18} for details about such extensions; for example, as shown in \cite[Example 2.4]{GLNN18}, extensions of a prescribed degree $\ell$ always exist for any $\sigma : F \to F$ in a finite field $F$. Unlike \cite[Section 2]{GLNN18}, we shall refrain from assigning the letters $n$ or $\mu$ to $\thord$ and $\sord$ respectively, since $n$ and $\mu$ will refer to the lengths of some codes and we are not disregarding the cases where $n < \thord$ or $\mu < \sord$.

The following is immediate, since $L^\theta = F^\sigma$.
\begin{lemma}\label{lemma:ranksubfield}
	The rank metric with respect to $L/L^\theta$, when restricted to elements with entries in the subfield $F$, matches the rank metric with respect to $F/F^\sigma$.
\end{lemma}

Hence, we can construct, from the target field $F$ with its automorphism $\sigma$, some $L,\theta$ according to Definition \ref{defn:extension} that can support the desired length $n$; that is, such that $\thord = \ell \sord$ is at least $n$.
One might also first consider a field $L$ with an automorphism $\theta$ whose order $\thord$ is at least the desired code length $n$ and then choose $F$ as $L^{\theta^m}$ for some suitable $m$ dividing $\thord$ such that the resulting size of $F$ is considered acceptable. Once the extension $L/F$ is constructed, one can get an $L$-linear code $\overline{\mathcal C}$ of length $n$ and then the subfield subcode with respect to the subfield $F$, $\mathcal C = \overline{\mathcal C} \cap F^n$. For our purposes, we shall consider codes of the form $\overline{\mathcal C} = C_{(\theta, \hv, T)}$, as in Definition \ref{defn:ccode}, by choosing suitable $\hv \in L^n, T \subset \ZZ$. Note that the $F$-dimension of $\mathcal C$ might, or might not, be less than the $L$-dimension of $\overline{\mathcal C}$; this is predictable if $\overline{\mathcal C}$ follows Definition \ref{defn:ccode} as we shall describe later in Remark \ref{remark:overlineT}.

Decoding over the subfield subcode $\mathcal C$ can be done by any decoding algorithm for $\overline{\mathcal C}$, since $F^n$ inherits the rank metric from $L^n$ as noted in Lemma \ref{lemma:ranksubfield}. When $\overline{\mathcal C} = C_{(\theta, \hv, T)}$, we may use any of either Algorithm \ref{alg:span} or Algorithm \ref{alg:locator}, in the context of the field $L$ and the automorphism $\theta$ instead of $F$ and $\sigma$, under the conditions studied in previous sections.
Despite the fact that the input is some $\yv \in F^n$, the algorithms will still work in the larger field $L$ in most circumstances, since at least $n - \sord$ entries of $\hv$ will not be in $F$, and choosing a length $n$ such that $n \le \sord$ would defeat our current motivation for considering a field $L$ larger than $F$.

The following result shows that, if the considered codes follow Definition \ref{defn:ccode}, whenever a suitable field extension $L/F$ with corresponding $\theta, \sigma$ exists, the previously discussed approach of interleaving a code with itself is equivalent to an instance of our current approach through subfield subcodes with the same set $T$.

\begin{prop}\label{prop:interleavedassubfieldsubcode}
	Let $\theta : L \to L$ be an extension of $\sigma : F \to F$ of degree $\ell$ as in Definition \ref{defn:extension}.
	Consider any $\{ b_1, \dots, b_\ell \}$ $F$-basis of $L$ and any vector $\hv' \in F^\mu$ of $F^\sigma$-linearly independent entries, and define $\hv = (b_1 \hv' | b_2 \hv' | \dots | b_\ell \hv')$.
	Then, for any finite $T \subset \ZZ$, the $F$-linear code $C_{(\sigma, \hv', T)}^{\times \ell}$ equals the subfield subcode $C_{(\theta, \hv, T)} \cap F^{\ell \mu}$.
\end{prop}
\begin{proof}
	Let $C_j$, $C_j'$ be the $j$-th column of, respectively, $H_{(\theta, \hv, T)}$ and $H_{(\sigma, \hv', T)}$. Then, $C_j$ is the vertical concatenation of $\theta^{T_j}(b_1)C_j', \dots, \theta^{T_j}(b_\ell)C_j'$ for some element $T_j$ in $T$. Since $\{ \theta^{T_j}(b_1), \dots, \theta^{T_j}(b_\ell) \}$ is an $F$-basis of $L$, it follows that, for any $\cv_1, \dots, \cv_\ell \in F^\mu$, $\cv_l C_j' = 0$ for all $1 \le l \le \ell$ if and only if $(\cv_1 | \dots | \cv_\ell) C_j$, which equals $(\theta^{T_j}(b_1) \cv_1 + \dots + \theta^{T_j}(b_\ell) \cv_\ell) C_j'$, is $0$. Applying this fact to all the columns in the aforementioned matrices, the product $\cv_l H_{(\sigma, \hv', T)}$ is the zero vector in $F^\mu$ for all $1 \le l \le \ell$ if and only if $(\cv_1 | \dots | \cv_\ell) H_{(\theta, \hv, T)}$ is the zero vector in $L^{\ell \mu}$. Therefore:
	\begin{align*}
		C_{(\sigma, \hv', T)}^{\times \ell}
		& = \{ (\cv_1 | \dots | \cv_\ell) ~|~ \cv_1, \dots, \cv_\ell \in C_{(\sigma, \hv', T)} \subseteq F^\mu \} \\
		& = \{ (\cv_1 | \dots | \cv_\ell) \in F^{\ell \mu} ~|~ \cv_l H_{(\sigma, \hv', T)} = \zerov \text{ for all } 1 \le l \le \ell \} \\
		& = \{ \cv \in F^{\ell \mu} ~|~ \cv H_{(\theta, \hv, T)} = \zerov \} \\
		& = C_{(\theta, \hv, T)} \cap F^{\ell \mu}. \qedhere
	\end{align*}
\end{proof}

On the other hand, the approach through subfield subcodes yields codes which are, up to possible shortening, $F^\sigma$-linearly rank equivalent (recall Definition \ref{defn:fsigmalinearrankequivalence}) to a particular construction through interleaved codes.

\begin{prop}\label{prop:subfieldsubcodeequivalenttointerleaved}
	Consider any code of the form $\mathcal C = \overline{\mathcal C} \cap F^n$, where $\overline{\mathcal C} = C_{(\theta, \hv, T)}$ where $\theta : L \to L$ is an extension of $\sigma : F \to F$ of degree $\ell$ as in Definition \ref{defn:extension}, and $\hv \in L^n$ has $L^\theta$-linearly independent entries.
	Consider any $\bv' = (a_1, a_2, \dots, a_\sord) \in F^\sord$ whose entries constitute an $F^\sigma$-basis of $F$.
	Then, $\mathcal C$ is a $(\thord - n)$-times shortened code from a code which is $F^\sigma$-linearly rank equivalent to $C_{(\sigma, \bv', T)}^{\times \ell}$, where $\sigma$ is the restriction of $\theta$ to $F$.
\end{prop}
\begin{proof}
	Consider any $F$-basis $\{ b_1, \dots, b_\ell \}$ of $L$. Then, $\{ a_i b_l ~|~ 1 \le i \le \sord, 1 \le l \le \ell \}$ is an $L^\theta$-basis of $L$.
Let us define $\bv$ as the vector $(b_1 \bv' | b_2 \bv' | \dots | b_\ell \bv')$. As the entries of $\hv \in L^n$ are $L^\theta$-linearly independent, there exists an $n$-rank matrix $P \in M_{n \times \thord}\left(L^\theta\right)$ such that $\hv^T = P \bv^T$.
Following part of the argument shown in the proof of Lemma \ref{lemma:rankequivalence}, the code $\overline{\mathcal C} = C_{(\theta, \hv, T)}$ is the left kernel of the matrix $H_{(\theta, \hv, T)} = P H_{(\theta, \bv, T)}$, defined according to \eqref{eq:H}.

Now define $\mathcal D = C_{(\theta, \bv, T)} \cap F^\thord$.
The matrix $P$ can be completed into an invertible matrix $\tilde P \in M_{\thord \times \thord}\left(L^\theta\right)$ by adding $\thord - n$ linearly independent rows. For any such matrix $\tilde P$, $\mathcal C = \{ \cv \in F^n ~|~ cP \in \mathcal D \}$ is a $(\thord - n)$-times shortened code from the code $\mathcal D \tilde P^{-1} = \{ \cv \in F^\thord ~|~ c \tilde P \in \mathcal D \}$, which has the same length, dimension and rank weight distribution as $\mathcal D$, which by Proposition \ref{prop:interleavedassubfieldsubcode} is $C_{(\sigma, \bv', T)}^{\times \ell}$.
\end{proof}

We can control the dimension of the resulting subfield subcode through the following result.

\begin{lemma}\label{lemma:subcodedimension}
	Let $\theta : L \to L$ be an extension of degree $\ell$ of $\sigma : F \to F$ as in Definition \ref{defn:extension}.
	If $n = \thord$, the $F$-dimension of the code $\mathcal C = C_{(\theta, \hv, T)} \cap F^n$ is $\ell(\sord - \kappa)$, where $\kappa$ denotes the number of elements in $T$ modulo $\sord$.
\end{lemma}
\begin{proof}
	By Proposition \ref{prop:subfieldsubcodeequivalenttointerleaved} in our case where $n = \thord$, $\mathcal C$ is $F^\sigma$-linearly rank equivalent to $C_{(\sigma, \bv', T)}^{\times \ell}$. The result follows from the dimension of $C_{(\sigma, \bv', T)}$ being $\sord - \kappa$ as noted right after Definition \ref{defn:ccode}.
\end{proof}

	For the general case where $n$ could be less than $\thord$, the dimension of $\mathcal C$ is reduced by at most $\thord - n$ with respect to the case $n = \thord$ considered in the lemma.

\begin{corollary}\label{cor:Textension}
	Let $\theta : L \to L$ be an extension of $\sigma : F \to F$ as in Definition \ref{defn:extension}.
	Consider any $T_1, T_2 \subset \ZZ$ such that $T_1 + \sord \ZZ = T_2 + \sord \ZZ$; that is, $T_1$ and $T_2$ have the same entries modulo $\sord$.
	Then, for any $\hv \in L^n$ of $L^\theta$-linearly independent entries, the codes $C_{(\theta, \hv, T_1)} \cap F^n$ and $C_{(\theta, \hv, T_2)} \cap F^n$ are equal.
\end{corollary}
\begin{proof}
	If $n = \thord$, this follows from the code $C_{(\theta, \hv, T_1 \cup T_2)}$ being a subset of both codes, as the dimensions of the three codes match by the previous result.
	In general, consider any $\hv' \in L^\thord$ with $L^\theta$-linearly independent entries such that its first $n$ entries are the ones in $\hv$. The considered codes are shortened codes, on the last $\thord - n$ coordinates, from respectively $C_{(\theta, \hv', T_1)} \cap F^n$ and $C_{(\theta, \hv', T_2)} \cap F^n$. These codes are equal as their length is $\thord$, so the respective shortened codes on the same coordinate set are also equal.
\end{proof}

\begin{remark}\label{remark:overlineT}
	In particular, when considering subfield subcodes (or, by Proposition \ref{prop:interleavedassubfieldsubcode}, interleaved codes where the extension exists), instead of a given set $T$ we may consider the set $T' \subset \{ 0, \dots, \sord - 1 \}$ of elements in $T$ modulo $\sord$, or the set $\overline T = T' + \sord \{ 0, \dots, \ell - 1 \} \subseteq \{ 0, \dots, \thord - 1 \}$.
	If $n = \thord$, by Lemma \ref{lemma:subcodedimension} the $F$-dimension of $C_{(\theta, \hv, T)} \cap F^n$ is equal to $n$ minus the number of elements in $\overline T$, which is also the $L$-dimension of $C_{(\theta, \hv, T)}$; in fact, the $L$-dimension of $C_{(\theta, \hv, T)}$ is kept as the $F$-dimension of $C_{(\theta, \hv, T)} \cap F^n$ if and only if $T$ is, up to modulo $\thord$, of the form given by $\overline T$.
	In \cite{GLNN18}, which deals with skew codes (also under $n = \thord$), the set $\overline T$ is considered so that the coefficients of the skew polynomial that generates the code as a left ideal are in the subfield.
\end{remark}

\begin{example}\label{example:interleavedasasubfieldsubcode}
	Consider the code $\mathcal C$ introduced in Section \ref{section:example}. This code is $\mathcal C = C_{(\sigma, \hv, T)}$ for $\sigma$ the Frobenius endomorphism in $F = \FF_{2^{14}} = \FF_2(a)$ where $a \in F$ is such that $a^{14} + a^7 + a^5 + a^3 + 1 = 0$, $\hv = \left( a^7, \sigma(a^7), \dots, \sigma^{\sord-1}(a^7) \right)$, and $T = \{ 0, 1, 2, 3, 4,\allowbreak 8, 9, 10, 11, 12 \}$.
	The result of interleaving $\mathcal C$ with itself, $\mathcal C^{\times 2}$, is the set of possible results of concatenating two codewords in $\mathcal C$. Since $\mathcal C$ was shown in Section \ref{section:example} to be a $[14, 4, 7]$ rank code, $\mathcal C^{\times 2}$ is a $[28, 8, 7]$ rank code.
	
	The automorphism $\sigma$ admits an extension $\theta : L \to L$ for $L = \FF_{2^{28}} = F(u)$ as in Definition \ref{defn:extension} for some $u \in L$, so $\{ 1, u \}$ is an $F$-basis of $L$.
	By Proposition \ref{prop:interleavedassubfieldsubcode}, the $F$-linear subfield subcode $C_{(\theta, \bar \hv, T)} \cap F^{28}$, where $\bar \hv = (\hv | u \hv) = \left( a^7, \sigma(a^7), \dots, \sigma^{\sord-1}(a^7), ua^7, \sigma(a^7), \dots, u\sigma^{\sord-1}(a^7) \right)$, equals $\mathcal C^{\times 2}$.
	By Corollary \ref{cor:Textension}, $C_{(\theta, \bar \hv, T)} \cap F^{28}$ is also equal to $C_{(\theta, \bar \hv, T + \{ 0, 14 \})} \cap F^{28}$.
	Stated otherwise, we may describe $\mathcal C^{\times 2}$ as the left kernel in $F^{28}$ of either $H_{(\theta, \bar \hv, T)}$, which is a $28 \times 10$ matrix with entries in $L$, or $H_{(\theta, \bar \hv, T + \{ 0, 14 \})}$, which is a $28 \times 20$ matrix, again with entries in $L$.
\end{example}

\begin{example}\label{example:subfieldsubcodeasinterleaved}
	By Lemma \ref{lemma:subcodedimension}, there is only one nontrivial subfield subcode of the code $\mathcal C = C_{(\sigma, \hv, T)}$ from Section \ref{section:example}. As $T = \{ 0, 1, 2, 3, 4,\allowbreak 8, 9, 10, 11, 12 \}$ has $4$ elements modulo $4$ and $2$ elements modulo $2$, the subfield subcode of $\mathcal C$ with respect to the subfields of $F = \FF_{2^{14}}$ with $4$ and $2$ elements both have dimension $0$.
	The subfield subcode of $\mathcal C$ with respect to the subfield $K$ of $F$ of order $2^7$, $\mathcal C \cap K^{14}$, has dimension $2(7 - 6) = 2$ by Lemma \ref{lemma:subcodedimension}, as $\ell = 2$ is the degree of the extension $F/K$, the order of the restriction of $\sigma$ to $K$ is $7$, and $T$ has $\kappa = 6$ distinct elements modulo $7$.
	Since the minimum rank distance of $\mathcal C$, which is at least the one of $\mathcal C \cap K^{14}$, was shown in Section \ref{section:example} to be at least $7$ by both Theorem \ref{thm:ht} and Theorem \ref{thm:roos} and there cannot be vectors with entries in $K$ of rank weight beyond $7$ as $7$ is the degree of the extension $K/F^\sigma$, we not only conclude that $\mathcal C \cap K^{14}$ is a $[14, 2, 7]$ rank code: furthermore, it is a code whose $|K|^2 - 1 = 16\,383$ nonzero codewords all have rank $7$. This also proves that the minimum rank distance of $\mathcal C$ is exactly $7$.

	Note that the length of $\mathcal C$ equals $\sord$; otherwise its subfield subcode with respect to $K$ would have dimension at most $2$. By Remark \ref{remark:overlineT}, $\mathcal C \cap K^{14}$ is equal to both $C_{(\sigma, \hv, \{ 0, \dots, 5 \})} \cap K^{14}$ and $C_{(\sigma, \hv, \{ 0, \dots, 5 \} \cup \{ 7, \dots, 12 \})} \cap K^{14}$.
	Proposition \ref{prop:subfieldsubcodeequivalenttointerleaved} shows that this code is $F^\sigma$-linearly rank equivalent to any code of the form $C_{(\rho, \hv', T')}^{\times 2}$, where $\hv' \in K^7$ is any $F^\sigma$-basis of $K$, $\rho$ is the restriction of $\sigma$ to $K$, and $T'$ can be $T$, $\{ 0, \dots, 5 \}$, or $\{ 0, \dots, 5 \} \cup \{ 7, \dots, 12 \}$. This also implies that all nonzero codewords of such codes have rank weight $7$, which is to be expected as they would all have a BCH-like lower bound of $7$ and each one of the two components has length $7$.
\end{example}

\subsection{Decoding over interleaved codes}

As previously noted, decoding over the subfield subcode of a given code can be done through decoding in the parent code. This also results in the possibility of decoding interleaved codes constructed from interleaving a code following Definition \ref{defn:ccode} with itself, since by Proposition \ref{prop:interleavedassubfieldsubcode} they can be described as some subfield subcode.
As we shall show, this decoding procedure can be rearranged so that the field extension $L/F$ is not explicitly constructed, and as a result the algorithm works within $F$. The result of this essentially extends \cite[Section VIII]{SJB11}, since such an extension always exists for finite fields, as shown in \cite[Example 2.4]{GLNN18}.

We therefore consider the $\mu$-length $F$-linear code $C_{(\sigma, \hv', T)}$ for some $\sigma : F \to F$, some $\hv' \in F^\mu$ of $F^\sigma$-linearly independent entries (so $\mu \le \sord$) and some $T \subset \ZZ$, as well as its corresponding interleaved code $C_{(\sigma, \hv', T)}^{\times \ell}$ of length $n = \ell \mu$.
This interleaved structure considers each $\vv \in F^n$ as the concatenation of some $\vv_1, \dots, \vv_\ell \in F^\mu$.
Consequently, for any decomposition $\yv = \cv + \ev$ where $\yv, \ev \in F^n$ and $\cv \in C_{(\sigma, \hv', T)}^{\times \ell}$, we also get $\yv_l = \cv_l + \ev_l$, where $\cv_l \in C_{(\sigma, \hv', T)}$, for each block $1 \le l \le \ell$.

We assume the existence of an extension $\theta : L \to L$ of $\sigma$ of degree $\ell$. In addition, we take an $F$-basis of $L$ $\{ b_1, \dots, b_\ell \}$ and consider $\hv = (b_1 \hv' | b_2 \hv' | \dots | b_\ell \hv')$. By Proposition \ref{prop:interleavedassubfieldsubcode}, $C_{(\sigma, \hv', T)}^{\times \ell} = C_{(\theta, \hv, T)} \cap F^n$. 
The $\hv$-defining set of this code, by Remark \ref{remark:overlineT}, is $T + \sord \ZZ$.
For any $d \in \ZZ$, $\{ \theta^d(b_1), \dots, \theta^d(b_\ell) \}$ is an $F$-basis of $L$, so the $d$-th syndrome of $\yv$ with respect to $\hv$, $S_d = \yv \theta^d(\hv)^T$, is decomposed as
\begin{equation}\label{eq:sdlbasis}
	S_d = \sum_{l=1}^\ell \theta^d(b_l) S_d^l,
\end{equation}
and this decomposition is unique in the sense that the $\ell$ coefficients $S_d^l = \yv_l \sigma^d(\hv')^T$ are the only ones in $F$ such that this is satisfied. If, in addition, $d \in T + \sord \ZZ$, then $S_d = \yv \theta^d(\hv)^T = \ev \theta^d(\hv)^T$ (by \eqref{eq:sdeh} in the context of $C_{(\theta, \hv, T)}$) and
\begin{equation}\label{eq:sdl}
	S_d^l = \yv_l \sigma^d(\hv')^T = \ev_l \sigma^d(\hv')^T.
\end{equation}

The $\nu$-rank $\nu \times n$ matrix $B$ such that $\ev = \epv B$ as in \eqref{eq:epsilon}, where $\epv$ is a $\nu$-length basis for the $F^\sigma$-span of the entries in $\ev$, is the horizontal concatenation of the $\nu \times \mu$ matrices $B_1, \dots, B_\ell$ such that
\begin{equation}\label{eq:epsilonl}
	\ev_l = \epv B_l \quad \text{ for each } 1 \le l \le \ell.
\end{equation}
Note that the rank of $B_l$ corresponds to the rank weight of the error in the $l$-th block, $\ev_l$, which might be less than $\nu$, the rank weight of $\ev$.
Correspondingly, $\etav^T = B \hv^T$ as defined in \eqref{eq:eta} takes the form
\begin{equation}\label{eq:etalc}
	\etav = \sum_{l=1}^\ell b_l \etav_l
\end{equation}
where
\begin{equation}\label{eq:etal}
	\etav_l^T = B_l (\hv')^T \quad \text{ for each } 1 \le l \le \ell.
\end{equation}
This leads to the following equations, as in \eqref{eq:syndromeeq}:
\begin{equation}\label{eq:epsilonetal}
	S_d^l = \ev_l \sigma^d(\hv')^T = \epv B_l \sigma^d(\hv')^T =  \epv \sigma^d(\etav_l)^T \quad \text{for all } d \in T + \sord \ZZ, 1 \le l \le \ell.
\end{equation}
Again, the rank weight of $\etav_l$ might not reach $\nu$ for any given $l$. Decoding can be done by first finding $\epv$ and then solving for \eqref{eq:epsilonetal}, \eqref{eq:etal} and \eqref{eq:epsilonl} for every $1 \le l \le \ell$.
We shall remark that $\epv$ not being part of a relevant solution for every block (e.g. some block having an error below $\rkw(\epv) = \nu$) is not an issue since, by Proposition \ref{prop:epsilonoreta}, a solution $\epv, \etav_l$ for the $l$-th block such that the elements of $\etav_l$ are within the $F^\sigma$-span of $\hv$ can be found from $\epv$, and then the remaining steps are successful by Remark \ref{remark:decompositionfromsolution}. In order to find $\epv$, we may now take advantage of the fact that its entries are in $F$, as well as from the following result.

\begin{prop}\label{prop:SFSRmultisequenceextension}
	Let $L/F$ be a field extension of degree $\ell$ and $\theta : L \to L$, $\sigma : F \to F$ be field automorphisms such that the restriction of $\theta$ to $F$ equals $\sigma$. Let $\{ b_1, \dots, b_\ell \}$ be any $F$-basis of $L$. Consider any sequence $s_0, \dots, s_{N-1}$ in $L$. Then, for every $0 \le i < N$ there exist unique $\sq i 1, \dots, \sq i \ell \in F$ such that $\sum_{l=1}^\ell \theta^i(b_l) \sq i l = s_i$. Furthermore, for any $\nu \in \ZZ^+$ and any $\vv = (v_0, \dots, v_\nu) \in L^{\nu + 1}$, $\vv$ is a $\theta$-SFSR for the sequence if and only if it is a $\sigma$-SFSR for the $\ell$ sequences $\sq 0 l, \dots, \sq {N-1} l$ for $1 \le l \le \ell$.
\end{prop}

Note that the conditions on $L/F, \sigma, \theta$ are more general than the ones required by Definition \ref{defn:extension}; e.g. $\thord$ might be less than $\ell \sord$. In fact, this result applies to LFSRs by taking the identity automorphism as $\theta$.

\begin{proof}
	Applying powers of a field automorphism $\theta : L \to L$ such that $\theta(F) = F$ to an $F$-basis of $L$ results in another $F$-basis, hence the uniqueness of $\sq i 1, \dots, \sq i \ell$.
	For any $\ell \le j < N$,
	\begin{align*}
		\sum_{i=0}^\nu v_i \theta^i(s_{j-i})
		& = \sum_{i=0}^\nu \sum_{l=1}^\ell v_i \theta^i\left(\theta^{j-i}(b_l) \sq {j-i} l\right) \\
		& = \sum_{l=1}^\ell \theta^j(b_l) \sum_{i=0}^\nu v_i \sigma^i\left(\sq {j-i} l\right).
	\end{align*}
	It follows from Definition \ref{defn:SFSR} that, if $\vv$ is a $\sigma$-SFSR generating the sequences $\sq i l$, then it is a $\theta$-SFSR generating the sequence $s_{i}$. The reciprocal follows from the fact that $\{\theta^j(b_l) ~|~ 1 \le l \le \ell\}$ is an $F$-basis of $L$ for any $j \in \ZZ$.
\end{proof}

The first major step in Algorithm \ref{alg:span} consists of computing the error span vector $\shv$, which leads to the error values $\epv$, by finding the shortest $\theta^{t_1}$-SFSR that generates the $r+1$ sequences $\s i j = S_{b + t_1 i + t_2 k_j}$ for each $0 \le i \le \delta - 2, 0 \le j \le r$, once some parameters $\delta, r, b, t_1, t_2, k_0 < \dots < k_r$ have been chosen.
For each $0 \le j \le r$, we may apply Proposition \ref{prop:SFSRmultisequenceextension} to the $(\delta - 1)$-length sequence $\s 0 j, \dots, \s {\delta - 2} j$, taking $\theta^{t_1}$ and $\sigma^{t_1}$ as the automorphisms and $\{ \sigma^{b + t_2 k_j}(b_1), \dots, \sigma^{b + t_2 k_j}(b_l) \}$ as the $F$-basis of $L$. By using Proposition \ref{prop:SFSRmultisequenceextension} this way, it yields $\ell$ sequences in $F$, which we shall denote $\s i {j, l}$ for each $1 \le l \le \ell$, such that being a $\sigma^{t_1}$-SFSR for all these $\ell$ sequences is equivalent to being a $\theta^{t_1}$-SFSR for the considered sequence. In addition, the coefficients of $\s i {j,l}$ are such that
$\sum_{l=1}^\ell \theta^{t_1 i}(\theta^{b + t_2 k_j}(b_l)) \s i {j, l} = \s i j = S_{b + t_1 i + t_2 k_j}$, which by the uniqueness of \eqref{eq:sdlbasis} means that 
\begin{equation}\label{eq:sl}
	\s i {j,l} = S_{b + i t_1 + k_j t_2}^l = \yv_l \sigma^{b + i t_1 + k_j t_2}(\hv')^T.
\end{equation}

Combining this observation for all $0 \le j \le r$, generating the $r+1$ sequences given by $\s i j$ as a $\theta^{t_1}$-SFSR is equivalent to generating the $(r+1)\ell$ sequences $\s i {j,l}$ as a $\sigma^{t_1}$-SFSR: if there is a single shortest (up to scalar multiplication) $\theta^{t_1}$-SFSR with entries in $L$ generating $\s i j$, then it is also the shortest (up to scalar multiplication) $\sigma^{t_1}$-SFSR with entries in $F$ generating $\s i {j,l}$. The reciprocal follows by Proposition \ref{prop:SFSRsubfield}.

In order to be able to decode with respect to $C_{(\theta, \hv, T)} \cap F^n$, the remaining requirement is that the error span vector $\shv$ should be retrievable as the shortest $\theta^{t_1}$-SFSR for $\s i j$. For this to be the case, the parameters $\delta \ge 2, r, b, t_1, t_2, k_0 < \dots < k_r$ have to be such that the right-hand side of \eqref{eq:T} is a subset of the defining set $T + \sord \ZZ$. Additionally, while \eqref{eq:T} requires $(\thord, t_1) = 1$, for our current purpose we only need $(\sord, t_1) = 1$, which does not guarantee that $(\ell, t_1) = 1$, which is in principle necessary as $\thord = \ell \sord$. However, $\sigma^{t_1} = \sigma^{t_1 + k \sord}$ for all $k \in \ZZ$, and by Dirichlet's theorem on arithmetic progressions (see e.g. \cite{Selberg49}) there exists some $k \in \NN$ such that $t_1 + k \sord$ is a prime number larger than, hence relatively prime to, $\thord$, so being a $\theta^{t_1}$-SFSR for $\s i j$ is equivalent to being a $(\sigma^{t_1 + k \sord})$-SFSR for $\s i {j, l}$ which is equivalent to being a $(\theta^{t_1+ k \sord})$-SFSR for $\s i j$.
Finally, in order to guarantee that up to $\nu$ errors can be corrected for $C_{(\theta, \hv, T)}$, we may enforce Assumption \ref{assumption:Tnusubset}. While, in principle, we should replace $\thord$ for $\sord$ since the considered twist map for $C_{(\theta, \hv, T)}$ is $\thord$, for the same argument as for $t_1$ we may consider without loss of generality that $t_1'$ and $t_2'$ are relatively prime to $\ell$, and therefore $(\sord, t_1') = (\thord, t_1)$, $(\sord, t_2') = (\thord, t_2)$. As a result, Assumption \ref{assumption:Tnusubset} (or any of the assumptions in Section \ref{section:ontheassumption} that lead to it being true) can be considered verbatim without replacing $\thord$ for $\sord$.
Under these conditions, $\shv$ (whose entries are in $F$ since so are the error values) can be retrieved as the shortest $\theta^{t_1}$-SFSR generating $\s i j$, and therefore as the shortest $\sigma^{t_1}$-SFSR generating $\s i {j,l}$, as previously noted.

In order to perform the next steps of decoding, we shall also define
	\begin{equation}\label{eq:stijl}
		\st i {j,l} = \sigma^{-b - t_1 i - t_2 k_j}\left( \s i {j,l} \right),
	\end{equation}
which replaces \eqref{eq:epsilonetal} with the equivalent
	\begin{equation}\label{eq:steel}
		\st i {j,l} = \sigma^{-b - t_1 i - t_2 k_j}\left( \epv \sigma^{b + t_1 i + t_2 k_j}(\etav_l)^T\right) = \etav_l \sigma^{-b - t_1 i - t_2 k_j}(\epv)^T.
	\end{equation}
Once $\epv$ is extracted from $\shv$, for each $l$ one can solve for $\etav_l$ in \eqref{eq:steel} and then for $B_l$ in \eqref{eq:etal}, $\ev_l$ being then computable from \eqref{eq:epsilonl}. This yields $\ev$ by concatenating $\ev_l$, and finally $\cv = \yv - \ev$. If no codeword is at a distance satisfying Assumption \ref{assumption:Tnusubset}, the same observations as in Section \ref{section:toomanyerrors} apply, so the handling of possible decoding failures is identical.

This is summarized in Algorithm \ref{alg:interleavedspan}, which works for codes that follow Definition \ref{defn:ccode} interleaved with themselves.

\begin{algorithm}
	\caption{Decoding algorithm for interleaved codes}\label{alg:interleavedspan}
	\begin{algorithmic}[1]
		\Input A word $\yv \in F^n$ and parameters $\hv \in F^\mu, \sigma : F \to F, b, \delta \ge 2, t_1, t_2, \allowbreak {k_0, \dots, k_r}$ where $n = \ell \mu$ and $(\ell, t_1) = 1 = (\sord, t_1)$.
		\Require $b + t_1 \{ 0, \dots, \delta - 2 \} + t_2 \{ k_0, \dots, k_r \} \subset T + \sord \ZZ$.
		\Require $\sigma$ can be extended into $\theta : L \to L$ for some field extension $L/F$ of degree $\ell$.
		\Output $\cv \in F^n$ or `decoding failure'.
		\Ensure If `decoding failure' is returned, there is no $\cv \in C_{(\sigma, \hv, T)}^{\times \ell}$ such that $\rkd(\cv, \yv) = \nu$ satisfies Assumption \ref{assumption:Tnusubset}. Otherwise, $\cv$ is the nearest codeword in $C_{(\sigma, \hv, T)}^{\times \ell}$ to $\yv$.
		\Statex % blank line

		\State $\s i {j,l} \leftarrow \yv_l \sigma^{b + t_1 i + t_2 k_j}\left( \hv \right)^T$ for $0 \le i \le \delta - 2, 0 \le j \le r, 1 \le l \le \ell$ \Comment{see \eqref{eq:sl}}
		\If{$\s i {j,l} = 0$ for all computed $i, j, l$}
			\State \Return $\yv$
		\EndIf
		\State $\shv = (\sh_0, \dots, \sh_{\bar \nu}) \leftarrow$ the shortest $\sigma^{t_1}$-SFSR generating $\s i {j,l}$
		\State Find $\epv = (\ep_1, \dots, \ep_{\bar \nu})$ from $\shv$ \Comment{see Remark \ref{remark:shld2epeta}}
		\If{$\epv$ cannot be computed from $\shv$}
			\State \Return `decoding failure'
		\EndIf
		\ForAll{$1 \le l \le \ell$}
			\State $\st i {j,l} \leftarrow \sigma^{-b - t_1 i - t_2 k_j}\left(\s i {j,l}\right)$ for $0 \le i \le \delta - 2, 0 \le j \le r$ \Comment{see \eqref{eq:stijl}}
			\State Find $\etav_l$ from $\epv$ and $\st i {j,l}$ by solving \eqref{eq:steel}
			\If{$\etav_l$ cannot be computed} \Comment{possible if $r > 0$}
				\State \Return `decoding failure'
			\EndIf
			\State Find $B_l$ from $\etav_l$ and $\hv$ by solving $\etav_l^T = B_l \hv^T$ \Comment{see \eqref{eq:etal}}
			\If{$B_l$ cannot be computed} \Comment{possible if $\mu < \sord$}
				\State \Return `decoding failure'
			\EndIf
			\State $\ev_l \leftarrow \epv_l B_l$ \Comment{see \eqref{eq:epsilonl}}
		\EndFor
		\State $\ev \leftarrow (\ev_1 | \dots | \ev_\ell)$
		\State $\cv \leftarrow \yv - \ev$
		\If{$\cv \notin C_{(\sigma, \hv, T)}^{\times \ell}$}
			\State \Return `decoding failure'
		\EndIf
		\State \Return $\cv$
	\end{algorithmic}
\end{algorithm}

\begin{remark}\label{remark:alginterleavedspangabidulin}
	For $b = r = 0$ and $\sigma^{t_1}$ the Frobenius endomorphism of $F$ a finite field (that is, when considering a Gabidulin code, or a subcode thereof, interleaved with itself), Algorithm \ref{alg:interleavedspan} is equivalent to \cite[Algorithm 4]{SJB11}. The observations given in Remark \ref{remark:algspangabidulin} for the actual identity of the algorithms analogously apply for these two algorithms.
	Note that, while Algorithm \ref{alg:interleavedspan} does now have decoding failure handling, it does not completely match the one in \cite[Algorithm 4]{SJB11}. There, a decoding failure is raised if the SFSR synthesis solution is not unique (which is written as the equivalent check $\eta \ne 0$, see \cite[Corollary 7]{SJB11}. This event is likely (but not guaranteed in general) to lead to some of the remaining decoding failures described in Section \ref{section:toomanyerrors} and checked in Algorithm \ref{alg:interleavedspan}.
\end{remark}

\begin{remark}\label{remark:lessdecodingfailuresininterleaved}
	For $\ell = 1$, Algorithm \ref{alg:interleavedspan} is Algorithm \ref{alg:span} with decoding failure handling, as described in Section \ref{section:toomanyerrors}, instead of Assumption \ref{assumption:Tnusubset} as a precondition. For example, in every instance of decoding in Section \ref{section:example} following Algorithm \ref{alg:span}, the steps match Algorithm \ref{alg:interleavedspan} for $\ell = 1$ if `decoding failure' is returned when finding a decoding failure. The following observation means that it is particularly reasonable to use decoding failure handling in Algorithm \ref{alg:interleavedspan}, since for $\ell > 1$ it is not highly improbable to decode beyond the limit given by the assumption.
	The rank of the matrix $H_\nu$ as given in \eqref{eq:HE} is equal to the rank of the $\ell(\delta - 1 - \nu)(r + 1) \times \nu$ matrix $\tilde H_\nu$, defined as the result from replacing each row in $H_\nu$, which is of the form $\theta^\iota(\etav)$ for some $\iota$, with $\ell$ rows of the form $\sigma^\iota(\etav_l)$ for each $1 \le l \le \ell$ (recall \eqref{eq:etalc}), so $\tilde H_\nu$ is the result of vertically concatenating the matrices $\sigma^\iota(M)$ for each $\iota \in T_\nu = b + t_1 \nu + t_1 \{ \nu, \dots, \delta - 2 \} + t_2 \{ k_0, \dots, k_r \}$, where $M$ is the vertical concatenation of $\etav_l$ for each $1 \le l \le \ell$.
	By Lemma \ref{lemma:nurank} and Proposition \ref{prop:SFSRmultisequenceextension}, if $\nu \le \delta - 2$ and the rank of $\tilde H_\nu$ is $\nu$, $\shv$ will be, up to nonzero scalar multiplication, the shortest $\sigma^{t_1}$-SFSR for the $(r+1)\ell$ sequences.
	If $F$ is a finite field of order $|F|$, following the ideas in \cite[Section V]{SB10}, which apply the results in \cite[Section 3.2]{Overbeck07} (mainly \cite[Lemma 3.12]{Overbeck07}), and assuming that $\etav$ is chosen uniformly among the elements in $L^\nu$ of rank weight $\nu$, if there is some subset of the form (or equal modulo $\sord$ to one of the form) $b' + s\{ 0, \dots, d \}$ in $T_\nu$ for some $d, s$ such that $(s, \sord) = 1$ and $\nu \le \ell(d + 1)$, then the probability of the rank of $\tilde H_\nu$ being less than $\nu$, and therefore the probability of failing to correct a given error of rank weight $\nu$, can be upper bounded by $4/|F|$ as a consequence of the structure in $\tilde H_\nu$.
	The most immediate way to get such a subset of $T_\nu$ is $b' = b + t_1 \nu, s = t_1, d = \delta - 2 - \nu$, so the success probability is close to $1$, unless $F$ is too small, when $\nu \le \ell(\delta - 1 - \nu)$, that is, when $\nu(\ell + 1) \le \ell(\delta - 1)$, or equivalently $\nu \le \frac{\ell}{\ell + 1}(\delta - 1)$. This bound for $\nu$ might be above the greatest $\nu$ allowed by Assumption \ref{assumption:Tnusubset} (in fact, it might reach almost $2 \tau$ if $\ell \gg 1$ and $0 \approx r \ll \delta$). Since \cite{SB10} and \cite{SJB11} deal with Gabidulin codes, this is done there for $\sigma$ the Frobenius endomorphism, $s = 1$ and $b' = \nu$ (compare the bound above for $\nu$ with \cite[Eq. (43)]{SJB11}). Generalizing this to a general finite field automorphism and a general $s$ should be straightforward, while it is possible that the results of \cite[Section 3.2]{Overbeck07} could be refined to take further advantage of the structure of the defining set when $r > 0$.
\end{remark}

\begin{example}\label{example:interleavederrorcorrection}
	Recall the code $\mathcal C = C_{(\sigma, \hv, T)}$ from Section \ref{section:example}, where $\sigma$ is the Frobenius endomorphism in $F = \FF_{2^{14}} = \FF_2(a)$ for some $a \in F$ such that $a^{14} + a^7 + a^5 + a^3 + 1 = 0$, $\hv$ is $\left( a^7, \sigma(a^7), \dots, \sigma^{13}(a^7) \right)$ and $T$ is $\{ 0, 1, 2, 3, 4,\allowbreak 8, 9, 10, 11, 12 \}$.
	In Section \ref{section:example}, an error of rank weight $3$ with respect to $\mathcal C$ was corrected.
	For any $\ell > 1$, $\lfloor\frac{\ell}{\ell + 1}(\delta - 1)\rfloor = \lfloor 5\frac{\ell}{\ell + 1}\rfloor$ is $4$ for $\ell \ge 4$ and $3$ for $\ell \in \{ 2, 3 \}$, so by the observations in Remark \ref{remark:lessdecodingfailuresininterleaved}, by running Algorithm \ref{alg:interleavedspan} with the parameters that successfully corrected errors of rank weight $3$ in Section \ref{section:example} ($b = 8$, $t_1 = 1$, $t_2 = 3$, $\delta = 6$, $r = 1$, $k_0 = 0$ and $k_1 = 2$), errors of rank weight $4$ with respect to $\mathcal C^{\times \ell}$ should be corrected with high probability for $\ell \ge 4$.
	In fact, since \cite[Algorithm 4]{SJB11} should also be able to correct errors of rank weight $4$ (see \cite[Eq. (43)]{SJB11}) for $\ell \ge 4$, by Remark \ref{remark:alginterleavedspangabidulin} there should be a high probability of correcting such errors for the parameters $b = 0$ (or $b = 8$), $t_1 = 1 = t_2$, $\delta = 6$, $r = 0$ and $k_0 = 0$.
	In order to illustrate that Algorithm \ref{alg:interleavedspan} can correct beyond these cases, we shall use it in order to correct an error of rank weight $\nu = 4$ with respect to $\mathcal C^{\times 2}$, that is, for the interleaved code from Example \ref{example:interleavedasasubfieldsubcode}, which, as described in the example, can also be seen as a subfield subcode from a code according to Definition \ref{defn:ccode}.

	Consider the error values $\epv = \left( 1, a, a^{11}, a^5 \right)$, whose error span vector has as its entries the coefficients of
	\[
		\begin{aligned}
			\sh & = \lclm{z - 1, z - a{-1} \sigma(a), z - a^{-11} \sigma(a^{11}), z - a^{-5} \sigma(a^{5})} \\
			& = z^4 + a^{14980}z^3 + a^{7019}z^2 + a^{5972}z + a^{13080} \in F[z;\sigma]
		\end{aligned}
	\]
	and is therefore
	\begin{equation}\label{eq:exampleshvl}
		\shv = \left( a^{13080}, a^{5972}, a^{7019}, a^{14980}, 1 \right).
	\end{equation}
	Let $B$ be the horizontal concatenation of the matrices $B_1$ and $B_2$, defined as
\begin{equation}\label{eq:exampleB12}
	\begin{aligned}
	B_1 & =
	\begin{pmatrix}
		1 & 1 & 1 & 1 & 1 & 1 & 1 & 1 & 1 & 1 & 1 & 1 & 1 & 1 \\
		0 & 0 & 1 & 1 & 1 & 1 & 1 & 1 & 1 & 1 & 1 & 1 & 1 & 1 \\
		0 & 1 & 1 & 1 & 1 & 1 & 0 & 0 & 0 & 0 & 0 & 0 & 0 & 0 \\
		1 & 0 & 0 & 0 & 0 & 0 & 0 & 0 & 0 & 0 & 0 & 0 & 0 & 0
	\end{pmatrix}
	, \\
	B_2 & =
	\begin{pmatrix}
		1 & 1 & 1 & 1 & 1 & 1 & 1 & 1 & 1 & 1 & 1 & 1 & 1 & 1 \\
		1 & 1 & 1 & 1 & 1 & 1 & 1 & 1 & 1 & 1 & 1 & 1 & 0 & 0 \\
		0 & 0 & 0 & 0 & 0 & 0 & 0 & 0 & 1 & 1 & 1 & 1 & 1 & 0 \\
		0 & 0 & 0 & 0 & 0 & 0 & 1 & 1 & 1 & 1 & 1 & 1 & 0 & 0
	\end{pmatrix},
	\end{aligned}
\end{equation}
	The first three rows of $B_1$ are the rows of the matrix $B$ in \eqref{eq:exampleepvB}, while the first three rows of $B_2$ are the result of horizontally flipping the three rows of $B$. Hence, $\epv B_1$ is the result of adding an error of rank weight $1$ to the error of rank weight $3$ considered in Section \ref{section:example}, while $\epv B_2$ is the result of adding an error of rank weight $1$ to the result of flipping the entries of that error vector.
	Then, $\ev = \epv B = (\epv B_1 | \epv B_2)$ is a vector in $F^{28}$ of rank weight $4$ whose error values are the entries of $\epv$ (or any other four elements in $F$ which are a basis of the same $F^\sigma$-subspace of $F$) and whose error locators (recall \eqref{eq:etal}) are
	\begin{equation}\label{eq:exampleetavl}
		\begin{aligned}
			\etav_1 & = B_1 \hv^T = \left(1, a^{9414}, a^{12430}, a^{7} \right)
			,
			\\
			\etav_2 & = B_2 \hv^T = \left(1, a^{10545}, a^{1889}, a^{5408} \right).
		\end{aligned}
	\end{equation}
	
	If $\cv_1$ is the codeword in $\mathcal C$ described in \eqref{eq:examplec} and $\cv_2 = (0, \dots, 0) \in F^{14}$, then $\cv = (\cv_1 | \cv_2)$ is a codeword in $\mathcal C^{\times 2}$.
	The first step when running Algorithm \ref{alg:interleavedspan} with $\yv = \cv + \ev$ as the input and the parameters $\hv$, $\sigma$, $b = 8$, $\delta = 6$, $t_1 = 1$, $t_2 = 3$, $r = 1$, $k_0 = 0$ and $k_1 = 2$ is the computation of the syndromes $\s i {j, l}$ for $l = 1$ and $l = 2$, and for each $0 \le i \le \delta - 2 = 4$ and each $0 \le j \le r = 1$. These syndromes are
	\[
	\begin{aligned}
		\left( \s i {j,1} \right)_{\substack{0 \le j \le 1 \\ 0 \le i \le 4}} & =
		\begin{pmatrix}
			a^{5903} & a^{12024} & a^{954} & a^{885} & a^{15923} \\
			a^{4230} & a^{12338} & a^{12893} & a^{12777} & a^{15865}
		\end{pmatrix}
		,
		\\
		\left( \s i {j,2} \right)_{\substack{0 \le j \le 1 \\ 0 \le i \le 4}} & =
		\begin{pmatrix}
			a^{3538} & a^{14609} & a^{14246} & a^{13194} & a^{8858} \\
			a^{6680} & a^{12714} & a^{13177} & a^{14706} & a^{1497}
		\end{pmatrix}.
	\end{aligned}
	\]
	Up to nonzero $F$-multiples, the shortest $\sigma$-SFSR which generates the four sequences given by the four rows is $\shv$ as described in \eqref{eq:exampleshvl}, so a nonzero $F$-multiple of $\shv$ will be retrieved by solving the SFSR synthesis problem.
	This uniqueness is lost if any of the four sequences is omitted.
	For example, if only the two sequences for $l = 1$ or for $l = 2$ are used as the input of Algorithm \cite[Algorithm 2]{SJB11} (and the warning of absence of uniqueness for the shortest SFSR is ignored), the result is not $\shv$ but, respectively, $a^{-14526}\left(a^{14526}, a^{1014}, a^{339}, a^{14214}, 1\right)$ or $a^{-8159}\left(a^{8159}, a^{462}, a^{4442}, a^{3516}, 1\right)$. The dimensions of the corresponding kernels are $0$ and $1$. Note that these would be the sequences of syndromes that would be computed when decoding, with either Algorithm \ref{alg:span} or Algorithm \ref{alg:interleavedspan} for $\ell = 1$, the first block of $\yv$, $\yv_1 = \cv_1 + \epv B_1$ or, respectively, the second, $\yv_2 = \cv_2 + \epv B_2$, with respect to $\mathcal C$. Hence, the kernels would not reach dimension $4$, leading to a decoding failure. This would also happen if the error vector $\ev$ is chosen as $(\ev_1 | \ev_1)$ or, respectively, $(\ev_2 | \ev_2)$, since this would replace the sequences for $l = 2$ with the ones for $l = 1$ or, respectively, the sequences for $l = 1$ with the ones for $l = 2$.
	As another example, if only the two sequences for $j = 0$ or for $j = 1$ are considered, Algorithm 2 in \cite{SJB11} gives two vectors of length $5$ whose corresponding kernels have dimensions $2$ and, respectively, $1$. Therefore, if Algorithm \ref{alg:interleavedspan} is run with the same parameters except $r = 0$, or with the same parameters except $r = b = 0$, the result would also be a decoding failure.
	By Remark \ref{remark:alginterleavedspangabidulin}, these would respectively correspond to decoding with respect to the interleaved Gabidulin codes $C_{(\sigma, \hv, \{ 8, \dots, 12 \})}^{\times 2}$ and $C_{(\sigma, \hv, \{ 0, \dots, 4 \})}^{\times 2}$ using \cite[Algorithm 4]{SJB11}.
	In summary, this example cannot be decoded by independently decoding each one of the two blocks or by using \cite[Algorithm 4]{SJB11}.

	After getting $\shv$, the error values $\epv$ are obtained as an $F^\sigma$-basis of the kernel of $\sh(\sigma)$, such as $\epv = \left( 1, a, a^{11}, a^5 \right)$. Then, the remaining steps consist of computing, for each $l$-th block, the error locators $\etav_l$, the matrix $B_l$ and the error vector $\ev_l$. As for Algorithm \ref{alg:span} and Algorithm \ref{alg:locator}, the error locators can be computed through the version of Gabidulin's algorithm described in Section \ref{section:gabidulinsalgorithm} instead of solving Equation \eqref{eq:steel}, which itself requires computing $\st i {j,l}$ as given in Equation \eqref{eq:stijl}. When running Gabidulin's algorithm with $(a_1, \dots, a_4) = \epv$, $(X_1, \dots, X_4) = \etav_1$, $(b_0, \dots, b_3) = (\s 0 {0,1}, \dots, \s 0 {3,1})$, $\theta = \sigma$ and $\bar b = 8$, the values for $A_k^{(j)}$ and $B_i^{(j)}$ are
\[
	\begin{aligned}
	\left( A_k^{(j)} \right)_{\substack{1 \le j \le 4 \\ j \le k \le 4}} & =
	\begin{pmatrix}
		1 & a & a^{11} & a^{5} \\
		  & a^{6449} & a^{675} & a^{9444} \\
		  &  & a^{4433} & a^{12255} \\
		  &  &  & a^{1902}
	\end{pmatrix}
	, \\
	\left( B_i^{(j)} \right)_{\substack{1 \le j \le 4 \\ 0 \le i \le 4 - j}} & =
	\begin{pmatrix}
		a^{5903} & a^{12024} & a^{954} & a^{885} \\
		a^{8371} & a^{3029} & a^{12427} &  \\
		a^{15663} & a^{944} &  &  \\
		a^{3694} &  &  &
	\end{pmatrix},
	\end{aligned}
\]
and this gives $\etav_1$ as in \eqref{eq:exampleetavl}, since the same $\epv$ that led to that value for $\etav_1$ has been used. This, in turn, gives $B_1$ as in \eqref{eq:exampleB12}, and then $\ev_1 = \epv_1 B_1$.
	For $l = 2$, the input is the same except $(X_1, \dots, X_4) = \etav_2$ and $(b_0, \dots, b_3) = (\s 0 {0,2}, \dots, \s 0 {3,2})$. This leads to the same values for $A_k^{(j)}$, as they only depend on $(a_1, \dots, a_4)$, so this computation only has to be done once, and in this case
\[
	\left( B_i^{(j)} \right)_{\substack{1 \le j \le 4 \\ 0 \le i \le 4 - j}} =
	\begin{pmatrix}
		a^{3538} & a^{14609} & a^{14246} & a^{13194} \\
		a^{8780} & a^{11113} & a^{4967} &  \\
		a^{5294} & a^{6995} &  &  \\
		a^{10178} &  &  & 
	\end{pmatrix},
\]
	which gives $\etav_2$, which gives $B_2$ as in \eqref{eq:exampleB12}, which gives $\ev_2$ as $\epv_2 B_2$. In this case, $\ev_2$ matches $\yv_2$, the second block of $\yv$, so $\yv_2 - \ev_2 = \cv_2$, which was defined as the zero vector in $F^{14}$. For the first block, the result is $\yv_1 - \ev_1 = \cv_1$, so Algorithm \ref{alg:interleavedspan} returns the intended codeword of $\mathcal C^{\times 2}$, $(\cv_1 | \cv_2) = \yv - (\ev_1 | \ev_2)$, successfully correcting an error of rank weight $4$.
\end{example}

\section*{Acknowledgements}

This work is part of a PhD thesis under the supervision of F. J. Lobillo, whom the author wants to thank kindly. The author also wants to thank A. Neri kindly for notably productive discussions which, as a side effect, led to a significant simplification of a few of the core aspects in the decoding procedure and the results thereof.

\end{document}